\documentclass[reqno,12pt,a4paper]{article}

\usepackage{amsfonts,amssymb,amsthm, 
amsmath,amscd, bbm, bm, mathabx, 
mathrsfs,delarray,subfigure}
\usepackage[dvipsnames]{xcolor}
\usepackage{hyperref}
\usepackage{cite}
\usepackage{xypic}
\usepackage{sectsty}

\usepackage{accents}

\sectionfont{\fontsize{12}{15}\selectfont}
\subsectionfont{\fontsize{12}{15}\selectfont}

\usepackage{enumitem}

\makeatletter 
\newcommand\raggedtop{%
  \topskip=1\topskip plus 10pt} 

\makeatother 
\raggedtop

\DeclareMathOperator{\Ad}{Ad}

\DeclareMathOperator{\ran}{ran} 
\DeclareMathOperator{\id}{id}

\DeclareMathOperator{\Fun}{Fun}
\DeclareMathOperator{\iFun}{\textsc{Fun}} 

\DeclareMathOperator{\Map}{Map}
\DeclareMathOperator{\iMap}{\textsc{Map}}        

\DeclareMathOperator{\Der}{Der}

\DeclareMathOperator{\ee}{e}

\DeclareMathOperator{\ZZ}{Z\hspace{-1pt}}
\DeclareMathOperator{\DD}{D\hspace{-1pt}}

\DeclareMathOperator{\OOO}{Op\hspace{-1pt}}
\DeclareMathOperator{\iOOO}{\textsc{Op}\hspace{-1pt}}     

\numberwithin{equation}{subsection} 
\numberwithin{subsection}{section} 

\newcommand{\ceqref}[1]{{\textcolor{blue}{\eqref{#1}}}}
\newcommand{\cref}[1]{{\textcolor{blue}{\ref{#1}}}}
\newcommand{\ccite}[1]{{\textcolor{blue}{\!\cite{#1}}}}

\newcommand{\ddd}{{\hbox{$\bigoplus$}}}

\newcommand{\ul}[1]{{\underline{#1}}}

\newcommand{\mathsans}[1]{{{\sf #1}}}

\font\euler=eusm10 at 12.8 truept
\font\scripteuler=eusm7
\font\scriptscripteuler=eusm5 
\def\eul{\fam=12}
\newcommand{\matheul}[1]{{{\eul #1}}}

\newtheorem{defi}{{\sf Definition}}[section]
\newtheorem{prop}{{\sf Proposition}}[section]

\newtheorem{lemma}{{\sf Lemma}}[section]

\newtheorem{rem}{{\sf Remark}}[section]

\begin{document}

\vskip1.5cm
\begin{large}
{\flushleft\textcolor{blue}{\sffamily\bfseries Operational total space theory of principal 2--bundles II:}}
{\flushleft\textcolor{blue}{\sffamily\bfseries 2--connections and 1-- and 2--gauge transformations}}  
\end{large}
\vskip1.3cm
\hrule height 1.5pt
\vskip1.3cm
{\flushleft{\sffamily \bfseries Roberto Zucchini}\\
\it Dipartimento di Fisica ed Astronomia,\\
Universit\`a di Bologna,\\
I.N.F.N., sezione di Bologna,\\
viale Berti Pichat, 6/2\\
Bologna, Italy\\
Email: \textcolor{blue}{\tt \href{mailto:roberto.zucchini@unibo.it}{roberto.zucchini@unibo.it}}, 
\textcolor{blue}{\tt \href{mailto:zucchinir@bo.infn.it}{zucchinir@bo.infn.it}}}


\vskip.7cm
\vskip.6cm 
{\flushleft\sc
Abstract:} 
The geometry of the total space of a principal bundle
with regard to the action of the bundle's structure group is elegantly described by the 
bundle's operation, 
a collection of derivations consisting of the de Rham differential and the contraction
and Lie derivatives of all vertical vector fields and satisfying the six Cartan relations. 
Connections and gauge transformations are defined by the way they behave under the 
action of the operation's derivations. In the first paper of a series of two
extending the ordinary theory, we constructed an operational total space theory of strict principal 2--bundles 
with reference to the action of the structure strict 2--group. Expressing this latter through a crossed module
$(\mathsans{E},\mathsans{G})$, the operation is based on the derived Lie group
$\mathfrak{e}[1]\rtimes\mathsans{G}$. In this paper, the second of the series, 
an original formulation of the theory of $2$--connections and $1$-- and $2$--gauge transformations
of principal 2--bundles based on the operational framework is  provided. 
\vspace{2mm}

\vfil\eject
\tableofcontents

\vfil\eject

\section{\textcolor{blue}{\sffamily Introduction}}\label{sec:introII}

Principal $2$--bundle theory is a topic of higher geometry 
important, among other reasons, for its relevance in higher gauge theory (see e. g. \ccite{Baez:2010ya}
for a review).   
Various approaches to this subject have been developed so far 
constituting a large body of literature  
\ccite{Schreiber:2005ff,Baez:2004in,Baez:2005qu,Giraud:1971cna,Bartels:2006hgtb,
Wockel:2008tspb,Schommer:2011ces,Aschieri:2003mw,Laurent:2009nag,Ginot:2008gpc,Nikolaus:2013fvng,
Breen:2001ie,Schreiber:2013pra,Jurco:2014mva,Jurco:2016qwv,Waldorf:2016tsct,Waldorf:2017ptpb}.

This is the second of a series of two papers aimed at the construction of an operational 
total space theory of principal $2$--bundles. In a companion paper, henceforth referred to as I \ccite{Zucchini:2019rpf}, 
we laid the foundations of the operational total space framework \ccite{Greub:1973ccc}. In this paper, 
referred to as II, based on the operational setup worked out in I,
we provide an original formulation of the theory of $2$--connections and $1$-- 
and $2$--gauge transformations. 


\subsection{\textcolor{blue}{\sffamily Operational theory of principal 2--bundles}}\label{subsec:approachI}

Before proceeding to illustrating the plan of II, we review briefly the content of I
to privide the reader with a general overview of the matter. 

A principal $2$--bundle consists of a morphism manifold $\hat P$ with an object submanifold $\hat P_0$ 
forming a groupoid, a base manifold $M$, compatible projection maps 
$\hat\pi:\hat P\rightarrow M$ and $\hat\pi_0:\hat P_0\rightarrow M$
describing a functor, a morphism group $\hat{\mathsans{K}}$ with an object subgroup $\hat{\mathsans{K}}_0$ organized 
as a strict Lie $2$-group and compatible right actions $\hat R:\hat P\times\hat{\mathsans{K}}\rightarrow\hat P$ and 
$\hat R_0:\hat P_0\times\hat{\mathsans{K}}_0\rightarrow\hat P_0$ constituting a functor and respecting $\hat\pi$ and $\hat\pi_0$.
The $2$--bundle is also locally trivializable, that is on any sufficiently small neighborhood $U$ 
of $M$ the groupoid $(\hat P|_U,\hat P_0|_U)$ is equivariantly projection preservingly 
equivalent 
to the groupoid $(U\times\hat{\mathsans{K}},U\times\hat{\mathsans{K}}_0)$ with the obvious projection
and right action structures. 

In I, we showed that there exists a synthetic structure adjoined to a principal $2$--bundle as above
consisting of morphism and object manifolds $P$ and $P_0$, the base manifold $M$, 
projections $\pi$ and $\pi_0$, morphism and object groups $\mathsans{K}$ and $\mathsans{K}_0$ 
and right $\mathsans{K}$-- and $\mathsans{K}_0$-- actions $R$ and $R_0$ on $P$ and $P_0$. 
The synthetic setup is formally obtained from the original non synthetic one 
as follows. Describe the strict Lie $2$--group $(\hat{\mathsans{K}},\hat{\mathsans{K}}_0)$ by its associated 
Lie group crossed module $(\mathsans{E},\mathsans{G})$ so that 
$\hat{\mathsans{K}}=\mathsans{E}\rtimes\mathsans{G}$ and $\hat{\mathsans{K}}_0=\mathsans{G}$.
Then, $\mathsans{K}=\mathfrak{e}[1]\rtimes\mathsans{G}$ and $\mathsans{K}_0=\mathsans{G}$. 
Formally extend further the $\hat{\mathsans{K}}$--action $\hat R$ to a $\mathsans{K}$--action. Then, $P$ is
the $\mathsans{K}$--action image of $\hat P_0$ and $P_0=\hat P_0$, $R$ is the restriction of $\hat R$ to $P$
and $R_0=\hat R$. Above, $\mathsans{K}$ and $P$ must be thought of as certain spaces of functions from $\mathbb{R}[-1]$ 
to $\mathsans{E}\rtimes\mathsans{G}$ and $P$, respectively, in the spirit of synthetic smooth geometry. 
Although the synthetic structure shares many of the properties of the underlying principal $2$--bundle,
it is not one because neither pairs $(\mathsans{K},\mathsans{K}_0)$ and $(P,P_0)$ 
have a groupoid structure. 

With any Lie group crossed module such as $(\mathsans{E},\mathsans{G})$, there are associated
the derived Lie group $\DD\mathsans{M}=\mathfrak{e}[1]\rtimes\mathsans{G}$ and its 
subgroup $\DD\mathsans{M}_0=\mathsans{G}$ whose rich proper\-ties were exhaustively studied in I.
When expressing $\hat{\mathsans{K}}$, $\hat{\mathsans{K}}_0$ in terms of the crossed module encoding their 
underlying $2$--group, one has $\mathsans{K}=\DD\mathsans{M}$ and $\mathsans{K}_0=\DD\mathsans{M}_0$.
The $\mathsans{K}$-- and $\mathsans{K}_0$--actions on $P$ and $P_0$
can in this way be described in terms of $\DD\mathsans{M}$ and $\DD\mathsans{M}_0$, respectively. 

As explained at length in I, the right $\DD\mathsans{M}$--action on $P$ is codified in an operation
$\iOOO S_P$. This is the geometrical structure consisting of the graded algebra 
$\iFun(T[1]P)$ of internal functions of $T[1]P$ and the collection of graded derivations of
$\iFun(T[1]P)$ comprising the de Rham vector field $d_P$ and the contraction and Lie vector fields
$j_{PZ}$, $l_{PZ}$, $Z\in\DD\mathfrak{m}$,  describing the action infinitesimally, where $\DD\mathfrak{m}$ 
is the Lie algebra of $\DD\mathsans{M}$.   
The derivations obey the six Cartan relations,
\begin{align}
&[d_P,d_P]=0,
\vphantom{\Big]}
\label{rintro1}
\\
&[d_P,j_{PZ}]=l_{PZ},
\vphantom{\Big]}
\label{rintro2}
\\
&[d_P,l_{PZ}]=0,
\vphantom{\Big]}
\label{rintro3}
\\
&[j_{PZ},j_{PW}]=0,
\vphantom{\Big]}
\label{rintro4}
\\
&[l_{PZ},j_{PW}]=j_{[Z,W]},
\vphantom{\Big]}
\label{rintro5}
\\
&[l_{PZ},l_{PW}]=l_{[Z,W]}.
\vphantom{\Big]}
\label{rintro6}
\end{align} 
It is possible to similarly construct an operation $\iOOO S_{P0}$ codifying the right $\DD\mathsans{M}_0$--action on $P_0$
consisting of the internal function algebra $\iFun(T[1]P_0)$ acted upon by the de Rham vector field
$d_{P_0}$ and the contraction and Lie vector fields $j_{P_0Z_0}$, $l_{P_0Z_0}$, $Z_0\in\DD\mathfrak{m}_0$.


\subsection{\textcolor{blue}{\sffamily 2--connections and 1-- and 2--gauge transformations}}\label{subsec:approachII}

The operational framework of I just reviewed is the geometric setup on which the theory of $2$--connections and $1$-- 
and $2$--gauge transformations presented in this paper rests. 

In sect. \cref{sec:ordgau}, we review the ordinary total space theory of principal bundle connections and gauge 
transformations concentrating on the two aspects of it which are most relevant for us, the operational 
description (cf. subsect. \cref{subsec:ordgau}) and and the basic theory (cf. subsect. \cref{subsec:ordgau1}). 
This will furnish a prototypical model inspiring the construction of 
the corresponding higher theory. 

In sect. \cref{sec:2conn1gau}, synthetic definitions of 
$2$--connections and $1$-- and $2$--gauge transformations are given in the operational 
framework (cf. subsect. \cref{subsec:setup}). 
A $2$--connection $A$ is a degree $1$ $\DD\mathfrak{m}$--valued internal function on $T[1]P$ 
behaving in a prescribed way under the action of the vector fields $d_P$, $j_{PZ}$, $l_{PZ}$ of 
the operation $\iOOO S_{P}$ (cf. subsect. \cref{subsec:cmconn}). 
The grading of $\DD\mathfrak{m}$ ensures that $A$ has a degree $1$ $\mathfrak{g}$--valued component $\omega$ 
and a degree $2$ $\mathfrak{e}$--valued component $\varOmega$ which directly correspond to and have properties 
closely related to those of the familiar components of a $2$--connection in strict higher gauge theory. Similarly, 
a $1$--gauge transformation $\varPsi$ is a degree $0$ $\DD\mathsans{M}$--valued internal function on $T[1]P$ acted 
upon in a certain way by $d_P$, $j_{PZ}$, $l_{PZ}$, which by the grading of $\DD\mathsans{M}$ has 
a degree $0$ $\mathsans{G}$--valued component $g$ and a degree $1$ $\mathfrak{e}$--valued component $J$ 
directly corresponding to and with properties closely related to those of the components 
of a $1$--gauge transformation in strict higher gauge theory (cf. subsect. \cref{subsec:cmgautr}). 
The action of a $1$--gauge transformation
$\varPsi$ on a $2$--connection $A$ can be defined and has the expected properties. $2$--gauge transformations and 
their action on $1$--gauge transformations can be incorporated into this operational framework as well (cf. 
subsect. \cref{subsec:cmgauforgau}). 

A $2$--connection $A$ can be pulled back from $T[1]P$ to $T[1]P_0$ using the inclusion map 
$I:P_0\rightarrow P$. The pull--back $I^*A$ behaves under the action of the vector fields 
$d_{P_0}$, $j_{P_0Z_0}$, $l_{P_0Z_0}$ of the operation $\iOOO S_{P0}$ in a way determined by the behaviour 
of $A$ under the vector fields $d_P$, $j_{PZ}$, $l_{PZ}$ of $\iOOO S_{P}$. It is possible to consistently 
impose the condition that the degree $2$ component $I^*\varOmega$ of $I^*A$ vanishes. Upon doing so, 
the degree $1$ component $I^*\omega$ of $I^*A$ formally functions in $\iOOO S_{P0}$ as a connection of
an ordinary principal bundle $P_0$ with structure group $\DD\mathsans{M}_0$, though $P_0$ is not one in general. 
Similarly, a $1$--gauge transformation $\varPsi$ can be pulled back from $T[1]P$ to $T[1]P_0$ using $I$. 
The pull--back $I^*\varPsi$ behaves under the vector fields 
$d_{P_0}$, $j_{P_0Z_0}$, $l_{P_0Z_0}$ of $\iOOO S_{P0}$ in a way determined by the behaviour 
of $\varPsi$ under the vector fields $d_P$, $j_{PZ}$, $l_{PZ}$ of $\iOOO S_{P}$. It is possible to consistently 
impose the condition that the degree $1$ component $I^*J$ of $I^*\varPsi$ vanishes. 
The degree $0$ component $I^*g$ of $I^*\varPsi$ then formally functions in $\iOOO S_{P0}$ as if it were 
a gauge transformation of $P_0$ as a would--be ordinary principal bundle. 

The internal functions of $T[1]P$ annihilated by all vector fields $j_{PZ}$, $l_{PZ}$ with $Z\in\DD\mathfrak{m}$
constitute the basic subalgebra $\iFun_{\mathrm{b}}(T[1]P)$ of $\iFun(T[1]P)$. Unlike for ordinary principal bundles,
$\iFun_{\mathrm{b}}(T[1]P)$ cannot be identified with $\iFun(T[1]M)$, as the $\DD\mathsans{M}$--action of $P$
is free but generally not fiberwise transitive. In the case of a trivial principal $2$--bundle, however, 
$P=M\times\DD\mathsans{M}$, the $\DD\mathsans{M}$--action is both free and fiberwise transitive
and $\iFun_{\mathrm{b}}(T[1]P)$ is isomorphic to $\iFun(T[1]M)$. 
So, since a principal $2$--bundle is locally weakly isomorphic to a trivial 
$2$--bundle with the same structure $2$--group by definition, 
the basic internal functions of $T[1]P$ can still be identified
with the internal functions of $T[1]M$ locally in a weak sense. By this feature, 
the basic theory of the higher case is definitely unlike that of the ordinary one. Appropriate notions
are so required for its formulation and construction. It is possible in principle to work out the basic 
theory also for the internal functions of $T[1]P_0$ and similar considerations apply. However, there 
apparently are no relevant applications of it. 

On a trivializing neighborhood $U\subset M$ of the principal $2$--bundle, $2$--connections 
and $1$-- and $2$--gauge transformations are described by basic Lie valued data on the 
portion of $T[1]P$ above $T[1]U$ (cf. subsects. \cref{subsec:maurer}, \cref{subsec:basic}). 
More specifically a $2$--connection $A$ is characterized by a local basic degree 
$1$ $\DD\mathfrak{m}$--valued internal function 
$A_{\mathrm{b}}$ comprising a degree $1$ $\mathfrak{g}$--valued function $\omega_{\mathrm{b}}$ 
and a degree $2$ $\mathfrak{e}$--valued function $\varOmega_{\mathrm{b}}$.
Similarly, a $1$--gauge transformation $\varPsi$ is characterized by a local basic degree $0$ $\DD\mathsans{M}$--valued internal 
function $\varPsi_{\mathrm{b}}$ comprising a degree $0$ $\mathsans{G}$--valued function $g_{\mathrm{b}}$ 
and a degree $1$ $\mathfrak{e}$--valued function $J_{\mathrm{b}}$.
$2$--gauge transformations too have a basic representation.
Local $2$--connection and $1$-- and $2$--gauge transformation data relative to distinct overlapping trivializing 
neighborhoods of $U,U'\subset M$ match trough a local basic degree $0$ $\DD\mathsans{M}$--valued internal function  
$D_{\mathrm{b}}$ decomposable in a degree $0$ $\mathsans{G}$--valued function $f_{\mathrm{b}}$ 
and a degree $1$ $\mathfrak{e}$--valued function $F_{\mathrm{b}}$. 

The local basic data mentioned in the previous paragraph
can be constructed for a full open covering of $M$ made of trivializing neighborhoods (cf. \cref{subsec:nadifcoh}). 
Under certain conditions, among which fake flatness, 
the local $2$--connection and matching data fit into a structure called a differential paracocycle
having formal properties analogous to those of a (trivial) differential cocycle but defined on the total space 
morphism manifold $P$ rather than the base manifold $M$. The paracocycle data are then expressed
through the pull--back of the bundle's projection map in terms of local Lie valued data
defined on $M$ constituting a genuine differential cocycle. Similarly, 
in the presence of a suitable differential paracocycle,
the local $1$--gauge transformation data fit into a structure called a gauge paraequivalence
subordinated to it. The paraequivalence data are then expressed
through the projection map's pull--back in terms of local Lie valued data
defined on $M$. Further, the gauge transform of the paracocycle is defined.

In sect. \cref{sec:outlk}, we evaluate the results of the total space synthetic theory of 
$2$--connections and $1$-- and $2$--gauge transformations illustrated above by comparing it with other approaches 
to the topic (cf. subsect. \cref{subsec:wald}) and outlining a more geometric interpretation 
of it (cf. subscet. \cref{subsec:interpr}).

\subsection{\textcolor{blue}{\sffamily Outlook}}\label{subsec:outlook}

Our work is an attempt to formulate principal $2$--bundle geometry in a total space perspective,
while remaining committed as much as possible to the language and the techniques of 
graded differential geometry which have shown their usefulness in gauge theory. 
The operational formulation we propose enriches and completes the range of approaches to and 
descriptions of principal $2$--bundle geometry. It may provide, it is our hope, alternative more 
elegant proofs of known facts and point to new hitherto unknown developments. 

The operational framework has shown its power in the study of the differential topology, in 
particular the characteristic classes, of ordinary principal bundles \ccite{Greub:1973ccc}. It has thus
the potential of being useful in the study of the corresponding problems for strict 
principal $2$--bundles. 

More specific applications may include a strict Lie $2$--algebraic extension of the classic theory 
of coadjoint orbits  \ccite{Kirillov:1962rnlg} and the attendant 
Borel--Bott--Weil theory \ccite{Bott:1957abc}, which at key points invoke a total space 
description of principal bundles.
Coadjoint orbit and Borel--Bott--Weil play an important role in the one--dimensional path integral representation 
of Wilson lines (see ref. \ccite{Beasley:2009mb} for a nice review of this topic). It is conceivable that their higher 
counterparts may enter prominently in a two--dimensional path integral representation 
of Wilson surfaces \ccite{Alekseev:2015hda,Chekeres:2018kmh,Zucchini:2019mbz}. 

The operational framework has also some non standard features which call for
further investigation. The use of external function algebras introduces 
internal multiplicities and endows $2$--connections and $1$-- and $2$--gauge transformations 
with ghostlike partners rendering the whole geometrical framework akin to that used in 
the AKSZ formulation of BV theory \ccite{Alexandrov:1995kv}  (see also \ccite{Zucchini:2017nax}).
These are absent though  could be added in the ordinary operational framework.
In higher one, they are instead unavoidable.

\vfil\eject

\section{\textcolor{blue}{\sffamily Connections and gauge transformations}}\label{sec:ordgau}

In this section, we review the total space theory of principal bundle connections and gauge 
transformations from an operational perspective. This will furnish a guiding model for the construction
of the corresponding higher theory carried out later in sect. \cref{sec:2conn1gau}.
For a comprehensive treatment, we refer the reader to \ccite{Greub:1973ccc}.


\subsection{\textcolor{blue}{\sffamily Operational theory 
}}\label{subsec:ordgau}


The operational total space theory of principal bundles, expounded in this subsection, 
relies on the operational setup of 
subsect. 2.1 of I.
As shown in subsect. 2.2 of I, 
with a principal $\mathsans{G}$--bundle $P$
there is associated the Lie group space $S_P=(P,\mathsans{G},R)$ and with this the operation 
$\OOO S_P=(\Fun(T[1]P),\mathfrak{g})$. $\OOO S_P$ provides a powerful graded differential 
geometric framework for the study of connections and gauge transformations.
Following the customary point of view, 
the ordinary function algebra $\Fun(T[1]P)$ is considered here. Much of the theory
presented below could be formulated also assuming the internal function 
algebra $\iFun(T[1]P)$.
In higher gauge theory, the latter turns out to be the only available option, as we shall see in due course. 

\begin{defi}
A connection of $P$ is a 
pair of Lie algebra valued functions $\omega\in\Map(T[1]P,\mathfrak{g}[1])$ and 
$\theta\in\Map(T[1]P,\mathfrak{g}[2])$, called respectively connection and curvature component, 
on which the operation derivations act as \hphantom{xxxxxxxxxxxxxxx}
\begin{align} 
&d_P\omega=-\frac{1}{2}[\omega,\omega]+\theta,
\vphantom{\Big]}
\label{ordconn1}
\\
&d_P\theta=-[\omega,\theta],
\vphantom{\Big]}
\label{ordconn3}
\\
&j_{Px}\omega=x,
\vphantom{\Big]}
\label{ordconn5}
\\
&j_{Px}\theta=0,
\vphantom{\Big]}
\label{ordconn7}
\\
&l_{Px}\omega=-[x,\omega],
\vphantom{\Big]}
\label{ordconn9}
\\
&l_{Px}\theta=-[x,\theta]
\vphantom{\Big]}
\label{ordconn11}
\end{align} 
with $x\in\mathfrak{g}$.
\end{defi}

\noindent
\ceqref{ordconn1} is just the expression of the curvature component 
$\theta$ in terms of the connection one $\omega$. 
\ceqref{ordconn3} is the Bianchi identity obeyed by the 
curvature. The connection is said flat if $\theta=0$. 
The definition of connection we gave in subsect. 2.2 of I
is essentially the same as the one provided here.
Indeed, it can be shown that the horizontal $\mathsans{G}$--invariant distribution $H$ in term of which 
former definition is formulated corresponds to the annihilator of $\omega$ of the latter one and 
the flatness conditions of the two notions are equivalent.
 
\begin{defi}
A gauge transformation of $P$ is a pair 
of Lie group and algebra valued functions $g\in\Map(T[1]P,\mathsans{G})$
and $h\in\Map(T[1]P,\mathfrak{g}[1])$, called respectively transformation and shift component, 
on which the operation derivations act as 
\begin{align}
&d_Pgg^{-1}=-h,
\vphantom{\Big]}
\label{ordgautr5}
\\
&d_Ph=-\frac{1}{2}[h,h],
\vphantom{\Big]}
\label{ordgautr7}
\\
&j_{Px}gg^{-1}=0, 
\vphantom{\Big]}
\label{ordgautr9}
\\
&j_{Px}h=x-\Ad g(x), 
\vphantom{\Big]}
\label{ordgautr11}
\\
&l_{Px}gg^{-1}=-x+\Ad g(x),
\vphantom{\Big]}
\label{ordgautr13}
\\
&l_{Px}h=-[x,h]
\vphantom{\Big]}
\label{ordgautr15}
\end{align}
with $x\in\mathfrak{g}$.  
\end{defi}

\noindent
Relations \ceqref{ordgautr5} effectively defines
the shift component $h$  in terms of the transformation one $g$. \ceqref{ordgautr7}
is the associated Maurer--Cartan equation. 
The definition of gauge transformation  we gave in subsect. 2.2 of I
coincides with the one provided here.
The $\mathsans{G}$--equivariant fiber preserving diffeomorphism $\varPhi$ 
in the former definition corresponds to the transformation component $g$ in the latter one. 

As well--known, gauge transformations act on connections of $P$. 

\begin{defi}
The gauge transform of a connection of components $\omega$, $\theta$ by a gauge transformation 
of components $g$, $h$ is given by \hphantom{xxxxxxxxxxxx}
\begin{align}
&{}^{g,h}\omega=\Ad g(\omega)+h, 
\vphantom{\Big]}
\label{000cmgautr1}
\\
&{}^{g,h}\theta=\Ad g(\theta).
\vphantom{\Big]} 
\label{000cmgautr3}
\end{align}
\end{defi}

\noindent
Substituting above identity \ceqref{ordgautr5} expressing $h$ in terms of $g$, 
these relations are formally identical to the familiar ones of standard gauge theory. 

\begin{prop}
${}^{g,h}\omega$, ${}^{g,h}\theta$ are the components of a connection. Flatness 
is gauge invariant. 
\end{prop}

\noindent
Indeed, the action 
\ceqref{ordconn1}--\ceqref{ordconn11} and \ceqref{ordgautr5}--\ceqref{ordgautr15} of the operation derivations 
on the components $\omega$, $\theta$ and $g$, $h$ ensures that 
the action of those derivations on the transformed components 
${}^{g,h}\omega$, ${}^{g,h}\theta$ also obey to \ceqref{ordconn1}--\ceqref{ordconn11}.

Since the shift component of a gauge transformation is determined by the transformation one
by \ceqref{ordgautr5}, a gauge transformation is fully specified by these latter. As 
$\Map(T[1]P,\mathsans{G})=\Map(P,\mathsans{G})$, gauge transformations can be viewed as elements 
of the group $\Map(P,G)$ of $G$--valued maps. They form indeed a 
distinguished subgroup of this latter, the gauge group of $P$.
Gauge transformation is a left action of the gauge group on connection space.
As expected, the definitions of gauge group and the gauge transformation action on connection space
we gave in subsect. 2.2 of I precisely correlate to the operational 
theoretic definitions of the same notions provided  here. 


\subsection{\textcolor{blue}{\sffamily Basic theory 
}}\label{subsec:ordgau1}

Every  principal $\mathsans{G}$--bundle $P$ is 
trivializable on any sufficiently small neighborhood $U$ of the base $M$, 
that is $\pi^{-1}(U)$ is projection preservingly,
$\mathsans{G}$--equivariantly isomorphic to the trivial $\mathsans{G}$--bundle $U\times\mathsans{G}$. 
The existence of a trivializing isomorphism $\varPhi_U:\pi^{-1}(U)\rightarrow U\times\mathsans{G}$
provides structural information about the operation $\OOO S_{\pi^{-1}(U)}$ of $\pi^{-1}(U)$.
It entails the existence of coordinates of $\pi^{-1}(U)$ modelled on $U\times\mathsans{G}$
with special properties under the action of the operation's
derivations. In this way, an operational description of the local fibered geometry of
$P$ can be furnished.

\begin{prop}\label{prop:0adptcrd1}
There are coordinates of $\pi^{-1}(U)$ adapted to $U\times\mathsans{G}$,
namely functions $u\in\Map(T[1]\pi^{-1}(U),\mathbb{R}^{\dim M})$, 
$v\in\Map(T[1]\pi^{-1}(U),\mathbb{R}^{\dim M}[1])$  
for $U$ and $\gamma\in\Map(T[1]\pi^{-1}(U),\mathsans{G})$, 
$\sigma\in\Map(T[1]\pi^{-1}(U),\mathfrak{g}[1])$  for $\mathsans{G}$ 
on which the operation derivations act as follows. For $u$, $v$, one has 
\begin{align}
&d_{\pi^{-1}(U)}u=v, \qquad d_{\pi^{-1}(U)}v=0
\vphantom{\Big]}
\label{00maurer-1}
\end{align}
with trivial action of all operation derivations $j_{\pi^{-1}(U)x}$, $l_{\pi^{-1}(U)x}$
with $x\in\mathfrak{g}$. For $\gamma$, $\sigma$, the structure equations take the form 
\hphantom{xxxxxxxxxxxxxxxx}
\begin{align}
&\gamma{}^{-1}d_{\pi^{-1}(U)}\gamma=\sigma,
\vphantom{\Big]}
\label{00maurer2}
\\
&d_{\pi^{-1}(U)}\sigma=-\frac{1}{2}[\sigma,\sigma],
\vphantom{\Big]}
\label{00maurer4}
\\
&\gamma{}^{-1}j_{\pi^{-1}(U)x}\gamma=0,
\vphantom{\Big]}
\label{00maurer6}
\\
&j_{\pi^{-1}(U)x}\sigma=x,
\vphantom{\Big]}
\label{00maurer8}
\\
&\gamma{}^{-1}l_{\pi^{-1}(U)x}\gamma=x,
\vphantom{\Big]}
\label{00maurer10}
\\
&l_{\pi^{-1}(U)x}\sigma=-[x,\sigma]
\vphantom{\Big]}
\label{00maurer12}
\end{align} 
with $x\in\mathfrak{g}$. 
\end{prop}

\noindent
Relation \ceqref{00maurer-1} can be viewed as the definition of the generator $\upsilon$. 
Relation \ceqref{00maurer2} can similarly be viewed as the 
definition of the generator $\sigma$. Eq. \ceqref{00maurer4} states that $\sigma$ is a fiberwise
Maurer--Cartan form and \ceqref{00maurer4} itself is the classic Maurer--Cartan equation it satisfies. 

We can use the explicit description of the operation $\OOO S_{\pi^{-1}(U)}$ 
we detailed above to analyze such structures as connections and gauge transformations 
of the principal bundle $P$ in terms of data defined locally on $U$ in the base $M$. 
This will lead to basic theory. 

Consider a connection of $P$ of connection and curvature components 
$\omega$, $\theta$. 

\begin{defi}
The basic components of the connection on $U$ are defined as 
\begin{align}
&\omega_{\mathrm{b}}
=\Ad\gamma(\omega-\sigma),  
\vphantom{\Big]}
\label{00basic1}
\\
&\theta_{\mathrm{b}}
=\Ad\gamma(\theta).
\vphantom{\Big]}
\label{00basic3}
\end{align} 
\end{defi}

\noindent
Above, restriction of $\omega$, $\theta$ to $T[1]\pi^{-1}(U)$ is tacitly understood.
The name given to $\omega_{\mathrm{b}}$, $\theta_{\mathrm{b}}$ is motivated by the fact that, 
by construction,  \pagebreak they are annihilated by all derivations $j_{\pi^{-1}(U)x}$ 
and $l_{\pi^{-1}(U)x}$ with $x\in\mathfrak{g}$. 

\begin{prop}
$\omega_{\mathrm{b}}$, $\theta_{\mathrm{b}}$ are basic elements of the operation $\OOO S_{\pi^{-1}(U)}$. 
\end{prop}

\noindent
Hence, $\omega_{\mathrm{b}}$, $\theta_{\mathrm{b}}$ can be identified with certain functions 
$\omega_{\mathrm{b}}\in\Map(T[1]U,\mathfrak{g}[1])$, $\theta_{\mathrm{b}}\in\Map(T[1]U,\mathfrak{g}[2])$. 

\begin{prop}
$\omega_{\mathrm{b}}$, $\theta_{\mathrm{b}}$ satisfy the relations 
\begin{align}
&d_{\pi^{-1}(U)}\omega_{\mathrm{b}}=-\frac{1}{2}[\omega_{\mathrm{b}},\omega_{\mathrm{b}}]+\theta_{\mathrm{b}},
\vphantom{\Big]}
\label{00basic5}
\\
&d_{\pi^{-1}(U)}\theta_{\mathrm{b}}=-[\omega_{\mathrm{b}},\theta_{\mathrm{b}}].
\vphantom{\Big]}
\label{00basic7}
\end{align} 
\end{prop}

\noindent
These are formally identical to relations \ceqref{ordconn1}, \ceqref{ordconn3}.  
We recover in this way the familiar local base space description of connections used 
in the space--time formulation of gauge theory. 

Next, consider a gauge transformation of $P$ of transformation and shift components 
$g$, $h$. 

\begin{defi}
The basic components of the gauge transformation on $U$ are 
\begin{align}
&g_{\mathrm{b}}=\gamma g\gamma^{-1},
\vphantom{\Big]}
\label{00basic9}
\\
&h_{\mathrm{b}}=\Ad\gamma(h-\sigma+\Ad g(\sigma)).
\vphantom{\Big]}
\label{00basic11}
\end{align} 
\end{defi}

\noindent
Above, again, restriction of $g$, $h$ to $T[1]\pi^{-1}(U)$ is understood.
The name given to $g_{\mathrm{b}}$, $h_{\mathrm{b}}$ is motivated by the fact that, 
by construction, they are annihilated by all derivations $j_{\pi^{-1}(U)x}$ 
and $l_{\pi^{-1}(U)x}$ with $x\in\mathfrak{g}$. 

\begin{prop}
$g_{\mathrm{b}}$, $h_{\mathrm{b}}$ are basic elements of the operation $\OOO S_{\pi^{-1}(U)}$. 
\end{prop}

\noindent
Therefore, again, $g_{\mathrm{b}}$, $h_{\mathrm{b}}$ can be identified with functions 
$g_{\mathrm{b}}\in\Map(T[1]U,\mathsans{G})$, $h_{\mathrm{b}}\in\Map(T[1]U,\mathfrak{g}[1])$. 

\begin{prop}
$g_{\mathrm{b}}$, $h_{\mathrm{b}}$ satisfy the relations
\begin{align}
&d_{\pi^{-1}(U)}g_{\mathrm{b}}g_{\mathrm{b}}{}^{-1}=-h_{\mathrm{b}},
\vphantom{\Big]}
\label{00basic13}
\\
&d_{\pi^{-1}(U)}h_{\mathrm{b}}=-\frac{1}{2}[h_{\mathrm{b}},h_{\mathrm{b}}].
\vphantom{\Big]}
\label{00basic15}
\end{align} 
\end{prop}

\noindent
These are  formally identical \pagebreak to relations \ceqref{ordgautr5}, \ceqref{ordgautr7}. 
We recognize here the familiar local base space description of
gauge transformations of standard gauge theory. 

Next, consider the gauge transformed connection ${}^{g,h}\omega$, ${}^{g,h}\theta$.
A simple calculation yields the following result. \vspace{.5mm}

\begin{prop}
The basic components ${}^{g,h}\omega_{\mathrm{b}}$, ${}^{g,h}\theta_{\mathrm{b}}$ 
of the gauge transformed connection are given by 
\begin{align}
&({}^{g,h}\omega)_{\mathrm{b}}
=\Ad g_{\mathrm{b}}(\omega_{\mathrm{b}})+h_{\mathrm{b}},
\vphantom{\Big]}
\label{00basic21}
\\
&({}^{g,h}\theta)_{\mathrm{b}}
=\Ad g_{\mathrm{b}}(\theta_{\mathrm{b}}).
\vphantom{\Big]}
\label{00basic23}
\end{align} 
\end{prop}

\noindent
These have the   same form as relations \ceqref{000cmgautr1}, \ceqref{000cmgautr3}. 
If, with an abuse of notation, we read the above expressions as 
${}^{g,h}\omega_{\mathrm{b}}={}^{g_{\mathrm{b}},h_{\mathrm{b}}}\omega_{\mathrm{b}}$, 
${}^{g,h}\theta_{\mathrm{b}}={}^{g_{\mathrm{b}},h_{\mathrm{b}}}\theta_{\mathrm{b}}$, 
we recover the familiar local base space description of gauge transformations in gauge theory. 

For a given trivializing neighborhood $U\subset M$,
the basic components of connections and gauge transformations are Lie valued 
functions on $T[1]U$, so they are only locally defined. The problem arises
of matching the local data pertaining to distinct but overlapping trivializing 
neighborhoods $U,\, U'\subset M$. Below, we denote by $u,v,\gamma,\sigma$ and 
$u',v',\gamma',\sigma'$ the standard adapted coordinates of $\pi^{-1}(U)$
$\pi^{-1}(U')$, respectively. 

\begin{defi}
The local basic matching transformation and shift components are the Lie group and algebra valued 
functions $f_{\mathrm{b}}\in\Map(T[1]\pi^{-1}(U\cap U'),\mathsans{G})$ and 
$s_{\mathrm{b}}\in\Map(T[1]\pi^{-1}(U\cap U'),\mathfrak{g}[1])$ defined by 
\begin{align}
&f_{\mathrm{b}}=\gamma'\gamma^{-1},
\vphantom{\Big]}
\label{02basic6}
\\
&s_{\mathrm{b}}=\Ad\gamma(\sigma'-\sigma).
\vphantom{\Big]}
\label{02basic8}
\end{align}
\end{defi}

\noindent
Above, $\gamma$, $\sigma$ and $\gamma'$, $\sigma'$ are tacitly restricted to $T[1]\pi^{-1}(U\cap\,U')$. 

\begin{prop}\label{0prop:2basic1}
$f_{\mathrm{b}}$, $s_{\mathrm{b}}$ are basic elements of the operation $\OOO S_{\pi^{-1}(U\cap U')}$.
\end{prop}

\noindent
Therefore, $f_{\mathrm{b}}$, $s_{\mathrm{b}}$ can be identified with functions 
$f_{\mathrm{b}}\in\Map(T[1](U\cap U'),\mathsans{G})$, $s_{\mathrm{b}}\in\Map(T[1](U\cap U'),\mathfrak{g}[1])$. 

\begin{prop}
The local basic components $\omega_{\mathrm{b}}$, $\theta_{\mathrm{b}}$
and $\omega'{}_{\mathrm{b}}$, $\theta'{}_{\mathrm{b}}$ 
of a connection $\omega$, $\theta$ are related on $T[1](U\cap U')$ as 
\begin{align}
&\omega'{}_{\mathrm{b}}=\Ad f_{\mathrm{b}}(\omega_{\mathrm{b}}-s_{\mathrm{b}}),
\vphantom{\Big]}
\label{02basic14}
\\
&\theta'{}_{\mathrm{b}}=\Ad f_{\mathrm{b}}(\theta_{\mathrm{b}}).
\vphantom{\Big]}
\label{02basic16}
\end{align}
\end{prop}

\noindent
Upon observing that $s_{\mathrm{b}}=f_{\mathrm{b}}{}^{-1}d_{\pi^{-1}(U\cap U')}f_{\mathrm{b}}$, 
one recognizes above the well--known matching relations of local connection data. 

\begin{prop}
The local basic components $g_{\mathrm{b}}$, $h_{\mathrm{b}}$ 
and $g'{}_{\mathrm{b}}$, $h'{}_{\mathrm{b}}$ 
of a gauge transformation $g$, $h$ are related on $T[1](U\cap U')$ as  
\begin{align}
&g'{}_{\mathrm{b}}=f_{\mathrm{b}}g_{\mathrm{b}}f_{\mathrm{b}}{}^{-1},
\vphantom{\Big]}
\label{02basic20}
\\
&h'{}_{\mathrm{b}}=\Ad f_{\mathrm{b}}(h_{\mathrm{b}}-s_{\mathrm{b}}+\Ad g_{\mathrm{b}}(s_{\mathrm{b}})).
\vphantom{\Big]}
\label{02basic22}
\end{align}
\end{prop}

\noindent
The above are the matching relations of local gauge transformation data.

Upon choosing an open covering $\{U_i\}$ of $M$ and for each set $U_i$ adapted 
coordinates $u_i$, $v_i$, $\gamma_i$, $\sigma_i$, one can describe a connection, respectively   
a gauge transformation, by means of the collection $\{\omega_{\mathrm{b}i},\theta_{\mathrm{b}i}\}$,
respectively $\{g_{\mathrm{b}i},h_{\mathrm{b}i}\}$, of its local basic data defined according 
\ceqref{00basic1}, \ceqref{00basic3}, respectively \ceqref{00basic9}, \ceqref{00basic11}
on the $U_i$. 
The matching of the local connection and gauge transformation data is controlled through 
the rules \ceqref{02basic14}, \ceqref{02basic16} and \ceqref{02basic20}, \ceqref{02basic22}
by the local basic matching data $\{f_{\mathrm{b}ij},s_{\mathrm{b}ij}\}$ defined according to \ceqref{02basic6},
\ceqref{02basic8} on the non empty intersections $U_i\cap U_j$, respectively. This yields the 
familiar differential cocycle theory of connections and gauge transformations.





\vfil\eject

\section{\textcolor{blue}{\sffamily 2--connections and 1-- and 2--gauge transformations}}\label{sec:2conn1gau}

In this section, we construct the synthetic operational total space theory of 
$2$--connections and $1$-- and $2$--gauge transformations of a principal $2$--bundle
taking the standard connection and gauge transformation reviewed in 
theory of subsect. \cref{subsec:ordgau} as a model.  
We also show that, just as in the ordinary case, 
a basic framework can be 
worked out pointing in this way to a more conventional base space theory. 
Finally, an explanation of the eventual relation of the formulation presented 
to the theory of non Abelian differential cocycles is put forward. 


\subsection{\textcolor{blue}{\sffamily General remarks on the operational setup}}\label{subsec:setup}

In what follows, we systematically refer to the synthetic apparatus of principal 2--bundle theory of 
subsect. 3.2 of I. The basic geometrical datum is so 
a principal $\hat{\matheul{K}}$--$2$--bundle $\hat{\mathcal{P}}$. Its associated synthetic setup 
comprises the synthetic morphism and object Lie groups $\mathsans{K}$, $\mathsans{K}_0$ 
of $\hat{\matheul{K}}$, the synthetic morphism and object manifolds $P$, $P_0$ of $\hat{\mathcal{P}}$
together with their projections $\pi$, $\pi_0$ on the base manifold $M$,  
the synthetic right $\mathsans{K}$--, $\mathsans{K}_0$--actions $R$, $R_0$ of $P$, $P_0$
and for any small open neighborhood $U\subset M$ synthetic $\mathsans{K}$--, $\mathsans{K}_0$--equivariant
trivializing maps $\varPhi_U$, $\varPhi_{U0}$, respectively. 

In the synthetic theory, $2$--connections and $1$-- and $2$--gauge transformations  
of $\hat{\mathcal{P}}$ are Lie valued graded differential forms on $P$
suitably transforming under the $\mathsans{K}$--action $R$.
These notions are best formulated by describing $\mathsans{K}$ 
as the derived Lie group $\DD\mathsans{M}$ of the Lie group crossed module 
$\mathsans{M}=(\mathsans{E},\mathsans{G})$ underlying $\hat{\matheul{K}}$
on one hand and the graded differential form algebra of $P$ as the internal function
algebra of $T[1]P$ on the other (cf. subsect. 3.8 of I). 
Because of the role of the $\DD\mathsans{M}$--action $R$ of $P$, the natural setting for 
studying $2$--connections and $1$-- and $2$--gauge transformations  
is provided then by 
the morphism space $S_{P}=(P,\mathsans{M},R)$ of $P$ \linebreak and the associated operation 
$\iOOO S_{P}=(\iFun(T[1]P),\mathfrak{m})$. 

The action of the derivations $j_{PZ}$, $l_{PZ}$ with $Z\in\DD\mathfrak{m}$ of $\iOOO S_{P}$
on the internal function algebra $\iFun(T[1]P)$ is expressed as a rule through the image
$\zeta_{\mathfrak{m}}Z$ of $Z$ under the isomorphism $\zeta_{\mathfrak{m}}:\DD\mathfrak{m}
\xrightarrow{~\simeq~}\DD\mathfrak{m}^+$ (cf. def. 3.19 and prop. 3.26 of I). 
When decomposing $Z$ in its components $x\in\mathfrak{g}$, $X\in\mathfrak{e}[1]$ 
according to 3.4.6 of I, the action is correspondingly expressed through
$x\in\mathfrak{g}$, $\zeta_{\mathfrak{e},1}X\in\mathfrak{e}[1]^+$.
The reason for this is slightly technical. 
The action of the vertical vector fields of $P$ on $\iFun(T[1]P)$  
is necessarily expressed in terms of constant $\DD\mathfrak{m}$--valued internal functions, 
i. e. functions of the space $\iMap(T[1]P,\DD\mathfrak{m})$ arising by pull--back 
by the map $T[1]P\rightarrow *$ of functions of the space $\iMap(*,\DD\mathfrak{m})=\DD\mathfrak{m}^+$,
the cross modality of $\DD\mathfrak{m}$ (cf. subsect. 3.6 of I). 
In an ungraded setting, this careful distinction would make no difference. In a graded one, 
it is demanded by overall consistency. 
However, to simplify the notation, we tacitly shall not distinguish notationally between
$Z$ and $\zeta_{\mathfrak{m}}Z$ and similarly $X$ and $\zeta_{\mathfrak{e},1}X$
in the following. 

The study of the properties of a $2$--connections and $1$--gauge transformations  
on the object manifold $P_0$ as a
submanifold of the morphism manifold $P$ can also be performed. As the right $\DD\mathsans{M}$--action $R$
of $P$ restricts to the the right $\DD\mathsans{M}_0$--action $R_0$ of $P_0$, the appropriate framework for this analysis 
is the object space $S_{P0}=(P_0,\mathsans{M}_0,R_0)$ of $P$ and the associated operation 
$\iOOO S_{P0}=(\iFun(T[1]P_0),\mathfrak{m}_0)$ (cf. subsect. 3.8 of I). 
The action of the derivations of $\iOOO S_{P0}$ fits with the restriction operation morphism 
$\iOOO L:\iOOO S_{P}\rightarrow\iOOO S_{P0}$. 


\subsection{\textcolor{blue}{\sffamily 2--connections}}\label{subsec:cmconn}

In the synthetic formulation, 
a $2$--connection of the 
$\hat{\matheul{K}}$--$2$--bundle $\hat{\mathcal{P}}$ is a degree $1$ 
$\mathfrak{k}$--valued graded differential form over 
$P$ suitably transforming under the $\mathsans{K}$--action $R$. 
Proceeding along the lines described in subsect. \cref{subsec:setup}, a
$2$--connection is most naturally defined making reference to the 
operation $\iOOO S_{P}$ of $P$. 

\begin{defi}\label{defi:2conn}
A $2$--connection of $P$ is a pair of Lie algebra valued internal functions 
$A\in\iMap(T[1]P,\DD\mathfrak{m}[1])$ and $B\in\iMap(T[1]P,\DD\mathfrak{m}[2])$, called respectively
connection and curvature component, on which the action of the derivations of the operation 
$\iOOO S_{P}$ is given by 
\begin{align}
&d_PA=-\frac{1}{2}[A,A]-d_{\dot\tau}A+B,
\vphantom{\Big]}
\label{cmconnx1}
\\
&d_PB=-[A,B]-d_{\dot\tau}B,
\vphantom{\Big]}
\label{cmconnx2}
\\
&j_{PZ}A=Z,
\vphantom{\Big]}
\label{cmconnx3}
\end{align} 
\vskip-.75cm\eject\noindent
\begin{align}
&j_{PZ}B=0,
\vphantom{\Big]} 
\label{cmconnx4}
\\
&l_{PZ}A=-[Z,A]+d_{\dot\tau}Z,
\vphantom{\Big]}
\label{cmconnx5}
\\
&l_{PZ}B=-[Z,B]
\vphantom{\Big]}
\label{cmconnx6}  
\end{align}
with $Z\in\DD\mathfrak{m}$. 
\end{defi}

\noindent 
Above, $[-,-]$ and $d_{\dot\tau}$  are the Lie bracket and the coboundary of 
the virtual Lie algebra $\iMap(T[1]P,\ZZ\DD\mathfrak{m})$ 
(cf. eqs. 3.5.13, 3.5.15 of I). 
$Z$ is tacitly viewed as an element 
of $\DD\mathfrak{m}^+$ as explained in subsect. \cref{subsec:setup}.
\ceqref{cmconnx1}--\ceqref{cmconnx6} 
are by design formally analogous to relations \ceqref{ordconn1}--\ceqref{ordconn11}
defining an ordinary connection, once one assumes $d_P+d_{\dot\tau}$ as relevant differential.
\ceqref{cmconnx1} is just the expression of the curvature component $B$ in terms of 
the connection component $A$. 
 \ceqref{cmconnx2} is the Bianchi 
identity obeyed by the curvature component. The $2$-connection is said flat if $B=0$. 

\begin{lemma} \label{lemma:2conn}
\ceqref{cmconnx1}--\ceqref{cmconnx6} respect the operation commutation relations 
2.1.1--2.1.6 of I. 
\end{lemma}

\begin{proof}
One has to show that the six derivation commutators in the left hand sides of eqs. 
2.1.1--2.1.6 of I act as the corresponding derivations in the right hand sides 
when they are applied to the functions $A$, $B$ and the \ceqref{cmconnx1}--\ceqref{cmconnx6} 
are used. The graded commutativity of $d_{\dot\tau}$ with all derivations must be taken into account. 
This is a straightforward verification. 
\end{proof}

By 3.5.12 of I, we can express the components $A$, $B$ of a $2$--connection as 
\begin{align}
&A(\alpha)=\omega-\alpha\varOmega,
\vphantom{\Big]}
\label{cmconnx7}
\\
&B(\alpha)=\theta+\alpha\varTheta, \quad \alpha\in\mathbb{R}[1],
\vphantom{\Big]}
\label{cmconnx8}
\end{align}
through projected connection and curvature components
$\omega\in\iMap(T[1]P,\mathfrak{g}[1])$, $\varOmega\in\iMap(T[1]P,\mathfrak{e}[2])$ 
and $\theta\in\iMap(T[1]P,\mathfrak{g}[2])$, $\varTheta\in\iMap(T[1]P,\mathfrak{e}[3])$.
We further write $Z\in\DD\mathfrak{m}$ as $Z(\bar\alpha)=x+\bar\alpha X$, $\bar\alpha\in\mathbb{R}[-1]$, 
with $x\in\mathfrak{g}$ and $X\in\mathfrak{e}[1]$ as in 3.4.6 of I. 

\begin{prop}\label{prop:2connexpl}
In terms of projected components, \pagebreak the operation relations 
\ceqref{cmconnx1}--\ceqref{cmconnx6} take the more explicit form \hphantom{xxxxxxxxxx} 
\begin{align} 
&d_P\omega=-\frac{1}{2}[\omega,\omega]+\dot\tau(\varOmega)+\theta,
\vphantom{\Big]}
\label{cmconn1}
\\
&d_P\varOmega=-\dot{}\mu\dot{}\,(\omega,\varOmega)+\varTheta,
\vphantom{\Big]}
\label{cmconn2}
\\
&d_P\theta=-[\omega,\theta]-\dot\tau(\varTheta),
\vphantom{\Big]}
\label{cmconn3}
\\
&d_P\varTheta=-\dot{}\mu\dot{}\,(\omega,\varTheta)+\dot{}\mu\dot{}\,(\theta,\varOmega),
\vphantom{\Big]}
\label{cmconn4}
\\
&j_{PZ}\omega=x,
\vphantom{\Big]}
\label{cmconn5}
\\
&j_{PZ}\varOmega=X,
\vphantom{\Big]}
\label{cmconn6}
\\
&j_{PZ}\theta=0,
\vphantom{\Big]}
\label{cmconn7}
\\
&j_{PZ}\varTheta=0,
\vphantom{\Big]}
\label{cmconn8}
\\
&l_{PZ}\omega=-[x,\omega]+\dot\tau(X),
\vphantom{\Big]}
\label{cmconn9}
\\
&l_{PZ}\varOmega=-\dot{}\mu\dot{}\,(x,\varOmega)+\dot{}\mu\dot{}\,(\omega,X),
\vphantom{\Big]}
\label{cmconn10}
\\
&l_{PZ}\theta=-[x,\theta],
\vphantom{\Big]}
\label{cmconn11}
\\
&l_{PZ}\varTheta=-\dot{}\mu\dot{}\,(x,\varTheta)+\dot{}\mu\dot{}\,(\theta,X).
\vphantom{\Big]}
\label{cmconn12}
\end{align} 
\end{prop}

\noindent 
Above, $X$ is tacitly viewed as an element 
of $\mathfrak{e}[1]^+$ (cf. subsect. \cref{subsec:setup}).

\begin{proof}
To get these relations, we substitute the expressions of $A$, $B$ 
in terms of $\omega$, $\varOmega$, $\theta$, $\varTheta$ of eqs.  \ceqref{cmconnx7}, \ceqref{cmconnx8} 
and that of $Z$ in terms of $x$, $X$ into 
\ceqref{cmconnx1}--\ceqref{cmconnx6} and use relations 3.5.13 and 3.5.15 of I. 
The calculations are elementary. 
\end{proof}

\noindent 
\ceqref{cmconn1}, \ceqref{cmconn2} are just the expressions of the curvature components 
$\theta$, $\varTheta$  in terms of the connection components $\omega$, $\varOmega$
and \ceqref{cmconn3}, \ceqref{cmconn4} are the Bianchi identities obeyed by $\theta$, $\varTheta$  
familiar in strict higher gauge theory. The $2$--connection is flat if 
$\theta=0$, $\varTheta=0$ and it is said fake flat if $\theta=0$ only. 

The study of the properties of a $2$--connection on $P_0$ as a 
submanifold of $P$ can also be performed. Following the lines 
of subsect. \cref{subsec:setup}, the appropriate way of doing this is by making reference to the 
operation $\iOOO S_{P0}$. 

The restriction operation morphism $\iOOO L:\iOOO S_P\rightarrow\iOOO S_{P0}$ 
of subsect. 3.8 of I maps 
the  components $\omega$, $\varOmega$, $\theta$, $\varTheta$
of \pagebreak a $2$--connection of $P$ into 
\begin{align}
&\omega_0=I^*\omega,
\vphantom{\Big]}
\label{0cmconn1ol}
\\
&\varOmega_0=I^*\varOmega,
\vphantom{\Big]}
\label{0cmconn2ol}
\\
&\theta_0=I^*\theta,
\vphantom{\Big]}
\label{0cmconn3ol}
\\
&\varTheta_0=I^*\varTheta,
\vphantom{\Big]}
\label{0cmconn4ol}
\end{align} 
where $I^*:\iFun(T[1]P)\rightarrow\iFun(T[1]P_0)$ is the restriction morphism
associated to the inclusion map $I:P_0\rightarrow P$. 
The action of the derivations of the operation $\OOO S_{P0}$ on 
$\omega_0$, $\varOmega_0$, $\theta_0$, $\varTheta_0$  is given by the right hand side of eqs.
\ceqref{cmconn1}--\ceqref{cmconn12} \linebreak with $\omega$, $\varOmega$, $\theta$, $\varTheta$
replaced by $\omega_0$, $\varOmega_0$, $\theta_0$, $\varTheta_0$  and $X$ set to $0$. 
By inspecting the resulting expressions, it appears that one can consistently impose the conditions
\begin{align}
&\varOmega_0=0, 
\vphantom{\Big]}
\label{0cmconn5ol}
\\
&\varTheta_0=0. 
\vphantom{\Big]}
\label{0cmconn6ol}
\end{align} 
Upon doing so, the surviving components $\omega_0$, $\theta_0$ satisfy relations formally identical to 
\ceqref{ordconn1}--\ceqref{ordconn11}. 
In spite of the seeming similarities to a connection of a principal $\mathsans{G}$--bundle
there are two basic differences. First, $P_0$ is not a principal $\mathsans{G}$--bundle, as 
the $\mathsans{G}$--action on $P_0$ is free
but fiberwise transitive only up to isomorphism of $P$. Second, in the customary definition of connection
the ordinary function algebras $\Map(T[1]P_0,\mathfrak{g}[p])$ appears. 

In certain cases, it may be appropriate to restrict the range 
of $2$--connections to those enjoying \ceqref{0cmconn5ol}, \ceqref{0cmconn6ol}

\begin{defi}\label{defi:prop2conn}
A $2$--connection $\omega$, $\varOmega$, $\theta$, $\varTheta$ is 
special if \ceqref{0cmconn5ol}, \ceqref{0cmconn6ol}
are sa\-tisfied. 
\end{defi}

\subsection{\textcolor{blue}{\sffamily 1--gauge transformations}}\label{subsec:cmgautr}

In the synthetic formulation,  of higher gauge theory of subsect. 3.2 of I, 
a $1$--gau\-ge transformation of the  $\hat{\matheul{K}}$--$2$--bundle $\hat{\mathcal{P}}$
is a degree $0$ $\mathsans{K}$--valued graded differential form on $P$ 
suitably transforming under the $\mathsans{K}$--–action $R$. 
Proceeding along the lines described in subsect. \cref{subsec:setup}, a
$1$--gauge transformation is most naturally defined making reference to the 
operation $\iOOO S_{P}$, as for a $2$--connection. 

\begin{defi}\label{defi:1gauge}
A $1$--gauge transformation of $P$ is a pair of a Lie group valued internal function
$\varPsi\in\iMap(T[1]P,\DD\mathsans{M})$ and a Lie algebra 
valued internal function $\varUpsilon\in\iMap(T[1]P,\DD\mathfrak{m}[1])$, 
called respectively transformation and shift component, 
on which the action of the derivations of the operation $\iOOO S_{P}$ reads as 
\begin{align}
&d_P\varPsi\varPsi^{-1}=-d_{\dot\tau}\varPsi\varPsi^{-1}-\varUpsilon,
\vphantom{\Big]}
\label{cmgautrx1}
\\
&d_P\varUpsilon=-\frac{1}{2}[\varUpsilon,\varUpsilon]-d_{\dot\tau}\varUpsilon,
\vphantom{\Big]}
\label{cmgautrx2}
\\
&j_{PZ}\varPsi\varPsi^{-1}=0,
\vphantom{\Big]}
\label{cmgautrx3}
\\
&j_{PZ}\varUpsilon=Z-\Ad\varPsi(Z),
\vphantom{\Big]}
\label{cmgautr4x}
\\
&l_{PZ}\varPsi\varPsi^{-1}=-Z+\Ad\varPsi(Z),
\vphantom{\Big]}
\label{cmgautrx5}
\\
&l_{PZ}\varUpsilon=-[Z,\varUpsilon]+d_{\dot\tau}Z-\Ad\varPsi(d_{\dot\tau}Z)
\vphantom{\Big]}
\label{cmgautrx6} 
\end{align}
with $Z\in\DD\mathfrak{m}$. 

\end{defi}

\noindent
The above relations involve several algebraic constructs studied in subsect. 
3.5 of I. $[-,-]$ and $d_{\dot\tau}$ are respectively the Lie bracket and the coboundary of 
the virtual Lie algebra $\iMap(T[1]P,\ZZ\DD\mathfrak{m})$ defined 
in eqs. 3.5.13, 3.5.15) of I. 
$\Ad$ is the adjoint action of the virtual Lie group $\iMap(T[1]P,\DD\mathsans{M})$  
on $\iMap(T[1]P,\ZZ\DD\mathfrak{m})$ given in eq. 3.5.18 of I. 
The terms $D\varPsi\varPsi^{-1}$ with $D$ $=d_P,j_{PZ},l_{PZ}$ 
are the pull--back of the first Maurer--Cartan element of $\DD\mathsans{M}$ by $\varPsi$ 
followed by contraction with $D$ seen as a vector field on $T[1]P$, see eq. 3.5.22 of I. 
The term $d_{\dot\tau}\varPsi\varPsi^{-1}$ is similarly given by eq. 3.5.24 of I.  
$Z$ is tacitly viewed as an element 
of $\DD\mathfrak{m}^+$ as explained in subsect. \cref{subsec:setup}. 
Again, upon considering $d_P+d_{\dot\tau}$ as relevant differential, the 
\ceqref{cmgautrx1}--\ceqref{cmgautrx6} are formally analogous to relations \ceqref{ordgautr5}--\ceqref{ordgautr15}
defining an ordinary gauge transformation. 
Relation \ceqref{cmgautrx1} effectively defines
the shift component $\varUpsilon$ in terms of the transformation component $\varPsi$. \ceqref{cmgautrx2}
is the associated Maurer--Cartan equation. 

\begin{lemma} \label{lemma:1gau}
\ceqref{cmgautrx1}--\ceqref{cmgautrx6} respect the operation commutation relations 
2.1.1--2.1.6 of I. 
\end{lemma}

\begin{proof}
One has to verify that the six derivation \pagebreak commutators in the left hand sides of eqs. 
2.1.1--2.1.6 of I act as the corresponding derivations in the right hand sides 
when they are applied to the functions $\varPsi$, $\varUpsilon$ and the \ceqref{cmgautrx1}--\ceqref{cmgautrx6} 
are used. In the case of $\varPsi$, one must employ systematically the basic relation
$[D,D']GG^{-1}=D(D'GG^{-1})-(-1)^{|D\|D'|}D'(DGG^{-1})-[DGG^{-1},D'GG^{-1}]$ 
holding for two graded derivations $D$, $D'$ and a Lie group valued function $G$. 
The graded commutativity of $d_{\dot\tau}$ with all derivations must further be taken into account.
The verification is straightforward. 
\end{proof}

Since by \ceqref{cmgautrx1} the shift component $\varUpsilon$ of a $1$--gauge transformation can 
be expressed in terms of the transformation component $\varPsi$, a $1$--gauge transformation is effectively specified 
by this latter. $1$--gauge transformations can thus be viewed as elements of the virtual Lie group 
$\iMap(T[1]P,\DD\mathsans{M})$ of $\DD\mathsans{M}$--valued internal functions of $T[1]P$. 
As \ceqref{cmgautrx3}, \ceqref{cmgautrx5} are evidently preserved under the group operations of $\iMap(T[1]P,\DD\mathsans{M})$,
$1$--gauge transformations form in fact a distinguished subgroup of this latter, 
the $1$--gauge group in the present formulation. 

$1$--gauge transformations act on $2$--connections of $P$ (cf. subsect. \cref{subsec:cmconn}, def. \cref{defi:2conn}) 
compatibly with the  $\mathsans{K}$--–action on both types of items. 

\begin{prop}
\label{prop:congauresp}
If $A$, $B$ and $\varPsi$, $\varUpsilon$ are
the components of a $2$--connection and a $1$--gauge transformation, respectively, then
\begin{align}
&A'=\Ad\varPsi(A)+\varUpsilon,
\vphantom{\Big]}
\label{cmgautrx9/1}
\\
&B'=\Ad\varPsi(B)
\vphantom{\Big]}
\label{cmgautrx10/1}
\end{align}
are the components of a $2$--connection. 
\end{prop}

\begin{proof}
To show that $A'$, $B'$ are the components of a $2$--connection, 
we have to check that the action of the derivations of the operation
on $A'$, $B'$ conforms to \ceqref{cmconnx1}--\ceqref{cmconnx6} 
using that the action of those derivations on  $A$, $B$ and $\varPsi$, $\varUpsilon$ 
is given by \ceqref{cmconnx1}--\ceqref{cmconnx6} and \ceqref{cmgautrx1}--\ceqref{cmgautrx6}, respectively.  
This is a matter of a simple calculation. 
\end{proof}

\begin{defi}\label{defi:1gaugetransf}
The gauge transform 
of a $2$--connection of components $A$, $B$ 
by a $1$--gauge transformation of components $\varPsi$, $\varUpsilon$ is the $2$--connection of 
components  \vspace{1mm}\pagebreak 
\begin{align}
&{}^{\varPsi,\varUpsilon}A=\Ad\varPsi(A)+\varUpsilon,
\vphantom{\Big]}
\label{cmgautrx9}
\\
&{}^{\varPsi,\varUpsilon}B=\Ad\varPsi(B).
\vphantom{\Big]}
\label{cmgautrx10}
\end{align}
\end{defi}

\noindent
\ceqref{cmgautrx9}, \ceqref{cmgautrx10} are formally identical to relations \ceqref{000cmgautr1}, 
\ceqref{000cmgautr3} defining the gauge transform of a connection in ordinary principal bundle theory. 
Notice that flatness of a $2$--connections is a $1$--gauge invariant property. 
\ceqref{cmgautrx9}, \ceqref{cmgautrx10} yield a left action of the $1$--gauge transforma\-tion
group on the $2$--connection space, as it is readily verified.

Making use of 3.5.1, 3.5.12 of I, we can express the components $\varPsi$, $\varUpsilon$ of a 
$1$--gauge transformation as \hphantom{xxxxxxxxxxxx}
\begin{align}
&\varPsi(\alpha)=\ee^{\alpha J}g,
\vphantom{\Big]}
\label{cmgautrx7}
\\
&\varUpsilon(\alpha)=h-\alpha K, \quad \alpha\in\mathbb{R}[1], 
\vphantom{\Big]}
\label{cmgautrx8}
\end{align}
by means of projected transformation and shift components $g\in\iMap(T[1]P,\mathsans{G})$, 
$J\in\iMap(T[1]P,\mathfrak{e}[1])$ and $h\in\iMap(T[1]P,\mathfrak{g}[1])$, 
$K\in\iMap(T[1]P,\mathfrak{e}[2])$ (cf. sub\-sect. 3.5 of I).
We further write $Z\in\DD\mathfrak{m}$ as $Z(\bar\alpha)=x+\bar\alpha X$, $\bar\alpha\in\mathbb{R}[-1]$, 
with $x\in\mathfrak{g}$ and $X\in\mathfrak{e}[1]$ as in 3.4.6 of I. 

\begin{prop}\label{prop:1gauexpl}
In terms of projected components, 
the operation relations \ceqref{cmgautrx1}--\ceqref{cmgautrx6} 
take the explicit form 
\begin{align}
&d_Pgg^{-1}=-h-\dot\tau(J),
\vphantom{\Big]}
\label{cmgautr5}
\\
&d_PJ=-K-\frac{1}{2}[J,J]-\dot{}\mu\dot{}\,(h,J),
\vphantom{\Big]}
\label{cmgautr6}
\\
&d_Ph=-\frac{1}{2}[h,h]+\dot\tau(K),
\vphantom{\Big]}
\label{cmgautr7}
\\
&d_PK=-\dot{}\mu\dot{}\,(h,K),
\vphantom{\Big]}
\label{cmgautr8}
\\
&j_{PZ}gg^{-1}=0, 
\vphantom{\Big]}
\label{cmgautr9}
\\
&j_{PZ}J=0, 
\vphantom{\Big]}
\label{cmgautr10}
\\
&j_{PZ}h=x-\Ad g(x), 
\vphantom{\Big]}
\label{cmgautr11}
\\
&j_{PZ}K=\dot{}\mu\dot{}\,(\Ad g(x),J)+X-\mu\dot{}\,(g,X), 
\vphantom{\Big]}
\label{cmgautr12}
\\
&l_{PZ}gg^{-1}=-x+\Ad g(x),
\vphantom{\Big]} 
\label{cmgautr13}
\end{align}
\begin{align}
&l_{PZ}J=-\dot{}\mu\dot{}\,(x,J)-X+\mu\dot{}\,(g,X), 
\vphantom{\Big]}
\label{cmgautr14}
\\
&l_{PZ}h=-[x,h]+\dot\tau(X)-\Ad g(\dot\tau(X)),
\vphantom{\Big]}
\label{cmgautr15}
\\
&l_{PZ}K=-\dot{}\mu\dot{}\,(x,K)+\dot{}\mu\dot{}\,(h,X)+[\mu\dot{}\,(g,X),J].
\vphantom{\Big]}
\label{cmgautr16}
\end{align}
\end{prop}

\noindent 
Above, $X$ is tacitly viewed as an element 
of $\mathfrak{e}[1]^+$ (cf. subsect. \cref{subsec:setup}).

\begin{proof}
To obtain the above relations, we substitute the expressions of $\varPsi$, $\varUpsilon$ 
in terms of $g$ $J$, $h$, $K$ of eqs. \ceqref{cmgautrx7}, \ceqref{cmgautrx8} 
and that of $Z$ in terms of $x$, $X$ into \ceqref{cmgautrx1}--\ceqref{cmgautrx6} 
and use systematically relations 3.5.13, 3.5.15
and 3.5.18 as well as expressions 3.5.22 and 3.5.24 of I. 
This is again a straightforward though a bit length calculation. 
\end{proof}



In the projected framework we are using, so, the $1$--gauge group is the subgroup of 
$\iMap(T[1]P,\mathfrak{e}[1]\rtimes_\mu\mathsans{G})$ formed by the pairs
$g$, $J$ satisfying \ceqref{cmgautr9}, \ceqref{cmgautr10} and \ceqref{cmgautr13}, 
\ceqref{cmgautr14}. $g$, $J$ are indeed the data of a $1$--gauge transformation familiar 
in strict higher gauge theory. 

Expressing the components $A$, $B$ of a $2$--connection and $\varPsi$, $\varUpsilon$ of a 
$1$--gauge transformation in terms of projected components 
$\omega$, $\varOmega$, $\theta$, $\varTheta$  and $g$, $J$, $h$, $K$ 
using \ceqref{cmconnx7}, \ceqref{cmconnx8} and 
\ceqref{cmgautrx7}, \ceqref{cmgautrx8} respectively, we obtain 
projected component expressions of 
the transformation relations \ceqref{cmgautrx9}, \ceqref{cmgautrx10}. 

\begin{prop}\label{prop:1gautrsfexpl}
In terms of projected components, the transformation relations \ceqref{cmgautrx9}, \ceqref{cmgautrx10} 
take the explicit form 
\begin{align}
&{}^{g,J,h,K}\omega=\Ad g(\omega)+h, 
\vphantom{\Big]}
\label{cmgautr1}
\\
&{}^{g,J,h,K}\varOmega=\mu\dot{}\,(g,\varOmega)-\dot{}\mu\dot{}\,(\Ad g(\omega),J)+K, 
\vphantom{\Big]}
\label{cmgautr2}
\\
&{}^{g,J,h,K}\theta=\Ad g(\theta), 
\vphantom{\Big]} 
\label{cmgautr3}
\\
&{}^{g,J,h,K}\varTheta=\mu\dot{}\,(g,\varTheta)-\dot{}\mu\dot{}\,(\Ad g(\theta),J).
\vphantom{\Big]}
\label{cmgautr4}
\end{align} 
\end{prop}

\noindent
One recognizes here the standard expressions of the gauge transform of a $2$--connection 
of strict higher gauge theory. 

\begin{proof}
To obtain the above relations, \pagebreak we substitute the expressions of $A$, $B$ 
in terms of $\omega$, $\varOmega$, $\theta$, $\varTheta$ and $\varPsi$, $\varUpsilon$ 
in terms of $g$ $J$, $h$, $K$ of eqs.  \ceqref{cmconnx7}, \ceqref{cmconnx8} 
and \ceqref{cmgautrx7}, \ceqref{cmgautrx8} respectively, as anticipated above, 
and use 3.5.18 of I. 
\end{proof}


The study of the properties of a $1$--gauge transformation on $P_0$ as a
submanifold of $P$ can also be carried out. Following the lines 
of subsect. \cref{subsec:setup}, this is most naturally done making reference to the 
operation $\iOOO S_{P0}$, just as for a $2$--connection. 

Under the restriction operation morphism $\iOOO L:\iOOO S_{P}\rightarrow\iOOO S_{P0}$, 
of subsect. 3.8 of I, the  components $g$, $J$, $h$, $K$
of a $1$--gauge transformation of $P$ get 
\begin{align}
&g_0=I^*g,
\vphantom{\Big]}
\label{0cmgautr1ol}
\\
&J_0=I^*J,
\vphantom{\Big]}
\label{0cmgautr2ol}
\\
&h_0=I^*h,
\vphantom{\Big]}
\label{0cmgautr3ol}
\\
&K_0=I^*K,
\vphantom{\Big]}
\label{0cmgautr4ol}
\end{align} 
where $I^*:\iFun(T[1]P)\rightarrow\iFun(T[1]P_0)$ is the restriction morphism
associated to the inclusion map $I:P_0\rightarrow P$. 
The action of the derivations of the operation $\iOOO S_{P0}$ on 
$g_0$, $J_0$, $h_0$, $K_0$ is given by the right hand side of eqs.
\ceqref{cmgautr5}--\ceqref{cmgautr16} with $g$, $J$, $h$, $K$
replaced by $g_0$, $J_0$, $h_0$, $K_0$ and $X$ set to $0$.  
From the resulting expressions, it appears that one can consistently impose the conditions 
\begin{align}
&J_0=0, 
\vphantom{\Big]}
\label{0cmgautr5ol}
\\
&K_0=0. 
\vphantom{\Big]}
\label{0cmgautr6ol}
\end{align} 
Upon doing so, the surviving components $g_0$, $h_0$ satisfy 
relations formally identical to \ceqref{ordgautr11}--\ceqref{ordgautr15}.
Again, in spite of similarities to a gauge transformation of an ordinary principal $\mathsans{G}$--bundle, 
the differences recalled below eq. \ceqref{0cmconn6ol} hold and should be kept in mind. 

In certain cases, it may be befitting to restrict the range 
of $1$--gauge transformations so as to allow only those enjoying the above property. 

\begin{defi}\label{defi:prop1gau}
A $1$--gauge transformation $g$, $J$, $h$, $K$ is special if \ceqref{0cmgautr5ol}, \ceqref{0cmgautr6ol}
are met. 
\end{defi}

\eject\noindent
Special $1$--gauge transformations form a subgroup of the $1$--gauge group. 

In subsect. \cref{subsec:cmconn}, def. \cref{defi:prop2conn},   we introduced the notion of special $2$--connection
which pairs with that of special $1$--gauge transformation put forward above.
It turns out that the action of $1$--gauge transformations on $2$--connections is compatible with specialty
in the following sense. 

\begin{prop}
If a $2$--connection $\omega$, $\varOmega$, $\theta$, $\varTheta$
and a $1$--gauge transformation $g$, $J$, $h$, $K$ are both special
then the $1$--gauge transformed $2$--connection 
${}^{g,J,h,K}\omega$, ${}^{g,J,h,K}\varOmega$, ${}^{g,J,h,K}\theta$, ${}^{g,J,h,K}\varTheta$ also 
is. 
\end{prop}

\begin{proof}
Inspection of \ceqref{cmgautr2}, \ceqref{cmgautr4} shows that when 
the $\omega$, $\varOmega$, $\theta$, $\varTheta$ and $g$, $J$, $h$, $K$
satisfy respectively \ceqref{0cmconn5ol}, \ceqref{0cmconn6ol} and \ceqref{0cmgautr5ol}, \ceqref{0cmgautr6ol}, 
then ${}^{g,J,h,K}\omega$, ${}^{g,J,h,K}\varOmega$, ${}^{g,J,h,K}\theta$, ${}^{g,J,h,K}\varTheta$ also 
satisfy \ceqref{0cmconn5ol}, \ceqref{0cmconn6ol} as well, as required. 
\end{proof}

\noindent 
Comparison of \ceqref{cmgautr1}, \ceqref{cmgautr3}
and \ceqref{000cmgautr1}, \ceqref{000cmgautr3} shows further that 
\begin{align}
&{}^{g,J,h,K}\omega_0={}^{g_0,h_0}\omega_0,
\vphantom{\Big]}
\label{0cmgautr9ol}
\\
&{}^{g,J,h,K}\theta_0={}^{g_0,h_0}\theta_0, 
\vphantom{\Big]}
\label{0cmgautr10ol}
\end{align} 
where $\omega_0$, $\theta_0$ and $g_0$, $h_0$ are formally treated as an ordinary connection and
gauge transformation, respectively. In this sense, 
one recovers in this way the well--known expressions of the gauge transform of 
a connection.



\subsection{\textcolor{blue}{\sffamily 2--gauge transformations}}\label{subsec:cmgauforgau}

In the synthetic formulation, 
a $2$--gauge transformation of the $\hat{\matheul{K}}$--$2$--bundle $\hat{\mathcal{P}}$ 
is a degree $0$ $\mathsans{E}$--valued graded differential form on $P$ suitably transforming 
under the $\mathsans{K}$--action $R$. Hence, $\mathsans{E}$ instead of $\DD\mathsans{M}$ 
is the relevant target group in this case. The operational framework remains however
perfectly adequate. 
In this way, a $2$--gauge transformation is most naturally defined by making again reference to the
operation $\iOOO S_{P}$. 

To make contact with the standard higher 
gauge theoretic treatment of $2$--gauge transformations, it is necessary to express the 
action of the operation derivations with reference to a given $2$--connection of $\hat{\mathcal{P}}$
(cf. subsect. \cref{subsec:cmconn}). We assume so that a $2$--connection of projected components
$\omega$, $\varOmega$, $\theta$, $\varTheta$ is assigned. 

\begin{defi}\label{defi:2gauge}
We define a $2$--gauge transformation as a pair of a 
Lie group valued internal function $E\in\iMap(T[1]P,\mathsans{E})$
and a Lie algebra valued internal function $C\in\iMap(T[1]P,\mathfrak{e}[1])$, 
called respectively modification and variation component, 
which are acted upon by the operation derivations as
\begin{align}
&d_PEE^{-1}=-C-\dot{}\mu(\omega,E), 
\vphantom{\Big]}
\label{cmgauforgau7}
\\
&d_PC=-\frac{1}{2}[C,C]-\dot{}\mu\dot{}\,(\omega,C)-\dot{}\mu(\theta,E)
-\varOmega+\Ad E(\varOmega), 
\vphantom{\Big]}
\label{cmgauforgau8}
\\
&j_{PZ}EE^{-1}=0,
\vphantom{\Big]}
\label{cmgauforgau9}
\\
&j_{PZ}C=0, 
\vphantom{\Big]}
\label{cmgauforgau10}
\\
&l_{PZ}EE^{-1}=-\dot{}\mu(x,E),
\vphantom{\Big]}
\label{cmgauforgau11}
\\
&l_{PZ}C=-\dot{}\mu\dot{}\,(x,C)-X+\Ad E(X)
\vphantom{\Big]}
\label{cmgauforgau12}
\end{align} 
with $Z\in\DD\mathfrak{m}$ written in terms of its projected components
$x\in\mathfrak{g}$, $X\in\mathfrak{e}[1]$. 
\end{defi}

\noindent
Above, $X$ is tacitly viewed as an element 
of $\mathfrak{e}[1]^+$ as in earlier instances.
Relations \ceqref{cmgauforgau7} effectively defines 
the variation component $C$ in terms of the modification component $E$
and the reference $2$--connection $\omega$, $\varOmega$, $\theta$, $\varTheta$. 
\ceqref{cmgauforgau8} is the corresponding Bianchi type identity. 

\begin{lemma}
\ceqref{cmgauforgau7}--\ceqref{cmgauforgau12} respect 
the operation commutation relations 2.1.1--2.1.6 of I.
\end{lemma}

\begin{proof}
The proof consists in checking that the six derivation commutators in the left hand sides of eqs. 
2.1.1--2.1.6 act as the corresponding derivations in the right hand sides 
when they are applied to the functions $E$, $C$ and the \ceqref{cmgauforgau7}--\ceqref{cmgauforgau12} 
are used. In the case of $E$,  it is necessary to use the basic relation
$[D,D']GG^{-1}=D(D'GG^{-1})-(-1)^{|D\|D'|}D'(DGG^{-1})-[DGG^{-1},D'GG^{-1}]$ 
holding for two graded derivations $D$, $D'$ and a Lie group valued function $G$. 
The verification is straightforward. 
\end{proof}

$2$--gauge transformations act on $1$--gauge transformations (cf. subsect. \cref{subsec:cmgautr}) and 
do so in a proper way depending on the reference $2$--connection and compatibly with the $\mathsans{K}$ 
action on all these items.

\begin{prop}
If $g$, $J$, $h$, $K$ and $C$, $E$ are the components of a $1$-- and a $2$--
gauge transformation, respectively, then 
\begin{align}
&g'=\tau(E)g, 
\vphantom{\Big]}
\label{prcmgauforgau1}
\\
&J'=\Ad E(J)+\dot{}\mu(\omega-\Ad g(\omega)-h,E)+C,
\vphantom{\Big]}
\label{prcmgauforgau2}
\\
&h'=\Ad\tau(E)(h)+\dot\tau(\,\dot{}\mu(\Ad g(\omega)+h,E)),
\vphantom{\Big]}
\label{prcmgauforgau3}
\\
&K'=\Ad E(K)+\dot{}\mu\dot{}\,(\Ad(\tau(E)g)(\omega), \dot{}\mu(\omega-\Ad g(\omega)-h,E)+C)
\vphantom{\Big]}
\label{prcmgauforgau4}
\\
&\hspace{3cm}+\dot{}\mu(\Ad g(\theta)+\dot\tau(\mu\dot{}\,(g,\varOmega)-\dot{}\mu\dot{}\,(\Ad g(\omega),J)+K),E)
\vphantom{\Big]}
\nonumber
\end{align} 
are the components of a $1$--gauge transformation. 
\end{prop}

\begin{proof}
A straightforward algebraic calculation shows that the action
\ceqref{cmconn1}--\ceqref{cmconn12}, \ceqref{cmgautr5}--\ceqref{cmgautr16} 
and \ceqref{cmgauforgau7}-\ceqref{cmgauforgau12} of the operation derivations 
on the components $\omega$, $\varOmega$, $\theta$, $\varTheta$, $g$, $J$, $h$, $K$ and $E$, $C$ ensures that 
the action of those derivations on the transformed components 
$g'$, $J'$, $h'$, $K'$ satisfies \ceqref{cmgautr5}--\ceqref{cmgautr16} as well 
and that consequently $g'$, $J'$, $h'$, $K'$ 
are also the components of a $1$--gauge transformation as claimed. 
\end{proof}

\begin{defi}\label{defi:2gautransf}
The $2$--gauge transform of a 
$1$--gauge transformation of components $g$, $J$, $h$, $K$ by a $2$--gauge transformation 
of components $E$, $C$ is given by 
\begin{align}
&{}^{E,C}g=\tau(E)g, 
\vphantom{\Big]}
\label{cmgauforgau1}
\\
&{}^{E,C}J=\Ad E(J)+\dot{}\mu(\omega-\Ad g(\omega)-h,E)+C,
\vphantom{\Big]}
\label{cmgauforgau2}
\\
&{}^{E,C}h=\Ad\tau(E)(h)+\dot\tau(\,\dot{}\mu(\Ad g(\omega)+h,E)),
\vphantom{\Big]}
\label{cmgauforgau3}
\\
&{}^{E,C}K=\Ad E(K)+\dot{}\mu\dot{}\,(\Ad(\tau(E)g)(\omega), \dot{}\mu(\omega-\Ad g(\omega)-h,E)+C)
\vphantom{\Big]}
\label{cmgauforgau4}
\\
&\hspace{3cm}+\dot{}\mu(\Ad g(\theta)+\dot\tau(\mu\dot{}\,(g,\varOmega)-\dot{}\mu\dot{}\,(\Ad g(\omega),J)+K),E).
\vphantom{\Big]}
\nonumber
\end{align} 
\end{defi}

\noindent
Inserting above relations \ceqref{cmgautr5}, \ceqref{cmgautr6} expressing $h$, $K$ in terms of $g$, $J$, 
and relation \ceqref{cmgauforgau7} expressing $C$ in terms of $E$, 
these expressions are formally identical to the standard ones of strict higher gauge theory. 

Since for an assigned reference $2$--connection the variation component of a $2$--gauge transformation can 
be expressed in terms of the modification component by \ceqref{cmgauforgau7}, a $2$--gauge transformation is effectively specified 
by this latter. $2$--gauge transformations can hence be viewed as elements of the group 
$\iMap(T[1]P,\mathsans{E})$ of $\mathsans{E}$--valued internal functions. They form indeed a 
distinguished subgroup of this latter, the $2$--gauge group, as \ceqref{cmgauforgau9},
\ceqref{cmgauforgau11} are preserved under the group operations of $\iMap(T[1]P,\mathsans{E})$ 
(see eq. B.0.8 of I). 
$2$--gauge transformation is a left action of the $2$--gauge group on $1$--gauge transformations.

$2$--gauge transformation action has the further relevant property.

\begin{prop}
For a $1$-- and a $2$--gauge transformation of components $g$, $J$, $h$, $K$ and $E$, $C$,
respectively, one has 
\begin{align}
&{}^{{{}^{E,C}g,{}^{E,C}J,{}^{E,C}h,{}^{E,C}K}}\omega={}^{g,J,h,K}\omega,
\vphantom{\Big]}
\label{cmgauforgau5}
\\
&{}^{{{}^{E,C}g,{}^{E,C}J,{}^{E,C}h,{}^{E,C}K}}\varOmega={}^{g,J,h,K}\varOmega+\dot{}\mu(\Ad g(\theta),E),
\vphantom{\Big]}
\label{cmgauforgau6}
\\
&{}^{{{}^{E,C}g,{}^{E,C}J,{}^{E,C}h,{}^{E,C}K}}\theta={}^{g,J,h,K}\theta-\dot\tau(\dot{}\mu(\Ad g(\theta),E)),
\vphantom{\Big]}
\label{cmgauforgau5crv}
\\
&{}^{{{}^{E,C}g,{}^{E,C}J,{}^{E,C}h,{}^{E,C}K}}\varTheta
\vphantom{\Big]}
\label{cmgauforgau6crv}
\\
&\hspace{2cm}={}^{g,J,h,K}\varTheta+d_P\dot{}\mu(\Ad g(\theta),E)
+\dot{}\mu\dot{}\,({}^{g,J,h,K}\omega,\dot{}\mu(\Ad g(\theta),E)).
\vphantom{\Big]}
\nonumber
\end{align}
\end{prop}

\noindent
Above,  the $1$--gauge transformed connection components 
are given by \ceqref{cmgautr1}--\ceqref{cmgautr4}.

\begin{proof}
The proof is a matter of evaluating the right hand sides of relations
\ceqref{cmgautr1}--\ceqref{cmgautr4} with $g$, $J$, $h$, $K$  replaced by 
${}^{E,C}g$, ${}^{E,C}J$, ${}^{E,C}h$, ${}^{E,C}K$. The calculations are straightforward.
\end{proof}

\noindent
Hence, the gauge transformation action of ${g,J,h,K}$ and ${{}^{E,C}g,{}^{E,C}J,{}^{E,C}h,{}^{E,C}K}$ is the 
same on fake flat $2$--connections. With these qualifications, $2$--gauge transformation corresponds to gauge 
for gauge symmetry.


\subsection{\textcolor{blue}{\sffamily Local operational description of a principal 2--bundle}}\label{subsec:maurer}

The local trivializability of the relevant principal $\hat{\matheul{K}}$--$2$--bundle $\hat{\mathcal{P}}$
implies that of the associated synthetic manifold $P$ (cf. subsect. 3.2 of I). 
On any sufficiently small neighborhood $U$ of the base $M$, \pagebreak there exists so a projection preserving
$\mathsans{K}$--equivariant map $\varPhi_U\in\Map(\pi^{-1}(U),U\times\mathsans{K})$
(cf. def. 3.12 of I). 
$\varPhi_U$ provides a set of coordinates of $\pi^{-1}(U)$ 
modelled on $U\times\mathsans{K}$. These 
are in many ways analogous to the standard adapted coordinates of an ordinary principal bundle. 
One must keep in mind however that they are not anything like genuine coordinates, because they 
arise from a local trivializing functor $\hat{\varPhi}_U$ of $\hat{\mathcal{P}}$ that is only weakly invertible 
(cf.  subsect. 3.1 of I). 

$2$--connections and $1$-- and $2$-- gauge transformations are Lie valued internal 
functions on $T[1]P$ rather than ordinary functions on $P$ (cf. subsects. \cref{subsec:cmconn}--\cref{subsec:cmgauforgau}). 
For this reason, their local description on $U$ presumably requires a set of coordinates 
modelled on $U\times\DD\mathsans{M}$ which are internal functions on $T[1]\pi^{-1}(U)$ rather than ordinary functions 
on $\pi^{-1}(U)$. (Recalll that $\mathsans{K}=\DD\mathsans{M}$ by 3.8.1 of I.)
The coordinates furnished by the trivializing map $\varPhi_U$  
are thus not general enough to serve for our purposes. A more general and weaker notion 
of coordinates is necessary here. 

By the general philosophy of our operational framework, the natural setup for studying the desired kind 
of internal adapted coordinates is the operation $\iOOO S_{\pi^{-1}(U)}=(\iFun(T[1]\pi^{-1}(U)),\mathfrak{m})$, since 
the synthetic morphism manifold of the $\hat{\matheul{K}}$--$2$--bundle $\hat{\mathcal{P}}|_U$
is precisely $\pi^{-1}(U)$. 

A full set of internal coordinates of $\pi^{-1}(U)$ modelled on 
$U\times\DD\mathsans{M}$ comprises two subsets of coordinates modelled on $U$ and $\DD\mathsans{M}$ 
respectively. These require separate consideration. 

By virtue of prop. 3.4 of I, the internal coordinates of $\pi^{-1}(U)$ modelled on $U$
are yielded by the synthetic projection $\pi$ (cf. def. 3.10 of I). They are so 
ordinary functions. In the operational setup, they can be characterized as follows. 

\begin{prop}
A set of internal coordinates of $\pi^{-1}(U)$ modelled on $U$ is 
described by vector--valued ordinary functions $u\in\Map(T[1]\pi^{-1}(U),\mathbb{R}^{\dim M})$ and  
$v\in\Map(T[1]\pi^{-1}(U),\mathbb{R}^{\dim M}[1])$ on which the operation derivations act as 
\begin{align}
&d_{\pi^{-1}(U)}u=v, \qquad d_{\pi^{-1}(U)}v=0   \vphantom{\dot{\dot{\dot{\dot{a}}}}}
\vphantom{\Big]}
\label{maurer-1}
\end{align}
with trivial action of all derivations $j_{\pi^{-1}(U)Z}$, $l_{\pi^{-1}(U)Z}$
for all $Z\in\DD\mathfrak{m}$. 
\end{prop}

\noindent
Above, $u$, $v$ are treated as special cases of internal functions and as such are acted upon by 
the derivations of the operation. $\vphantom{\ul{\ul{\ul{g}}}}$

\begin{proof}
Upon composing the factor $\pi$ with a set of ordinary coordinates of $U$, we obtain an ordinary 
function $u\in\Map(\pi^{-1}(U),\mathbb{R}^{\dim M})\subset\Map(T[1]\pi^{-1}(U),\mathbb{R}^{\dim M})$.  
The scalar nature and $\DD\mathsans{M}$--invariance of $\pi$ (cf. prop. 3.3 of I)
entail that $u$ is annihilated by 
all the derivations $j_{\pi^{-1}(U)Z}$, $l_{\pi^{-1}(U)Z}$. With $u$ there is associated a
further ordinary function 
$v\in\Map(T[1]\pi^{-1}(U),\mathbb{R}^{\dim M}[1])$ defined through \ceqref{maurer-1}. 
The action of $d_{\pi^{-1}(U)}$ and the $j_{\pi^{-1}(U)Z}$, $l_{\pi^{-1}(U)Z}$ on $v$ follows from that on $u$ 
and the operation relations 2.1.1--2.1.6 of I.  
\end{proof}

The internal coordinates of $\pi^{-1}(U)$ modelled on $\DD\mathsans{M}$ are a novelty
requiring a precise definition. The one provided here is generic and may require some tuning at a later 
stage, but it is enough for our purposes for the time being.

\begin{defi}\label{prop:adptcrd1}
A set of internal coordinates of $\pi^{-1}(U)$ modelled on $\DD\mathsans{M}$
is constituted by a Lie group valued 
internal function $\varLambda\in\iMap(T[1]\pi^{-1}(U),\DD\mathsans{M})$ and a Lie algebra valued 
internal function $\varDelta\in\iMap(T[1]\pi^{-1}(U),\DD\mathfrak{m}[1])$ 
acted upon by the operation derivations as 
\begin{align}
&\varLambda^{-1}d_{\pi^{-1}(U)}\varLambda=-\varLambda^{-1}d_{\dot\tau}\varLambda+\varDelta,
\vphantom{\Big]}
\label{bmaurer1}
\\
&d_{\pi^{-1}(U)}\varDelta=-\frac{1}{2}[\varDelta,\varDelta]-d_{\dot\tau}\varDelta, 
\vphantom{\Big]}
\label{bmaurer2}
\\
&\varLambda^{-1}j_{\pi^{-1}(U)Z}\varLambda=0,
\vphantom{\Big]}
\label{bmaurer3}
\\
&j_{\pi^{-1}(U)Z}\varDelta=Z, 
\vphantom{\Big]} 
\label{bmaurer4}
\\
&\varLambda^{-1}l_{\pi^{-1}(U)Z}\varLambda=Z,
\vphantom{\Big]}
\label{bmaurer5}
\\
&l_{\pi^{-1}(U)Z}\varDelta=-[Z,\varDelta]+d_{\dot\tau}Z
\vphantom{\Big]}
\label{bmaurer6}
\end{align}
with $Z\in\DD\mathfrak{m}$. 
It is further required that $\varLambda|_{\pi^{-1}(U)}\in\Map(\pi^{-1}(U),\DD\mathsans{M})$ and that 
$\varLambda|_{\pi^{-1}(U)}=\mathrm{pr}_{\mathsans{K}}\circ\varPhi_U$,
where $\pi^{-1}(U)$ is embedded in $T[1]\pi^{-1}(U)$ as its zero section. 
\end{defi}


\noindent
The notational remarks stated below eqs. \ceqref{cmgautrx1}--\ceqref{cmgautrx6} 
apply here as well with obvious changes and will not be repeated. 
As in similar cases considered earlier, upon considering $d_{\pi^{-1}(U)}+d_{\dot\tau}$ as relevant differential
\ceqref{bmaurer1}--\ceqref{bmaurer6} are formally analogous to relations \ceqref{00maurer2}--\ceqref{00maurer12}
holding for the adapted coordinates of ordinary principal bundles.  
Eq. \ceqref{bmaurer1} defines the coordinate $\varDelta$ in terms of its partner $\varLambda$. \ceqref{bmaurer2}
is the associated Maurer--Cartan--like equation. 
Upon com\-paring  eqs. \ceqref{bmaurer2}, \ceqref{bmaurer4}, \ceqref{bmaurer6}
with \ceqref{cmconnx1}, \ceqref{cmconnx3}, \ceqref{cmconnx5}, it emerges also that 
$\varDelta$ is the connection component of a flat $2$--connection 
of the principal $2$--bundle $\hat{\mathcal{P}}|_U$ (cf. subsect. \cref{subsec:cmconn}). 

\begin{lemma} 
The operation commutation relations 2.1.1--2.1.6 of I are respected by 
\ceqref{bmaurer1}--\ceqref{bmaurer6}.
\end{lemma}

\begin{proof}
This is shown by checking that the six derivation commutators 
in the left hand sides of eqs. 2.1.1--2.1.6 of I 
act as the corresponding derivations in the right hand sides 
when they are applied to the functions $\varLambda$, $\varDelta$ and the 
\ceqref{bmaurer1}--\ceqref{bmaurer6} are used.
\end{proof}

\noindent
The requirement on $\varLambda|_{\pi^{-1}(U)}$ is added in order to render the definition 
of coordinates modelled on $\DD\mathsans{M}$ provided above 
geometrically meaningful, though it plays no direct role in the basic theory of subsects. 
\cref{subsec:basic}, \cref{subsec:nadifcoh}. Note also that the condition that $\varLambda|_{\pi^{-1}(U)}$ 
be an ordinary rather than internal function 
is {\it not} preserved by the derivations  $l_{\pi^{-1}(U)Z}$ by \ceqref{bmaurer5}.
This is expected on general grounds, since by the graded nature of $\DD\mathfrak{m}$
$l_{\pi^{-1}(U)Z}$ turns ordinary functions on $\pi^{-1}(U)$ into internal ones. 


Employing  3.5.1, 3.5.12 of I, we can expand the fiber coordinates $\varLambda$, $\varDelta$ as 
\begin{align}
&\varLambda(\alpha)=\ee^{\alpha\varGamma}\gamma, 
\vphantom{\Big]}
\label{bmaurer7}
\\
&\varDelta(\alpha)=\sigma-\alpha\varSigma, \quad \alpha\in\mathbb{R}[1].
\vphantom{\Big]}
\label{bmaurer8}
\end{align}
In the above relations,  $\gamma\in\iMap(T[1]\pi^{-1}(U),\mathsans{G})$, 
$\varGamma\in\iMap(T[1]\pi^{-1}(U),\mathfrak{e}[1])$  and 
$\sigma\in\iMap(T[1]\pi^{-1}(U),\mathfrak{g}[1])$, 
$\varSigma\in\iMap(T[1]\pi^{-1}(U),\mathfrak{e}[2])$ are the projected internal 
coordinates modelled on $\DD\mathsans{M}=\mathfrak{e}[1]\rtimes_{\mu\dot{}}\mathsans{G}$. We also write
$Z\in\DD\mathfrak{m}$ as $Z(\bar\alpha)=x+\bar\alpha X$, $\bar\alpha\in\mathbb{R}[-1]$, with 
$x\in\mathfrak{g}$ and $X\in\mathfrak{e}[1]$ as in 3.4.6 of I. 

\begin{prop} \label{prop:1localroj}
Expressed in terms of projected internal adapted coordinates, 
the operation relations \ceqref{bmaurer1}--\ceqref{bmaurer6} take the form \pagebreak 
\begin{align}
&\gamma^{-1}d_{\pi^{-1}(U)}\gamma=\sigma-\dot\tau(\mu\dot{}\,(\gamma^{-1},\varGamma)),
\vphantom{\Big]}
\label{maurer2}
\\
&\mu\dot{}\,(\gamma^{-1},d_{\pi^{-1}(U)}\varGamma)=\varSigma-\frac{1}{2}\mu\dot{}\,(\gamma^{-1},[\varGamma,\varGamma]),
\vphantom{\Big]}
\label{maurer3}
\\
&d_{\pi^{-1}(U)}\sigma=-\frac{1}{2}[\sigma,\sigma]+\dot\tau(\varSigma),
\vphantom{\Big]}
\label{maurer4}
\\
&d_{\pi^{-1}(U)}\varSigma=-\dot{}\mu\dot{}\,(\sigma,\varSigma),
\vphantom{\Big]}
\label{maurer5}
\\
&\gamma^{-1}j_{\pi^{-1}(U)Z}\gamma=0,
\vphantom{\Big]}
\label{maurer6}
\\
&\mu\dot{}\,(\gamma^{-1},j_{\pi^{-1}(U)Z}\varGamma)=0,
\vphantom{\Big]}
\label{maurer7}
\\
&j_{\pi^{-1}(U)Z}\sigma=x,
\vphantom{\Big]}
\label{maurer8}
\\
&j_{\pi^{-1}(U)Z}\varSigma=X,
\vphantom{\Big]}
\label{maurer9}
\\
&\gamma^{-1}l_{\pi^{-1}(U)Z}\gamma=x,
\vphantom{\Big]}
\label{maurer10}
\\
&\mu\dot{}\,(\gamma^{-1},l_{\pi^{-1}(U)Z}\varGamma)=X,
\vphantom{\Big]}
\label{maurer11}
\\
&l_{\pi^{-1}(U)Z}\sigma=-[x,\sigma]+\dot\tau(X),
\vphantom{\Big]}
\label{maurer12}
\\
&l_{\pi^{-1}(U)Z}\varSigma=-\dot{}\mu\dot{}\,(x,\varSigma)+\dot{}\mu\dot{}\,(\sigma,X).
\vphantom{\Big]}
\label{maurer13}
\end{align} 
Moreover, $\gamma|_{\pi^{-1}(U)}\in\Map(\pi^{-1}(U),\mathsans{G})$, $\varGamma|_{\pi^{-1}(U)}\in\Map(\pi^{-1}(U),\mathfrak{e}[1])$  
and  $\gamma|_{\pi^{-1}(U)}$, $\varGamma|_{\pi^{-1}(U)}$, in the combination \ceqref{bmaurer7}, yield 
$\mathrm{pr}_{\mathsans{K}}\circ\varPhi_U$. 
\end{prop}


\begin{proof}
The proof is a matter of a straightforward albeit lengthy calculation. 
We substitute the expressions of $\varLambda$, $\varDelta$
in terms of $\gamma$ $\varGamma$, $\sigma$, $\varSigma$ 
of eqs. \ceqref{bmaurer7}, \ceqref{bmaurer8}
and that of $Z$ in terms of $x$, $X$ into \ceqref{bmaurer1}-- \ceqref{bmaurer6}
and use relations 3.5.13, 3.5.15
and 3.5.18 as well as expressions 3.5.23 and 3.5.25 of I. 
\end{proof}   

Under the restriction operation morphism $\iOOO L_U:\iOOO S_{\pi^{-1}(U)}\rightarrow\iOOO S_{\pi^{-1}(U)0}$ 
of subsect. 3.8 of I, the projected internal coordinates 
$\gamma$, $\varGamma$, $\sigma$, $\varSigma$ of $\pi^{-1}(U)$ get 
\begin{align}
&\gamma_0=I_U{}^*\gamma,
\vphantom{\Big]}
\label{bmaurer11}
\\
&\varGamma_0=I_U{}^*\varGamma,
\vphantom{\Big]}
\label{bmaurer12}
\\
&\sigma_0=I_U{}^*\sigma,
\vphantom{\Big]}
\label{bmaurer13}
\\
&\varSigma_0=I_U{}^*\varSigma,
\vphantom{\Big]}
\label{bmaurer14}
\end{align} 
\vskip-7mm \eject\noindent
where $I_U{}^*:\iFun(T[1]\pi^{-1}(U))\rightarrow\iFun(T[1]\pi_0{}^{-1}(U))$ 
is the restriction morphism
associated with the inclusion map $I_U{}:\pi_0{}^{-1}(U)\rightarrow \pi^{-1}(U)$. 

\begin{defi}\label{defi:specadpcoo}
The internal coordinates $\pi^{-1}(U)$ modelled on $\DD\mathsans{M}$ are special if 
\begin{align}
&\varGamma_0=0, 
\vphantom{\Big]}
\label{bmaurer15}
\\
&\varSigma_0=0. 
\vphantom{\Big]}
\label{bmaurer16}
\end{align} 
$\gamma_0$, $\sigma_0$ are then a subset of internal coordinates 
of $\pi_0{}^{-1}(U)$ adapted to $\DD\mathsans{M}_0$. \vspace{3mm}
\end{defi}

\begin{prop} 
For special adapted coordinates, 
the action of the operation derivations on 
$\gamma_0$, $\sigma_0$ is given by the right hand side of eqs. 
\ceqref{maurer2}--\ceqref{maurer13} with 
$\gamma$,  $\sigma$ replaced by 
$\gamma_0$, $\sigma_0$ and $\varGamma$, $\varSigma$ and  $X$ set to $0$. 
\end{prop}

\begin{proof}
The action of the derivations of the operation $\iOOO S_{\pi^{-1}(U)0}$ on 
$\gamma_0$, $\varGamma_0$, $\sigma_0$, $\varSigma_0$ is given by the right hand side of eqs. 
\ceqref{maurer2}--\ceqref{maurer13} with $\gamma$, $\varGamma$, $\sigma$, $\varSigma$ 
replaced by $\gamma_0$, $\varGamma_0$, $\sigma_0$, $\varSigma_0$  and $X$ set to $0$. 
Taking \ceqref{bmaurer15}, \ceqref{bmaurer16} into account, the action on 
$\gamma_0$, $\sigma_0$ has the properties stated. 
\end{proof}


\subsection{\textcolor{blue}{\sffamily Basic formulation of  principal 2--bundle theory}}\label{subsec:basic} 

As recalled in subsect. \cref{subsec:maurer}, for the $\hat{\matheul{K}}$--$2$--bundle $\hat{\mathcal{P}}$,
on any sufficiently small neighborhood $U\subset M$ there exists a projection preserving 
$\mathsans{K}$--equivariant trivializing map $\varPhi_U\in\Map(\pi^{-1}(U),U\times\mathsans{K})$.
We saw further that it is possible to attach to $\varPhi_U$ a special set of internal coordinates 
of $\pi^{-1}(U)$ modelled on $U\times\DD\mathsans{M}$ 
the adapted coordinates $u$, $v$ and $\varLambda$, $\varDelta$, or 
$\gamma$, $\varGamma$, $\sigma$, $\varSigma$ in projected form,
for the factors $U$ and $\DD\mathsans{M}$, respectively. These are internal functions on 
$T[1]\pi^{-1}(U)$ with special properties in the operation 
$\iOOO S_{\pi^{-1}(U)}$ of the morphism space $S_{\pi^{-1}(U)}$. 

In this subsection, we shall use these coordinates  
to analyze $2$--connections and $1$-- and $2$--gauge transformations of $\hat{\mathcal{P}}$ 
in terms of basic Lie valued function data on $T[1]\pi^{-1}(U)$.
Remember that a function $F_{\mathrm{b}}\in\iFun(T[1]\pi^{-1}(U))$ is basic if it is annihilated by all
derivations $j_{\pi^{-1}(U)Z}$, $l_{\pi^{-1}(U)Z}$ 
(cf. subsect. 2.1 of I). 


Before proceeding further, we note that the inclusion map $N_U:\pi^{-1}(U)\rightarrow P$ 
yield a morphisms $Q_{\,U}:S_{\pi^{-1}(\,U)}\rightarrow S_{P}$ 
of the morphism spaces of $\pi^{-1}(\,U)$ and $P$ and through this a morphism
$\iOOO Q_{\,U}:\iOOO S_{P}\rightarrow\iOOO S_{\pi^{-1}(\,U)}$
of the associated operations (cf. subsect. 3.7 of I). Therefore, if a function $F\in\iFun(T[1]P)$ 
obeys certain relations under the actions of the derivations 
$j_{PZ}$, $l_{PZ}$ of $\iOOO S_{P}$, its restriction $F|_{T[1]\pi^{-1}(U)}=N_U{}^*F\in\iFun(T[1]\pi^{-1}(U))$
obeys formally identical relations under the actions of the derivations 
$j_{\pi^{-1}(U)Z}$, $l_{\pi^{-1}(U)Z}$ of $\iOOO S_{\pi^{-1}(\,U)}$. 

Consider a $2$--connection of $\hat{\mathcal{P}}$ with connection and curvature components $A$, $B$ 
(cf. subsect. \cref{subsec:cmconn}, def. \cref{defi:2conn}). 

\begin{defi}
The basic connection and curvature components of the $2$--connec\-tion 
are the Lie algebra valued internal functions 
$A_{\mathrm{b}}\in\iMap(T[1]\pi^{-1}(U)$, $\DD\mathfrak{m}[1])$ 
and $B_{\mathrm{b}}\in\iMap(T[1]\pi^{-1}(U),\DD\mathfrak{m}[2])$ 
defined by
\begin{align}
A_{\mathrm{b}}&=\Ad\varLambda(A-\varDelta),
\vphantom{\Big]}
\label{sfbasic1}
\\
B_{\mathrm{b}}&=\Ad\varLambda(B).
\vphantom{\Big]}
\label{sfbasic2}
\end{align}
\end{defi} 

\noindent
Above, restriction of $A$, $B$ to $T[1]\pi^{-1}(U)$ is tacitly understood in order not to clutter the notation.
The names given to $A_{\mathrm{b}}$, $B_{\mathrm{b}}$ are justified by the following proposition.

\begin{prop}\label{prop:basic1}
$A_{\mathrm{b}}$, $B_{\mathrm{b}}$ are basic elements of the operation $\iOOO S_{\pi^{-1}(\,U)}$.
\end{prop}

\begin{proof}
One has to show that $A_{\mathrm{b}}$, $B_{\mathrm{b}}$ are annihilated by all derivations 
$j_{\pi^{-1}(U)Z}$ and $l_{\pi^{-1}(U)Z}$ with $Z\in\DD\mathfrak{m}$.
This can be verified using relations  \ceqref{cmconnx3}--\ceqref{cmconnx6} and \ceqref{bmaurer3}--\ceqref{bmaurer6}.
\end{proof}


\begin{prop}\label{prop:basic2}
$A_{\mathrm{b}}$, $B_{\mathrm{b}}$ obey  the relations 
\begin{align}
&d_{\pi^{-1}(U)}A_{\mathrm{b}}=-\frac{1}{2}[A_{\mathrm{b}},A_{\mathrm{b}}]-d_{\dot\tau}A_{\mathrm{b}}+B_{\mathrm{b}},
\vphantom{\Big]}
\label{sfbasic3}
\\
&d_{\pi^{-1}(U)}B_{\mathrm{b}}=-[A_{\mathrm{b}},B_{\mathrm{b}}]-d_{\dot\tau}B_{\mathrm{b}}.
\vphantom{\Big]}
\label{sfbasic4}
\end{align}
\end{prop}

\noindent
These are formally analogous to \ceqref{cmconnx1}, \ceqref{cmconnx2}. 

\begin{proof}
Relations \ceqref{sfbasic3}, \ceqref{sfbasic4} to be proven follow 
from \ceqref{cmconnx1}, \ceqref{cmconnx2} and \ceqref{bmaurer1}, \ceqref{bmaurer2}
through a simple calculation.
\end{proof}

Just as the \pagebreak connection and curvature components $A$, $B$ can be expressed in terms of the projected 
connection and curvature components $\omega$, $\varOmega$, $\theta$, $\varTheta$ according to 
\ceqref{cmconnx7}, \ceqref{cmconnx8}, so the basic components 
$A_{\mathrm{b}}$, $B_{\mathrm{b}}$ can be expressed in terms of basic projected  
components $\omega_{\mathrm{b}}$, $\varOmega_{\mathrm{b}}$, $\theta_{\mathrm{b}}$, 
$\varTheta_{\mathrm{b}}$ as 
\begin{align}
&A_{\mathrm{b}}(\alpha)=\omega_{\mathrm{b}}-\alpha\varOmega_{\mathrm{b}},
\vphantom{\Big]}
\label{sfbasic11}
\\
&B_{\mathrm{b}}(\alpha)=\theta_{\mathrm{b}}+\alpha\varTheta_{\mathrm{b}}, \quad \alpha\in\mathbb{R}[1].
\vphantom{\Big]}
\label{sfbasic12}
\end{align}
In the above  relations, 
$\omega_{\mathrm{b}}\in\iMap(T[1]\pi^{-1}(U),\mathfrak{g}[1])$, $\varOmega_{\mathrm{b}}\in\iMap(T[1]\pi^{-1}(U),\mathfrak{e}[2])$, 
$\theta_{\mathrm{b}}\in\iMap(T[1]\pi^{-1}(U),\mathfrak{g}[2])$, $\varTheta_{\mathrm{b}}\in\iMap(T[1]\pi^{-1}(U),\mathfrak{e}[3])$. 

\begin{prop}\label{prop:basic3}
$\omega_{\mathrm{b}}$, 
$\varOmega_{\mathrm{b}}$, $\theta_{\mathrm{b}}$, $\varTheta_{\mathrm{b}}$  
are related to $\omega$, $\varOmega$, $\theta$, $\varTheta$ by 
\begin{align}
&\omega_{\mathrm{b}}
=\Ad\gamma(\omega-\sigma),  
\vphantom{\Big]}
\label{basic1}
\\
&\varOmega_{\mathrm{b}}
=\mu\dot{}\,(\gamma,\varOmega-\varSigma)  
-\dot{}\mu\dot{}\,(\Ad \gamma(\omega-\sigma),\varGamma),
\vphantom{\Big]}
\label{basic2}
\\
&\theta_{\mathrm{b}}
=\Ad\gamma(\theta),
\vphantom{\Big]}
\label{basic3}
\\
&\varTheta_{\mathrm{b}}
=\mu\dot{}\,(\gamma,\varTheta)-\dot{}\mu\dot{}\,(\Ad \gamma(\theta),\varGamma).
\vphantom{\Big]}
\label{basic4}
\end{align}
\end{prop}

\noindent
Restriction of $\omega$, $\varOmega$, $\theta$, $\varTheta$ to $T[1]\pi^{-1}(U)$ is here also tacitly understood. 

\begin{proof}
Inserting \ceqref{cmconnx7}, \ceqref{cmconnx8}, \ceqref{bmaurer7}, \ceqref{bmaurer8} 
and \ceqref{sfbasic11}, \ceqref{sfbasic12} into 
\ceqref{sfbasic1}, \ceqref{sfbasic2} and using 3.5.18 of I, 
one gets \ceqref{basic1}--\ceqref{basic4} by simple calculations.  
\end{proof}

\begin{prop}\label{prop:basic4}
$\omega_{\mathrm{b}}$, $\varOmega_{\mathrm{b}}$, $\theta_{\mathrm{b}}$, $\varTheta_{\mathrm{b}}$ satisfy 
\begin{align} 
&d_{\pi^{-1}(U)}\omega_{\mathrm{b}}=-\frac{1}{2}[\omega_{\mathrm{b}},\omega_{\mathrm{b}}]+\dot\tau(\varOmega_{\mathrm{b}})+\theta_{\mathrm{b}},
\vphantom{\Big]}
\label{basic5}
\\
&d_{\pi^{-1}(U)}\varOmega_{\mathrm{b}}=-\dot{}\mu\dot{}\,(\omega_{\mathrm{b}},\varOmega_{\mathrm{b}})+\varTheta_{\mathrm{b}},
\vphantom{\Big]}
\label{basic6}
\\
&d_{\pi^{-1}(U)}\theta_{\mathrm{b}}=-[\omega_{\mathrm{b}},\theta_{\mathrm{b}}]-\dot\tau(\varTheta_{\mathrm{b}}),
\vphantom{\Big]}
\label{basic7}
\\
&d_{\pi^{-1}(U)}\varTheta_{\mathrm{b}}=-\dot{}\mu\dot{}\,(\omega_{\mathrm{b}},\varTheta_{\mathrm{b}})
+\dot{}\mu\dot{}\,(\theta_{\mathrm{b}},\varOmega_{\mathrm{b}}).
\vphantom{\Big]}
\label{basic8}
\end{align}
\end{prop}

\noindent 
These relations are formally identical to \ceqref{cmconn1}--\ceqref{cmconn4}. 
Our basic formulation has so reproduced 
the familiar local description of $2$--connections of strict higher gauge theory. 
This statement will be qualified more precisely in subsect. \cref{subsec:nadifcoh}. 

\begin{proof}
One demonstrates \ceqref{basic5}--\ceqref{basic8} by substituting \ceqref{sfbasic11}, \ceqref{sfbasic12} into 
\ceqref{sfbasic3}, \ceqref{sfbasic4} and using 3.5.13 and 3.5.15 of I.
\end{proof}

The basic components of the $2$--connection behave as expected 
when the $2$--connection is special.

\begin{prop}
If the $2$--connection and the adapted coordinates are both special 
(cf. defs. \cref{defi:prop2conn}, \cref{defi:specadpcoo}), then one has \hphantom{xxxxxxxxxx}
\begin{align}
&I_U{}^*\varOmega_{\mathrm{b}}=0,
\vphantom{\Big]}
\label{propbas1}
\\
&I_U{}^*\varTheta_{\mathrm{b}}=0,
\vphantom{\Big]}
\label{propbas2}
\end{align}
where $I_U:\pi_0{}^{-1}(U)\rightarrow\pi^{-1}(U)$ is the inclusion map. 
\end{prop}

\begin{proof}
This follows from \ceqref{basic2}, \ceqref{basic4} upon substituting 
\ceqref{0cmconn5ol}, \ceqref{0cmconn6ol} and \ceqref{bmaurer15}, \ceqref{bmaurer16}
\end{proof}

Next, consider a $1$--gauge transformation of $\hat{\mathcal{P}}$ with transformation and shift components $\varPsi$, $\varUpsilon$ 
(cf. subsect. \cref{subsec:cmgautr}, def. \cref{defi:1gauge}). 

\begin{defi}
The basic transformation and shift components of the $1$--gau\-ge transformation are  
the Lie group and algebra 
valued internal functions $\varPsi_{\mathrm{b}}\in\iMap(T[1]\pi^{-1}(U)$, $\DD\mathsans{M})$ and 
$\varUpsilon_{\mathrm{b}}\in\iMap(T[1]\pi^{-1}(U),\DD\mathfrak{m}[1])$ defined  by
\begin{align}
\varPsi_{\mathrm{b}}&=\varLambda \varPsi\varLambda^{-1},
\vphantom{\Big]}
\label{sfbasic5}
\\
\varUpsilon_{\mathrm{b}}&=\Ad\varLambda( \varUpsilon-\varDelta+\Ad \varPsi(\varDelta)).
\vphantom{\Big]}
\label{sfbasic6}
\end{align}
\end{defi}

\noindent
Above, restriction of $\varPsi$, $\varUpsilon$ to $T[1]\pi^{-1}(U)$ is  understood. 
The name given to $\varPsi_{\mathrm{b}}$, $\varUpsilon_{\mathrm{b}}$ are justified by the following proposition.

\begin{prop}\label{prop:basic5}
$\varPsi_{\mathrm{b}}$, $\varUpsilon_{\mathrm{b}}$ are basic elements
of the operation $\iOOO S_{\pi^{-1}(\,U)}$.
\end{prop}

\begin{proof}
One has to show that $\varPsi_{\mathrm{b}}$, $\varUpsilon_{\mathrm{b}}$ are annihilated by all derivations 
$j_{\pi^{-1}(U)Z}$ and $l_{\pi^{-1}(U)Z}$ with $Z\in\DD\mathfrak{m}$.
This can be verified using relations  \ceqref{cmgautrx3}--\ceqref{cmgautrx6} and \ceqref{bmaurer3}--\ceqref{bmaurer6}.
\end{proof}


\begin{prop}\label{prop:basic6}
$\varPsi_{\mathrm{b}}$, $\varUpsilon_{\mathrm{b}}$ satisfy the relations 
\begin{align}
&d_{\pi^{-1}(U)}\varPsi_{\mathrm{b}}\varPsi_{\mathrm{b}}{}^{-1}
=-d_{\dot\tau}\varPsi_{\mathrm{b}}\varPsi_{\mathrm{b}}{}^{-1}-\varUpsilon_{\mathrm{b}},
\vphantom{\Big]}
\label{sfbasic7}
\\
&d_{\pi^{-1}(U)}\varUpsilon_{\mathrm{b}}=-\frac{1}{2}[\varUpsilon_{\mathrm{b}},\varUpsilon_{\mathrm{b}}]
-d_{\dot\tau}\varUpsilon_{\mathrm{b}}.
\vphantom{\Big]}
\label{sfbasic8}
\end{align}
\end{prop}

\noindent
These are analogous in form to \ceqref{cmgautrx1}, \ceqref{cmgautrx2}. 

\begin{proof}
Relations \ceqref{sfbasic7}, \ceqref{sfbasic8} to be shown 
follow from eqs. \ceqref{cmgautrx1}, \ceqref{cmgautrx2} 
and \ceqref{bmaurer1}, \ceqref{bmaurer2} through a simple calculation. 
\end{proof}

Next, consider 
the $1$--gauge transform ${}^{\varPsi,\varUpsilon}A$, ${}^{\varPsi,\varUpsilon}B$ of a $2$--connection $A$, $B$
(cf. subsect. \cref{subsec:cmgautr}, def. \cref{defi:1gaugetransf}).

\begin{prop} \label{prop:basic7}
The basic components ${}^{\varPsi,\varUpsilon}A_{\mathrm{b}}$,
${}^{\varPsi,\varUpsilon}B_{\mathrm{b}}$ 
of the $1$--gauge transformed transformed $2$--connection 
are given in terms of $A_{\mathrm{b}}$, $B_{\mathrm{b}}$ and $\varPsi_{\mathrm{b}}$, $\varUpsilon_{\mathrm{b}}$ 
by 
\begin{align}
&{}^{\varPsi,\varUpsilon}A_{\mathrm{b}}=\Ad\varPsi_{\mathrm{b}}(A_{\mathrm{b}})+\varUpsilon_{\mathrm{b}},
\vphantom{\Big]}
\label{sfbasic9}
\\
&{}^{\varPsi,\varUpsilon}B_{\mathrm{b}}=\Ad\varPsi_{\mathrm{b}}(B_{\mathrm{b}}).
\vphantom{\Big]}
\label{sfbasic10}
\end{align}
\end{prop}

\begin{proof}
Relations \ceqref{sfbasic9}, \ceqref{sfbasic10} can be straightforwardly verified 
combining \ceqref{sfbasic1}, \ceqref{sfbasic2} and \ceqref{sfbasic5}, \ceqref{sfbasic6}
with \ceqref{cmgautrx9}, \ceqref{cmgautrx10}.
\end{proof}

\noindent
Eqs. \ceqref{sfbasic9}, \ceqref{sfbasic10} suggest defining the basic component gauge transforms 
${}^{\varPsi_{\mathrm{b}},\varUpsilon_{\mathrm{b}}}A_{\mathrm{b}}$, 
${}^{\varPsi_{\mathrm{b}},\varUpsilon_{\mathrm{b}}}B_{\mathrm{b}}$ to be given by the right hand sides of  
\ceqref{sfbasic9}, \ceqref{sfbasic10} themsel\-ves. By doing so, 
${}^{\varPsi_{\mathrm{b}},\varUpsilon_{\mathrm{b}}}A_{\mathrm{b}}$, 
${}^{\varPsi_{\mathrm{b}},\varUpsilon_{\mathrm{b}}}B_{\mathrm{b}}$ are given be expressions 
formally analogous to those holding for the ordinary components, viz 
\ceqref{cmgautrx9}, \ceqref{cmgautrx10}. 

Again, in the same way as the transformation and shift components $\varPsi$, $\varUpsilon$ 
can be expanded in their projected transformation and shift components $g$, $J$, $h$, $K$
according to \ceqref{cmgautrx7}, \ceqref{cmgautrx8}, so their basic counterparts
$\varPsi_{\mathrm{b}}$, $\varUpsilon_{\mathrm{b}}$ can be expanded in basic projected 
components $g_{\mathrm{b}}$, $J_{\mathrm{b}}$, $h_{\mathrm{b}}$, $K_{\mathrm{b}}$ as 
\begin{align}
&\varPsi_{\mathrm{b}}(\alpha)=\ee^{\alpha J_{\mathrm{b}}}g_{\mathrm{b}},
\vphantom{\Big]}
\label{sfbasic13}
\\
&\varUpsilon_{\mathrm{b}}(\alpha)=h_{\mathrm{b}}-\alpha K_{\mathrm{b}}, \quad \alpha\in\mathbb{R}[1].
\vphantom{\Big]}
\label{sfbasic14}
\end{align}
In the above relations, $g_{\mathrm{b}}\in\iMap(T[1]\pi^{-1}(U),\mathsans{G})$, $J_{\mathrm{b}}\in\iMap(T[1]\pi^{-1}(U),\mathfrak{e}[1])$,
$h_{\mathrm{b}}\in\iMap(T[1]\pi^{-1}(U),\mathfrak{g}[1])$, $K_{\mathrm{b}}\in\iMap(T[1]\pi^{-1}(U),\mathfrak{e}[2])$. 

\begin{prop}\label{prop:basic8}
$g_{\mathrm{b}}$, $J_{\mathrm{b}}$, $h_{\mathrm{b}}$, $K_{\mathrm{b}}$ are related to $g$, $J$, $h$, $K$ as
\begin{align}
&g_{\mathrm{b}}=\gamma g\gamma^{-1},  \hspace{6.5cm}
\vphantom{\Big]}
\label{basic9}
\end{align} 
\begin{align}
&J_{\mathrm{b}}=\mu\dot{}\,(\gamma, J)+\varGamma-\mu\dot{}\,(\gamma g\gamma^{-1},\varGamma),
\vphantom{\Big]}
\label{basic10}
\\
&h_{\mathrm{b}}=\Ad\gamma(h-\sigma+\Ad g(\sigma)),
\vphantom{\Big]}
\label{basic11}
\\
&K_{\mathrm{b}}=\mu\dot{}\,(\gamma, K-\varSigma+\mu\dot{}\,(g,\varSigma)
-\dot{}\mu\dot{}\,(\Ad g(\sigma),J)
\vphantom{\Big]}
\label{basic12}
\\
&\hspace{4cm}-\dot{}\mu\dot{}\,(h-\sigma+\Ad g(\sigma),\mu\dot{}\,(\gamma^{-1},\varGamma))).
\vphantom{\Big]}
\nonumber
\end{align}
\end{prop}

\noindent
Restriction of $g$, $J$, $h$, $K$ to $T[1]\pi^{-1}(U)$ is here also tacitly understood. 

\begin{proof}
Substituting \ceqref{cmgautrx7}, \ceqref{cmgautrx8}, \ceqref{bmaurer7}, \ceqref{bmaurer8} 
and \ceqref{sfbasic13}, \ceqref{sfbasic14} into \ceqref{sfbasic5}, \ceqref{sfbasic6}
and using 3.5.2, 3.5.3 and 3.5.18 of I, 
one gets \ceqref{basic9}--\ceqref{basic12} by straightforward computations. 
\end{proof}

\begin{prop} \label{prop:basic9}
$g_{\mathrm{b}}$, $J_{\mathrm{b}}$, $h_{\mathrm{b}}$, $K_{\mathrm{b}}$ obey
\begin{align}
&d_{\pi^{-1}(U)}g_{\mathrm{b}}g_{\mathrm{b}}{}^{-1}=-h_{\mathrm{b}}-\dot\tau(J_{\mathrm{b}}),
\vphantom{\Big]}
\label{basic13}
\\
&d_{\pi^{-1}(U)}J_{\mathrm{b}}=-K_{\mathrm{b}}-\frac{1}{2}[J_{\mathrm{b}},J_{\mathrm{b}}]-\dot{}\mu\dot{}\,(h_{\mathrm{b}},J_{\mathrm{b}}),
\vphantom{\Big]}
\label{basic14}
\\
&d_{\pi^{-1}(U)}h_{\mathrm{b}}=-\frac{1}{2}[h_{\mathrm{b}},h_{\mathrm{b}}]+\dot\tau(K_{\mathrm{b}}),
\vphantom{\Big]}
\label{basic15}
\\
&d_{\pi^{-1}(U)}K_{\mathrm{b}}=-\dot{}\mu\dot{}\,(h_{\mathrm{b}},K_{\mathrm{b}}).
\vphantom{\Big]}
\label{basic16}
\end{align}
\end{prop}

\noindent
These relations are of the same form as  \ceqref{cmgautr5}--\ceqref{cmgautr8}. 
We have reobtained in this way in our basic formulation 
the familiar local description of $1$--gauge transformations of strict higher gauge theory. 
More on this in subsect. \cref{subsec:nadifcoh}. 

\begin{proof}
\ceqref{basic13}--\ceqref{basic16} are shown by inserting \ceqref{sfbasic13}, \ceqref{sfbasic14} into 
\ceqref{sfbasic7}, \ceqref{sfbasic8} and using 3.5.13, 3.5.15, 3.5.22 
and 3.5.24 of I. 
\end{proof}


Concerning the $1$--gauge transformed $2$--connection 
${}^{g,J,h,K}\omega$, ${}^{g,J,h,K}\varOmega$, ${}^{g,J,h,K}\theta$, ${}^{g,J,h,K}\varTheta$
we have the following result. 

\begin{prop} \label{prop:basic10}
The basic components ${}^{g,J,h,K}\omega_{\mathrm{b}}$, ${}^{g,J,h,K}\varOmega_{\mathrm{b}}$, 
${}^{g,J,h,K}\theta_{\mathrm{b}}$, ${}^{g,J,h,K}\varTheta_{\mathrm{b}}$ are given in terms of 
$\omega_{\mathrm{b}}$, $\varOmega_{\mathrm{b}}$, $\theta_{\mathrm{b}}$, $\varTheta_{\mathrm{b}}$ 
and $g_{\mathrm{b}}$, $J_{\mathrm{b}}$, $h_{\mathrm{b}}$, $K_{\mathrm{b}}$ by 
\begin{align}
&{}^{g,J,h,K}\omega_{\mathrm{b}}
=\Ad g_{\mathrm{b}}(\omega_{\mathrm{b}})+h_{\mathrm{b}},
\vphantom{\Big]}
\label{basic21}
\\
&{}^{g,J,h,K}\varOmega_{\mathrm{b}}
=\mu\dot{}\,(g_{\mathrm{b}},\varOmega_{\mathrm{b}})-\dot{}\mu\dot{}\,(\Ad g_{\mathrm{b}}(\omega_{\mathrm{b}}),J_{\mathrm{b}})+K_{\mathrm{b}},
\vphantom{\Big]}
\label{basic22}
\end{align}
\begin{align}
&{}^{g,J,h,K}\theta_{\mathrm{b}}
=\Ad g_{\mathrm{b}}(\theta_{\mathrm{b}}),
\vphantom{\Big]}
\label{basic23}
\\
&{}^{g,J,h,K}\varTheta_{\mathrm{b}}
=\mu\dot{}\,(g_{\mathrm{b}},\varTheta_{\mathrm{b}})-\dot{}\mu\dot{}\,(\Ad g_{\mathrm{b}}(\theta_{\mathrm{b}}),J_{\mathrm{b}}).
\vphantom{\Big]}
\label{basic24}
\end{align}
\end{prop}

\begin{proof}
\ceqref{basic21}--\ceqref{basic24} follow from inserting \ceqref{sfbasic11}, \ceqref{sfbasic12}, 
\ceqref{sfbasic13}, \ceqref{sfbasic14} into \ceqref{sfbasic9}, \ceqref{sfbasic10}
and using 3.5.18 of I. 
\end{proof}

\noindent
Following the remarks below eqs. \ceqref{sfbasic9}, \ceqref{sfbasic10}, 
we can regard the right hand sides of eqs. \ceqref{basic21}--\ceqref{basic24}
as the expressions of the basic projected  component
gauge transforms ${}^{g_{\mathrm{b}},J_{\mathrm{b}},h_{\mathrm{b}},K_{\mathrm{b}}}\omega_{\mathrm{b}}$,
${}^{g_{\mathrm{b}},J_{\mathrm{b}},h_{\mathrm{b}},K_{\mathrm{b}}}\varOmega_{\mathrm{b}}$,
${}^{g_{\mathrm{b}},J_{\mathrm{b}},h_{\mathrm{b}},K_{\mathrm{b}}}\theta_{\mathrm{b}}$,
${}^{g_{\mathrm{b}},J_{\mathrm{b}},h_{\mathrm{b}},K_{\mathrm{b}}}\varTheta_{\mathrm{b}}$, respectively.
Again, such expressions are formally identical to those holding for the ordinary projected components,
viz  \ceqref{cmgautr1}--\ceqref{cmgautr4} and so reproduce at the basic level the usual
local description of $2$--connection $1$--gauge transformation of strict higher gauge theory. 
This matter will be reconsidered in subsect. \cref{subsec:nadifcoh}. 

The basic components of the $1$--gauge transformation behave as expected 
when the $1$--gauge transformation is special.

\begin{prop}
If the $1$--gauge transformation and the adapted coordinates are both special 
(cf. defs. \cref{defi:prop1gau}, \cref{defi:specadpcoo}), then 
\begin{align}
&I_U{}^*J_{\mathrm{b}}=0,
\vphantom{\Big]}
\label{propbas3}
\\
&I_U{}^*K_{\mathrm{b}}=0,
\vphantom{\Big]}
\label{propbas4}
\end{align}
where $I_U:\pi_0{}^{-1}(U)\rightarrow\pi^{-1}(U)$ is the inclusion map. 
\end{prop}

\begin{proof}
This follows from \ceqref{basic10}, \ceqref{basic12} upon substituting 
\ceqref{0cmgautr5ol}, \ceqref{0cmgautr6ol} and \ceqref{bmaurer15}, \ceqref{bmaurer16}
\end{proof}

Finally, we consider a $2$--gauge transformation of $P$ of modification and variation components 
$E$, $C$ relative to the reference $2$--connection 
$\omega$, $\varOmega$, $\theta$, $\varTheta$ (cf. subsect. \cref{subsec:cmgauforgau}, def. \cref{defi:2gauge}). 

\begin{defi}
The basic modification and variation components of the $2$--gau\-ge transformation 
are the Lie group and algebra valued internal functions 
$E_{\mathrm{b}}\in\iMap(T[1]\pi^{-1}(U),\mathsans{E})$ and  
$C_{\mathrm{b}}\in\iMap(T[1]\pi^{-1}(U),\mathfrak{e}[1])$ given by \pagebreak 
\begin{align}
&E_{\mathrm{b}}=\mu(\gamma,E),
\vphantom{\Big]}
\label{basic17}
\\
&C_{\mathrm{b}}=\mu\dot{}\,(\gamma,C)+\varGamma-\Ad\mu(\gamma,E)(\varGamma).
\vphantom{\Big]}
\label{basic18}
\end{align}
\end{defi}

\noindent
Again, restriction of $E$, $C$ to $T[1]\pi^{-1}(U)$ is here also tacitly understood. 
Further, the names given to $E_{\mathrm{b}}$, $C_{\mathrm{b}}$ reflect their basicness.

\begin{prop}
$E_{\mathrm{b}}$, $C_{\mathrm{b}}$ are basic elements of the operation $\iOOO S_{\pi^{-1}(\,U)}$.
\end{prop}

\begin{proof}
To show that $E_{\mathrm{b}}$, $C_{\mathrm{b}}$ are annihilated by all derivations 
$j_{\pi^{-1}(U)Z}$ and $l_{\pi^{-1}(U)Z}$ with $Z\in\DD\mathfrak{m}$
we use relations  \ceqref{cmgauforgau9}--\ceqref{cmgauforgau12} and \ceqref{bmaurer3}--\ceqref{bmaurer6}.
\end{proof}


\begin{prop}
$E_{\mathrm{b}}$, $C_{\mathrm{b}}$ satisfy obey the relations
\begin{align}
&d_{\pi^{-1}(U)}E_{\mathrm{b}}E_{\mathrm{b}}{}^{-1}=-C_{\mathrm{b}}-\dot{}\mu(\omega_{\mathrm{b}},E_{\mathrm{b}}), 
\vphantom{\Big]}
\label{basic19}
\\
&d_{\pi^{-1}(U)}C_{\mathrm{b}}=-\frac{1}{2}[C_{\mathrm{b}},C_{\mathrm{b}}]-\dot{}\mu\dot{}\,(\omega_{\mathrm{b}},C_{\mathrm{b}})
-\dot{}\mu(\theta_{\mathrm{b}},E_{\mathrm{b}})-\varOmega_{\mathrm{b}}+\Ad E_{\mathrm{b}}(\varOmega_{\mathrm{b}}).
\vphantom{\Big]}
\label{basic20}
\end{align}
\end{prop}

\noindent
As expected by now, these are analogous in form to \ceqref{cmgauforgau7}, \ceqref{cmgauforgau8}. 

\begin{proof}
Combining \ceqref{cmgauforgau7}, \ceqref{cmgauforgau8} and \ceqref{bmaurer1}, \ceqref{bmaurer2}
and carrying out a simple computation, \ceqref{basic19}, \ceqref{basic20} are readily obtained. 
\end{proof}

As to the $2$--gauge transform ${}^{E,C}g$, ${}^{E,C}J$, ${}^{E,C}h$, ${}^{E,C}K$ of the $1$--gauge 
transformation $g$, $J$, $h$, $K$ (cf. subsect. \cref{subsec:cmgauforgau}, def. \cref{defi:2gautransf}), 
the following result holds. 

\begin{prop}
${}^{E,C}g_{\mathrm{b}}$, ${}^{E,C}J_{\mathrm{b}}$, ${}^{E,C}h_{\mathrm{b}}$, ${}^{E,C}K_{\mathrm{b}}$ are given in terms 
of $g_{\mathrm{b}}$, $J_{\mathrm{b}}$, $h_{\mathrm{b}}$, $K_{\mathrm{b}}$ and $E_{\mathrm{b}}$, $C_{\mathrm{b}}$ by 
the expressions 
\begin{align}
&{}^{E,C}g_{\mathrm{b}}
=\tau(E_{\mathrm{b}})g_{\mathrm{b}}, 
\vphantom{\Big]}
\label{basic25}
\\
&{}^{E,C}J_{\mathrm{b}}
=\Ad E_{\mathrm{b}}(J_{\mathrm{b}})+\dot{}\mu(\omega_{\mathrm{b}}
-\Ad g_{\mathrm{b}}(\omega_{\mathrm{b}})-h_{\mathrm{b}},E_{\mathrm{b}})+C_{\mathrm{b}},
\vphantom{\Big]}
\label{basic26}
\\
&{}^{E,C}h_{\mathrm{b}}
=\Ad\tau(E_{\mathrm{b}})(h_{\mathrm{b}})
+\dot\tau(\,\dot{}\mu(\Ad g_{\mathrm{b}}(\omega_{\mathrm{b}})+h_{\mathrm{b}},E_{\mathrm{b}})),
\vphantom{\Big]}
\label{basic27}
\\
&{}^{E,C}K_{\mathrm{b}}
=\Ad E_{\mathrm{b}}(K_{\mathrm{b}})+\dot{}\mu\dot{}\,(\Ad(\tau(E_{\mathrm{b}})g_{\mathrm{b}})(\omega_{\mathrm{b}}), 
\dot{}\mu(\omega_{\mathrm{b}}-\Ad g_{\mathrm{b}}(\omega_{\mathrm{b}})-h_{\mathrm{b}},E_{\mathrm{b}})
\vphantom{\Big]}
\label{basic28}
\\
&\hspace{2.2cm}
+C_{\mathrm{b}})+\dot{}\mu(\Ad g_{\mathrm{b}}(\theta_{\mathrm{b}})
+\dot\tau(\mu\dot{}\,(g_{\mathrm{b}},\varOmega_{\mathrm{b}})
-\dot{}\mu\dot{}\,(\Ad g_{\mathrm{b}}(\omega_{\mathrm{b}}),J_{\mathrm{b}})+K_{\mathrm{b}}),E_{\mathrm{b}}).
\vphantom{\Big]}
\nonumber
\end{align}
\end{prop}

\begin{proof}
Relations \ceqref{basic25}--\ceqref{basic28}
are obtained by combining \ceqref{basic9}--\ceqref{basic12} and 
\ceqref{basic17}, \ceqref{basic18} with 
\ceqref{cmgauforgau1}--\ceqref{cmgauforgau4} through simple 
computations. 
\end{proof}

\noindent
Above, we can regard the right hand sides of eqs. \ceqref{basic25}--\ceqref{basic28}
as the expressions of the basic projected  component gauge transform 
${}^{E_{\mathrm{b}},C_{\mathrm{b}}}g_{\mathrm{b}}$, ${}^{E_{\mathrm{b}},C_{\mathrm{b}}}J_{\mathrm{b}}$, 
${}^{E_{\mathrm{b}},C_{\mathrm{b}}}h$, ${}^{E_{\mathrm{b}},C_{\mathrm{b}}}K$, respectively.
Such expressions are formally identical to those holding for the ordinary projected components,
viz \ceqref{cmgauforgau1}--\ceqref{cmgauforgau4}. Moreover, they provide a 
local basic description of $1$--gauge transformation $2$--gauge transformation of strict higher gauge theory. 

\begin{rem}
The basic components $\omega_{\mathrm{b}}$, $\varOmega_{\mathrm{b}}$, $\theta_{\mathrm{b}}$, $\varTheta_{\mathrm{b}}$,
$g_{\mathrm{b}}$, $J_{\mathrm{b}}$, $h_{\mathrm{b}}$, $K_{\mathrm{b}}$ and $E_{\mathrm{b}}$, $C_{\mathrm{b}}$ 
satisfy relations formally identical to \ceqref{cmgauforgau5}--\ceqref{cmgauforgau6crv}.
\end{rem}

\begin{proof}
Indeed, the basic projected components formally obey the same relations as the
ordinary projected ones. Moreover, the basic projected component $2$-connection $1$--gauge transformation 
and $1$--gauge transformation $2$--gauge transformation are formally given by the same expressions
as their ordinary projected counterparts.
\end{proof}

For a given trivializing neighborhood $U\subset M$,
the basic components of $2$--connections and $1$-- and $2$ gauge transformations are Lie valued 
internal functions on $T[1]\pi^{-1}(U)$ so that they are only locally defined. The problem arises
of matching the local data pertaining to distinct but overlapping trivializing 
neighborhoods $U,\, U'\subset M$. Below, we denote by $\varLambda,\varDelta$ and 
$\varLambda',\varDelta'$ the adapted coordinates modelled on $\DD\mathsans{M}$ of $\pi^{-1}(U)$
$\pi^{-1}(U')$, respectively. 

The inclusion map $N_{U\cap U'}^U:\pi^{-1}(U\cap U')\rightarrow \pi^{-1}(U)$ 
induces a morphisms $Q_{U\cap\,U'}^{\,U}:S_{\pi^{-1}(\,U\cap\,U')}\rightarrow S_{\pi^{-1}(\,U)}$
of the morphism spaces of $\pi^{-1}(\,U\cap\,U')$ and $\pi^{-1}(\,U)$ and via this a morphism
$\iOOO Q_{U\cap\,U'}^{\,U}:\iOOO S_{\pi^{-1}(\,U)}\rightarrow\iOOO S_{\pi^{-1}(\,U\cap\,(U')}$
of the associated operations. Thus, if a function $F\in\iFun(T[1]\pi^{-1}(U))$ 
obeys certain relations under the action of the derivations 
$j_{\pi^{-1}(U)Z}$, $l_{\pi^{-1}(U)Z}$ of $\iOOO S_{\pi^{-1}(\,U)}$, 
its restriction $F|_{T[1]\pi^{-1}(U\cap U')}:=N_{U\cap U'}^U{}^*F\in\iFun(T[1]\pi^{-1}(U\cap U'))$
obeys formally identical relations under the action of the derivations 
$j_{\pi^{-1}(U\cap U')Z}$, $l_{\pi^{-1}(U\cap U')Z}$ of $\iOOO S_{\pi^{-1}(\,U\cap\,U')}$. 
Similar remarks hold with $U$ replaced by $U'$. 

\begin{defi} \label{defi:matchfnc}
The local basic matching transformation and shift components are the Lie group and algebra valued internal 
functions $G_{\mathrm{b}}\in\iMap(T[1]\pi^{-1}(U\cap U')$, $\DD\mathsans{M})$ and 
$D_{\mathrm{b}}\in\iMap(T[1]\pi^{-1}(U\cap U'),\DD\mathfrak{m}[1])$ defined by 
\begin{align}
G_{\mathrm{b}}&=\varLambda'\varLambda^{-1},
\vphantom{\Big]}
\label{2basic1}
\\
D_{\mathrm{b}}&=\Ad \varLambda(\varDelta'-\varDelta). 
\vphantom{\Big]}
\label{2basic2}
\end{align}
\end{defi}

\noindent
Above, restriction of $\varLambda$, $\varDelta$ and $\varLambda'$, $\varDelta'$ to $T[1]\pi^{-1}(U\cap\,U')$ 
is understood for simplicity. The names given to $G_{\mathrm{b}}$, $D_{\mathrm{b}}$ are justified by the following proposition.  

\begin{prop}\label{prop:2basic1}
$G_{\mathrm{b}}$, $D_{\mathrm{b}}$ are basic elements of the operation $\iOOO S_{\pi^{-1}(\,U\cap\,U')}$.
\end{prop}

\begin{proof}
Using relations \ceqref{bmaurer3}--\ceqref{bmaurer6} and their primed counterpart, one easily verifies that 
$G_{\mathrm{b}}$, $D_{\mathrm{b}}$ are annihilated by all derivations 
$j_{\pi^{-1}(U\cap U')Z}$ and $l_{\pi^{-1}(U\cap U')Z}$ with $Z\in\DD\mathfrak{m}$. 
\end{proof}


\begin{prop}\label{prop:2basic2}
The following relation holds, 
\begin{equation}
D_{\mathrm{b}}=G_{\mathrm{b}}{}^{-1}d_{\pi^{-1}(U\cap U')}G_{\mathrm{b}}+G_{\mathrm{b}}{}^{-1}d_{\dot\tau}G_{\mathrm{b}}.
\label{2basic3}
\end{equation}
\end{prop}

\begin{proof}
Identity \ceqref{2basic3} is easily verified substituting relation \ceqref{bmaurer1} and its primed counterpart
into eq. \ceqref{2basic2}. 
\end{proof}

By virtue of 3.5.1, 3.5.12 of I, we can expand the local basic matching 
components $G_{\mathrm{b}}$, $D_{\mathrm{b}}$ as \hphantom{xxxxxxxxx}
\begin{align}
&G_{\mathrm{b}}(\alpha)=\ee^{\alpha F_{\mathrm{b}}}f_{\mathrm{b}},
\vphantom{\Big]}
\label{2basic4}
\\
&D_{\mathrm{b}}(\alpha)=s_{\mathrm{b}}-\alpha S_{\mathrm{b}},\qquad\alpha\in\mathbb{R}[1],
\vphantom{\Big]}
\label{2basic5}
\end{align}
where 
$f_{\mathrm{b}}\in\iMap(T[1]\pi^{-1}(U\cap U'),\mathsans{G})$, 
$F_{\mathrm{b}}\in\iMap(T[1]\pi^{-1}(U\cap U'),\mathfrak{e}[1])$
and 
$s_{\mathrm{b}}\in$ $\iMap(T[1]\pi^{-1}(U\cap U')$, $\mathfrak{g}[1])$, 
$S_{\mathrm{b}}\in\iMap(T[1]\pi^{-1}(U\cap U'),\mathfrak{e}[2])$
are suitable projected local basic matching components.

Let $\gamma$, $\varGamma$, $\sigma$, $\varSigma$ and $\gamma'$, 
$\varGamma'$, $\sigma'$, $\varSigma'$ be the projected components 
of $\varLambda$, $\varDelta$ and $\varLambda'$, $\varDelta'$, 
respectively (cf. eqs. \ceqref{bmaurer7}--\ceqref{bmaurer8}). 

\begin{prop}
The  basic projected components $f_{\mathrm{b}}$, $F_{\mathrm{b}}$, $s_{\mathrm{b}}$, $S_{\mathrm{b}}$
can be expressed in terms of the projected components $\gamma$, $\varGamma$, $\sigma$, $\varSigma$ and $\gamma'$, 
$\varGamma'$, $\sigma'$, $\varSigma'$ as 
\begin{align}
&f_{\mathrm{b}}=\gamma'\gamma^{-1}, \hspace{5cm}
\vphantom{\Big]}
\label{2basic6}
\end{align} 
\vspace{-.9cm}\eject\noindent
\begin{align}
&F_{\mathrm{b}}=\varGamma'-\mu{}\dot{}(\gamma'\gamma^{-1},\varGamma),
\vphantom{\Big]}
\label{2basic7}
\\
&s_{\mathrm{b}}=\Ad\gamma(\sigma'-\sigma),
\vphantom{\Big]}
\label{2basic8}
\\
&S_{\mathrm{b}}=\mu{}\dot{}(\gamma,\varSigma'-\varSigma)-{}\dot{}\mu{}\dot{}(\Ad\gamma(\sigma'-\sigma),\varGamma).
\vphantom{\Big]}
\label{2basic9}
\end{align}
\end{prop}

\noindent
Here, $\gamma$, $\varGamma$, $\sigma$, $\varSigma$ and $\gamma'$, 
$\varGamma'$, $\sigma'$, $\varSigma'$ are restricted to $T[1]\pi^{-1}(U\cap\,U')$. 

\begin{proof}
Relations \ceqref{2basic6}--\ceqref{2basic9} follow straightforwardly 
from inserting \ceqref{bmaurer7}--\ceqref{bmaurer8},
its primed counterpart and \ceqref{2basic4}, \ceqref{2basic5} 
into \ceqref{2basic1}, \ceqref{2basic2} and applying 3.5.2, 
3.5.3 and 3.5.18 of I.
\end{proof}

\begin{prop}
The basic projected components $f_{\mathrm{b}}$, $F_{\mathrm{b}}$, $s_{\mathrm{b}}$, $S_{\mathrm{b}}$
are related as 
\begin{align}
&s_{\mathrm{b}}=f_{\mathrm{b}}{}^{-1}d_{\pi^{-1}(U\cap U')}f_{\mathrm{b}}+\dot\tau(\mu{}\dot{}(f_{\mathrm{b}}{}^{-1},F_{\mathrm{b}})),
\vphantom{\Big]}
\label{2basic10}
\\
&S_{\mathrm{b}}=\mu{}\dot{}(f_{\mathrm{b}}{}^{-1},d_{\pi^{-1}(U\cap U')}F_{\mathrm{b}}+[F_{\mathrm{b}},F_{\mathrm{b}}]/2).
\vphantom{\Big]}
\label{2basic11}
\end{align}
\end{prop}

\begin{proof}
Substituting \ceqref{2basic4} into \ceqref{2basic3} and 
applying 3.5.23 and 3.5.25 of I, one readily 
obtains relations \ceqref{2basic10}--\ceqref{2basic11}. 
\end{proof}

Consider next a $2$--connection of $\hat{\mathcal{P}}$ of components $A$, $B$. 

\begin{prop}
The basic components $A_{\mathrm{b}}$, $B_{\mathrm{b}}$ and $A'{}_{\mathrm{b}}$, $B'{}_{\mathrm{b}}$ 
of the $2$--connection are related on $T[1]\pi^{-1}(U\cap U')$ as 
\begin{align}
A'{}_{\mathrm{b}}&=\Ad G_{\mathrm{b}}(A_{\mathrm{b}}-D_{\mathrm{b}}),
\vphantom{\Big]}
\label{2basic12}
\\
B'{}_{\mathrm{b}}&=\Ad G_{\mathrm{b}}(B_{\mathrm{b}}).
\vphantom{\Big]}
\label{2basic13}
\end{align}
\end{prop}

\begin{proof}
Exploiting relations \ceqref{sfbasic1}, \ceqref{sfbasic2}, we can express 
$A$, $B$ in terms of $A_{\mathrm{b}}$, $B_{\mathrm{b}}$, $\varLambda$, $\varDelta$. 
Inserting these identities into 
the primed counterparts of \ceqref{sfbasic1}, \ceqref{sfbasic2}, we obtain expressions of
$A'{}_{\mathrm{b}}$, $B'{}_{\mathrm{b}}$ in terms of $A_{\mathrm{b}}$, $B_{\mathrm{b}}$, $\varLambda$, $\varDelta$, 
$\varLambda'$, $\varDelta'$. These latter  
can be cast in the form \ceqref{2basic12}, \ceqref{2basic13} employing \ceqref{2basic1}, \ceqref{2basic2}.
\end{proof}

\noindent
The matching relation \ceqref{2basic12}, \ceqref{2basic13} can be written in terms of projected components.

\begin{prop}
The projected basic components $\omega_{\mathrm{b}}$, $\varOmega_{\mathrm{b}}$, $\theta_{\mathrm{b}}$, $\varTheta_{\mathrm{b}}$ 
and $\omega'{}_{\mathrm{b}}$, $\varOmega'{}_{\mathrm{b}}$, $\theta'{}_{\mathrm{b}}$, $\varTheta'{}_{\mathrm{b}}$ 
of the $2$--connection are related on $T[1]\pi^{-1}(U\cap U')$ as \pagebreak 
\begin{align}
&\omega'{}_{\mathrm{b}}=\Ad f_{\mathrm{b}}(\omega_{\mathrm{b}}-s_{\mathrm{b}}),
\vphantom{\Big]}
\label{2basic14}
\\
&\varOmega'{}_{\mathrm{b}}=\mu{}\dot{}(f_{\mathrm{b}},\varOmega_{\mathrm{b}}-S_{\mathrm{b}})
-{}\dot{}\mu{}\dot{}(\Ad f_{\mathrm{b}}(\omega_{\mathrm{b}}-s_{\mathrm{b}}),F_{\mathrm{b}}),
\vphantom{\Big]}
\label{2basic15}
\\
&\theta'{}_{\mathrm{b}}=\Ad f_{\mathrm{b}}(\theta_{\mathrm{b}}),
\vphantom{\Big]}
\label{2basic16}
\\
&\varTheta'{}_{\mathrm{b}}=\mu{}\dot{}(f_{\mathrm{b}},\varTheta_{\mathrm{b}})
-{}\dot{}\mu{}\dot{}(\Ad f_{\mathrm{b}}(\theta_{\mathrm{b}}),F_{\mathrm{b}}).
\vphantom{\Big]}
\label{2basic17}
\end{align}
\end{prop}

\noindent
Relations \ceqref{2basic10}, \ceqref{2basic11} entail that, at the basic level, eqs. \ceqref{2basic14}--\ceqref{2basic17}
are of the same form as the 
matching relations of the projected components of a $2$--connection 
in strict higher gauge theory. See subsect. \cref{subsec:nadifcoh} for more in this. 

\begin{proof}
Inserting \ceqref{sfbasic11}, \ceqref{sfbasic12}, their primed counterparts and \ceqref{2basic4}, \ceqref{2basic5},
into \ceqref{2basic12}, \ceqref{2basic13} and using 3.5.18 of I, we obtain the \ceqref{2basic14}--\ceqref{2basic17}
by simple calculations. 
\end{proof}

Next, consider a $1$--gauge transformation of $\hat{\mathcal{P}}$ of components $\varPsi$, $\varUpsilon$. 

\begin{prop}
The basic components $\varPsi_{\mathrm{b}}$, $\varUpsilon_{\mathrm{b}}$ 
and $\varPsi'{}_{\mathrm{b}}$, $\varUpsilon'{}_{\mathrm{b}}$ 
of the $1$--gauge transformation are related on $T[1]\pi^{-1}(U\cap U')$ as 
\begin{align}
\varPsi'{}_{\mathrm{b}}&=G_{\mathrm{b}}\varPsi_{\mathrm{b}}G_{\mathrm{b}}^{-1},
\vphantom{\Big]}
\label{2basic18}
\\
\varUpsilon'{}_{\mathrm{b}}&=\Ad G_{\mathrm{b}}(\varUpsilon_{\mathrm{b}}-D_{\mathrm{b}}+\Ad \varPsi_{\mathrm{b}}(D_{\mathrm{b}})).
\vphantom{\Big]}
\label{2basic19}
\end{align}
\end{prop}

\begin{proof}
Exploiting relations \ceqref{sfbasic5}, \ceqref{sfbasic6}, we can express 
$\varPsi$, $\varUpsilon$ in terms of $\varPsi_{\mathrm{b}}$, $\varUpsilon_{\mathrm{b}}$, $\varLambda$, $\varDelta$. 
Inserting these identities into 
the primed counterparts of \ceqref{sfbasic5}, \ceqref{sfbasic6}, we obtain expressions of
$\varPsi'{}_{\mathrm{b}}$, $\varUpsilon'{}_{\mathrm{b}}$ in terms of 
$\varPsi_{\mathrm{b}}$, $\varUpsilon_{\mathrm{b}}$, $\varLambda$, $\varDelta$, 
$\varLambda'$, $\varDelta'$. These latter  
can be cast in the form \ceqref{2basic18}, \ceqref{2basic19} employing \ceqref{2basic1}, \ceqref{2basic2}.
\end{proof}

\noindent
The matching relation \ceqref{2basic18}, \ceqref{2basic19} can be written in terms of projected components.

\begin{prop}
The projected basic components $g_{\mathrm{b}}$, $J_{\mathrm{b}}$, $h_{\mathrm{b}}$, $K_{\mathrm{b}}$ 
and $g'{}_{\mathrm{b}}$, $J'{}_{\mathrm{b}}$, $h'{}_{\mathrm{b}}$, $K'{}_{\mathrm{b}}$ 
of the $1$--gauge transformation are related on $T[1]\pi^{-1}(U\cap U')$ as 
\begin{align}
&g'{}_{\mathrm{b}}=f_{\mathrm{b}}g_{\mathrm{b}}f_{\mathrm{b}}{}^{-1},
\vphantom{\Big]}
\label{2basic20}
\\
&J'{}_{\mathrm{b}}=\mu{}\dot{}(f_{\mathrm{b}},J_{\mathrm{b}})+F_{\mathrm{b}}
-\mu{}\dot{}(f_{\mathrm{b}}g_{\mathrm{b}}f_{\mathrm{b}}{}^{-1},F_{\mathrm{b}}),
\vphantom{\Big]}
\label{2basic21}
\end{align} 
\begin{align}
&h'{}_{\mathrm{b}}=\Ad f_{\mathrm{b}}(h_{\mathrm{b}}-s_{\mathrm{b}}+\Ad g_{\mathrm{b}}(s_{\mathrm{b}})),
\vphantom{\Big]}
\label{2basic22}
\\
&K'{}_{\mathrm{b}}=\mu{}\dot{}(f_{\mathrm{b}},K_{\mathrm{b}}-S_{\mathrm{b}}
+\mu{}\dot{}(g_{\mathrm{b}},S_{\mathrm{b}})-{}\dot{}\mu{}\dot{}(\Ad g_{\mathrm{b}}(s_{\mathrm{b}}),J_{\mathrm{b}})
\vphantom{\Big]}
\label{2basic23}
\\
&\hspace{4cm}
-{}\dot{}\mu{}\dot{}(h_{\mathrm{b}}-s_{\mathrm{b}}
+\Ad g_{\mathrm{b}}(s_{\mathrm{b}}),\mu{}\dot{}(f_{\mathrm{b}}{}^{-1},F_{\mathrm{b}}))).
\vphantom{\Big]}
\nonumber
\end{align}
\end{prop}

\noindent
In view of eqs. \ceqref{2basic10}, \ceqref{2basic11},  the \ceqref{2basic20}--\ceqref{2basic23}
reproduce at the basic level 
the matching relations of the projected components of a $1$--gauge transformation of strict higher gauge theory.
We will come back to this in subsect. \cref{subsec:nadifcoh}. 

\begin{proof}
Inserting \ceqref{sfbasic13}, \ceqref{sfbasic14}, their primed counterparts and \ceqref{2basic4}, \ceqref{2basic5},
into \ceqref{2basic18}, \ceqref{2basic19} and using 3.5.2, 3.5.3 and 3.5.18 of I, 
we obtain the \ceqref{2basic20}--\ceqref{2basic23} through straightforward computations. 
\end{proof}

Finally consider a $2$--gauge transformation of $\hat{\mathcal{P}}$ of components $E$, $C$.

\begin{prop}
The projected basic components $E_{\mathrm{b}}$, $C_{\mathrm{b}}$ 
and $E'{}_{\mathrm{b}}$, $C'{}_{\mathrm{b}}$
of the $2$--gauge transformation are related on $T[1]\pi^{-1}(U\cap U')$ as 
\begin{align}
&E'{}_{\mathrm{b}}=\mu(f_{\mathrm{b}},E_{\mathrm{b}}),
\vphantom{\Big]}
\label{2basic24}
\\
&C'{}_{\mathrm{b}}=\mu\dot{}\,(f_{\mathrm{b}},C_{\mathrm{b}})+F_{\mathrm{b}}-\Ad\mu(f_{\mathrm{b}},E_{\mathrm{b}})(F_{\mathrm{b}}). 
\vphantom{\Big]}
\label{2basic25}
\end{align}
\end{prop}

\begin{proof}
By relations \ceqref{basic17}, \ceqref{basic18}, we can express 
$E$, $C$ in terms of $E_{\mathrm{b}}$, $C_{\mathrm{b}}$, $\gamma$, $\varGamma$. 
Inserting these identities into 
the primed counterparts of \ceqref{basic17}, \ceqref{basic18}, we obtain expressions of
$E'{}_{\mathrm{b}}$, $C'{}_{\mathrm{b}}$ in terms of $E_{\mathrm{b}}$, $C_{\mathrm{b}}$, 
$\gamma$, $\varGamma$, $\gamma'$, $\varGamma'$. These latter  
can be rewritten in the form \ceqref{2basic24}, \ceqref{2basic25} employing \ceqref{2basic6}, \ceqref{2basic7}. 
\end{proof}

\noindent
It is noteworthy that the matching relations do not involve the underlying  reference $2$--connection. 

The basic matching components behave as expected when the adapted coordinates used are special. 

\begin{prop}
If the two sets of adapted coordinates involved are both special 
(cf. defs. \cref{defi:specadpcoo}), then one has \hphantom{xxxxxxxxxx}
\begin{align}
&I_{U\cap U'}{}^*F_{\mathrm{b}}=0,
\vphantom{\Big]}
\label{propbas5}
\\
&I_{U\cap U'}{}^*S_{\mathrm{b}}=0,
\vphantom{\Big]}
\label{propbas6}
\end{align}
where $I_{U\cap U'}:\pi_0{}^{-1}(U\cap U')\rightarrow\pi^{-1}(U\cap U')$ is the inclusion map. 
\end{prop}

\begin{proof}
This follows from \ceqref{2basic7}, \ceqref{2basic9} upon substituting 
\ceqref{bmaurer15}, \ceqref{bmaurer16} and its primed counterpart. 
\end{proof}

\noindent
Note that this property renders the matching relations \ceqref{2basic15}, \ceqref{2basic17},
respectively \ceqref{2basic21}, \ceqref{2basic23}, 
compatible with \ceqref{propbas1}, \ceqref{propbas2}, respectively \ceqref{propbas3}, \ceqref{propbas4}, 
in case the relevant $2$--connection, respectively $1$--gauge transformation, is special.


\subsection{\textcolor{blue}{\sffamily Relation to non Abelian differential cocycles}}\label{subsec:nadifcoh}

In this subsection, we shall explore whether $2$--connections and $1$--gauge transformations as defined in the 
synthetic theory of subsects. \cref{subsec:cmconn}, \cref{subsec:cmgautr}
can be related to non Abelian differential cocycles and their equivalences \ccite{Breen:2001ie,Schr:2013cnag}. 
We consider again a synthetic principal $\hat{\matheul{K}}$--$2$--bundle $\hat{\mathcal{P}}$ and its
associated synthetic setup. 

In subsect. \cref{subsec:maurer}, we have seen that we can describe the portion $\pi^{-1}(U)$ of $P$ lying above 
a trivializing neighborhood $U$ of $M$ by means of adapted coordinates $\gamma$, $\varGamma$, $\sigma$, $\varSigma$. 
Since $\sigma$ and $\varSigma$ are expressible in terms of $\gamma$, $\varGamma$ through relations \ceqref{maurer2},
\ceqref{maurer3}, only these latter are truly independent. So, we shall limit ourselves to their sole consideration. 

In subsect. \cref{subsec:basic}, using adapted coordinates we have constructed via \ceqref{basic1}--\ceqref{basic4}
the local basic data
$\omega_{\mathrm{b}}$, $\varOmega_{\mathrm{b}}$, $\theta_{\mathrm{b}}$, $\varTheta_{\mathrm{b}}$
associated with a $2$--connection on $\pi^{-1}(U)$. 
Of these, $\theta_{\mathrm{b}}$, $\varTheta_{\mathrm{b}}$. can be expressed in terms of 
$\omega_{\mathrm{b}}$, $\varOmega_{\mathrm{b}}$ by eqs. \ceqref{basic5}, \ceqref{basic6}
and so can be disregarded in the following. 
Similarly, through adapted coordinates 
we have constructed via \ceqref{basic9}--\ceqref{basic12} also the local basic data
$g_{\mathrm{b}}$, $J_{\mathrm{b}}$, $h_{\mathrm{b}}$, $K_{\mathrm{b}}$ 
associated with a $1$--gauge transformation on $\pi^{-1}(U)$. 
Again, of these $h_{\mathrm{b}}$ and $K_{\mathrm{b}}$ can be given in term of $g_{\mathrm{b}}$, $J_{\mathrm{b}}$ 
by eqs. \ceqref{basic13}, \ceqref{basic14} and so can be once more disregarded. 
In this way, the basic data ${}^{g,J}\omega_{\mathrm{b}}$, ${}^{g,J}\varOmega_{\mathrm{b}}$
of the $1$--gauge transformed $2$--connection can be expressed in terms of 
$\omega_{\mathrm{b}}$, $\varOmega_{\mathrm{b}}$
and $g_{\mathrm{b}}$, $J_{\mathrm{b}}$ only by the familiar higher gauge theoretic relations, 
as we found out by inserting \ceqref{basic13}, \ceqref{basic14} into \ceqref{basic21}, \ceqref{basic22}.

In subsect. \cref{subsec:basic}, further, \pagebreak we have seen that the matching of local basic $2$--connection and 
$1$--gauge transformation data relative to overlapping neighborhoods $U$, $U'$ of $M$ 
is governed by local basic transition data $f_{\mathrm{b}}$, $F_{\mathrm{b}}$, $s_{\mathrm{b}}$, $S_{\mathrm{b}}$
given by eqs. \ceqref{2basic6}--\ceqref{2basic9} 
of which the latter two are expressible in terms of the former two by eqs. \ceqref{2basic10}, \ceqref{2basic11}
and so can be also safely left aside in the following. 

We now choose an trivializing covering $\{U_i\}$ of $M$ and for each set $U_i$ adapted coordinates 
$\gamma_i$, $\varGamma_i$ and consider the associated local $2$--connection, $1$--gauge transformation 
and transition data. To relate the present framework to non Abelian differential cocycle theory, 
we shall restrict ourselves fake flat $2$--connections as appropriate. 

For a $2$--connection, there are then defined for every set $U_i$ of the covering
lo\-cal basic data $\omega_{\mathrm{b}i}\in\iMap(T[1]\pi^{-1}(U_i),\mathfrak{g}[1])$, 
$\varOmega_{\mathrm{b}i}\in\iMap(T[1]\pi^{-1}(U_i),\mathfrak{e}[2])$
via \ceqref{basic1}, \ceqref{basic2}. By the assumed fake flatness,
these satisfy 
\begin{equation}
d_{\pi^{-1}(U_i)}\omega_{\mathrm{b}i}+\frac{1}{2}[\omega_{\mathrm{b}i},\omega_{\mathrm{b}i}]-\dot\tau(\varOmega_{\mathrm{b}i})=0.
\label{nadifcoh1}
\end{equation}
For a $1$--gauge transformation, local basic data 
$g_{\mathrm{b}i}\in\iMap(T[1]\pi^{-1}(U_i),\mathsans{G})$, $J_{\mathrm{b}i}\in\iMap(T[1]\pi^{-1}(U_i),\mathfrak{e}[1])$
can be similarly defined on each $U_i$ via \ceqref{basic9}, \ceqref{basic10}. 

For every couple of intersecting sets $U_i$, $U_j$ of the covering, 
transition data $f_{\mathrm{b}ij}\in\iMap(T[1]\pi^{-1}(U_i\cap U_j),\mathsans{G})$, 
$F_{\mathrm{b}ij}\in\iMap(T[1]\pi^{-1}(U_i\cap U_j),\mathfrak{e}[1])$ are likewise 
built through \ceqref{2basic6}, \ceqref{2basic7}. The local $2$--connection data $\omega_{\mathrm{b}i}$, $\varOmega_{\mathrm{b}i}$,
match as 
\begin{align}
&\omega_{\mathrm{b}i}=\Ad f_{\mathrm{b}ij}(\omega_{\mathrm{b}j})-d_{\pi^{-1}(U_i\cap U_j)}f_{\mathrm{b}ij}f_{\mathrm{b}ij}{}^{-1}
-\dot\tau(F_{\mathrm{b}ij}),
\vphantom{\Big]}
\label{nadifcoh2}
\\
&\varOmega_{\mathrm{b}i}=\mu{}\dot{}\,(f_{\mathrm{b}ij},\varOmega_{\mathrm{b}j})
-d_{\pi^{-1}(U_i\cap U_j)}F_{\mathrm{b}ij}-\frac{1}{2}[F_{\mathrm{b}ij},F_{\mathrm{b}ij}]
-{}\dot{}\mu{}\dot{}\,(\omega_{\mathrm{b}i},F_{\mathrm{b}ij})
\vphantom{\Big]}
\label{nadifcoh3}
\end{align}
on $U_i\cap U_j$, as follows readily from eqs. \ceqref{2basic14}, \ceqref{2basic15}
using the \ceqref{2basic10}, \ceqref{2basic11}. Similarly, 
 the the local $1$--gauge transformation data $g_{\mathrm{b}}$, $J_{\mathrm{b}}$
match as 
\begin{align}
&g_{\mathrm{b}i}=f_{\mathrm{b}ij}g_{\mathrm{b}j}f_{\mathrm{b}ij}{}^{-1},
\vphantom{\Big]}
\label{nadifcoh4}  
\\
&J_{\mathrm{b}i}=\mu{}\dot{}(f_{\mathrm{b}ij},J_{\mathrm{b}j})+F_{\mathrm{b}ij}
-\mu{}\dot{}(g_{\mathrm{b}i},F_{\mathrm{b}ij})
\vphantom{\Big]}
\label{nadifcoh5}   
\end{align}
by eqs. \ceqref{2basic20}, \ceqref{2basic21}. 

By virtue of relations \pagebreak \ceqref{2basic6}, \ceqref{2basic7}, the data $f_{\mathrm{b}ij}$. $F_{\mathrm{b}ij}$
form a $\DD\mathsans{M}$--valued $1$--cocycle on $P$, as on every non empty triple intersection $U_i\cap U_j\cap U_k$
\begin{align}
&f_{\mathrm{b}ik}=f_{\mathrm{b}ij}f_{\mathrm{b}jk},
\vphantom{\Big]}
\label{nadifcoh6}
\\
&F_{\mathrm{b}ik}=F_{\mathrm{b}ij}+\mu{}\dot{}\,(f_{\mathrm{b}ij},F_{\mathrm{b}jk}).
\vphantom{\Big]}
\label{nadifcoh7}
\end{align}
By the way it is constructed, this cocycle is trivial. 

Combining \ceqref{basic13}, \ceqref{basic14} into \ceqref{basic21}, \ceqref{basic22}, 
the local basic data ${}^{g,J}\omega_{\mathrm{b}i}$, ${}^{g,J}\varOmega_{\mathrm{b}i}$
of the $1$--gauge transformed $2$--connection are found to be given by 
\begin{align}
&{}^{g,J}\omega_{\mathrm{b}i}=\Ad g_{\mathrm{b}i}(\omega_{\mathrm{b}i})-d_{\pi^{-1}(U_i)}g_{\mathrm{b}i}g_{\mathrm{b}i}{}^{-1}
-\dot\tau(J_{\mathrm{b}i}),
\vphantom{\Big]}
\label{nadifcoh2/1}
\\
&{}^{g,J}\varOmega_{\mathrm{b}i}=\mu{}\dot{}\,(g_{\mathrm{b}i},\varOmega_{\mathrm{b}i})
-d_{\pi^{-1}(U_i)}J_{\mathrm{b}i}-\frac{1}{2}[J_{\mathrm{b}i},J_{\mathrm{b}i}]
-{}\dot{}\mu{}\dot{}\,({}^{g,J}\omega_{\mathrm{b}i},J_{\mathrm{b}i})
\vphantom{\Big]}
\label{nadifcoh3/1}
\end{align} 
for any covering set $U_i$. 

Our aim next is ascertaining whether the above setup can be naturally related to (some internal variant of)
non Abelian differential cocycle theory. We are going to submit a proposal in this sense. 
Before proceeding further, however,  the following remark is in order. 
In an ordinary principal $\mathsans{G}$--bundle $P$, basic forms of $P$ are pull--backs via
the bundle's projection map $\pi$ of ordinary forms of the base $M$. The proof of this 
important property requires crucially that the right $\mathsans{G}$--action of $P$ is transitive 
on the fibers. In a principal $\hat{\matheul{K}}$--$2$--bundle $\hat{\mathcal{P}}$, transitiveness holds only up to isomorphism.
For this reason, basic forms of $P$ do not necessarily arise as pull--backs via
the bundle's projection map $\pi$ of ordinary forms of the base $M$, though they may do. 
Our reformulation of differential cocycle theory hinges on this property.

We have found the following notion useful. 

\begin{defi}
A quasi trivializer consists in an assignment of a basic 
Lie group valued internal function $T_{\mathrm{b}ij}\in\iMap(T[1]\pi^{-1}(U_i\cap U_j),\mathsans{E})$ 
for each pair of intersecting covering sets $U_i$, $U_j$. 
\end{defi}

\noindent We stress that the basicness of the $T_{\mathrm{b}ij}$ is crucial. 


\begin{defi}\label{defi:diffparacocl}
A differential paracocycle is a pair of a fake flat $2$--connection $\{\omega_{\mathrm{b}i}, \varOmega_{\mathrm{b}i}\}$
and a quasi trivializer $\{T_{\mathrm{b}ij}\}$ enjoying the following properties.
\begin{enumerate}

\item For any set $U_i$, Lie algebra valued internal functions 
$\bar\omega_i\in\iMap(T[1]U_i,\mathfrak{g}[1])$, $\bar\varOmega_i\in\iMap(T[1]U_i,\mathfrak{e}[2])$ exist 
with the property that 
\begin{align}
&\omega_{\mathrm{b}i}=\pi^*\bar\omega_i,
\vphantom{\Big]}
\label{nadifcoh8}
\\
&\varOmega_{\mathrm{b}i}=\pi^*\bar\varOmega_i.
\vphantom{\Big]}
\label{nadifcoh9}
\end{align}

\item For any two intersecting sets $U_i$, $U_j$, 
Lie group and algebra valued internal functions 
$\bar f_{ij}\in\iMap(T[1](U_i\cap U_j),\mathsans{G})$, $\bar F_{ij}\in\iMap(T[1](U_i\cap U_j),\mathfrak{e}[1])$
exist such that on $U_i\cap U_j$
\begin{align}
&f_{\mathrm{b}ij}=\tau(T_{\mathrm{b}ij})\pi^*\bar f_{ij},
\vphantom{\Big]}
\label{nadifcoh10}
\\
&F_{\mathrm{b}ij}=\Ad T_{\mathrm{b}ij}(\pi^*\bar F_{ij})-{}\dot{}\mu(\pi^*\bar\omega_i,T_{\mathrm{b}ij})
-d_{\pi^{-1}(U_i\cap U_j)}T_{\mathrm{b}ij}T_{\mathrm{b}ij}{}^{-1}. 
\vphantom{\Big]}
\label{nadifcoh11}
\end{align}

\item 
For any three intersecting sets $U_i$, $U_j$, $U_k$, there is
a Lie group valued internal function 
$\bar T_{ijk}\in\iMap(T[1](U_i\cap U_j\cap U_k),\mathsans{E})$ such that on $U_i\cap U_j\cap U_k$
\begin{equation}
T_{\mathrm{b}ik}{}^{-1}\mu(f_{\mathrm{b}ij},T_{\mathrm{b}jk})T_{\mathrm{b}ij}=\pi^*\bar T_{ijk}.
\label{nadifcoh12}
\end{equation}
\end{enumerate}
\end{defi}

\noindent
The content of the above definition is motivated by the following result which it leads to. 

\begin {prop}
The local $2$--connection and transition data $\{\bar\omega_i,\bar\varOmega_i,\bar f_{ij},\bar F_{ij},\bar T_{ijk}\}$
of a differential paracocycle $\{\omega_{\mathrm{b}i},\varOmega_{\mathrm{b}i},T_{\mathrm{b}ij}\}$ 
constitute 
a  differential cocycle. Indeed, the $2$--connection data $\bar\omega_i$, $\bar\varOmega_i$
satisfy the fake flatness condition 
\begin{equation}
d_{U_i}\bar\omega_i+\frac{1}{2}[\bar\omega_i,\bar\omega_i]-\dot\tau(\bar\varOmega_i)=0
\label{nadifcoh13}
\end{equation}
on every set $U_i$ and the matching conditions 
\begin{align}
&\bar\omega_i=\Ad \bar f_{ij}(\bar\omega_j)-d_{U_i\cap U_j}\bar f_{ij}\bar f_{ij}{}^{-1}
-\dot\tau(\bar F_{ij}),
\vphantom{\Big]}
\label{nadifcoh14}
\\
&\bar\varOmega_i=\mu{}\dot{}\,(\bar f_{ij},\bar\varOmega_j)
-d_{U_i\cap U_j}\bar F_{ij}-\frac{1}{2}[\bar F_{ij},\bar F_{ij}]
-{}\dot{}\mu{}\dot{}\,(\bar\omega_i,\bar F_{ij})
\vphantom{\Big]}
\label{nadifcoh15}
\end{align}
on every non empty intersection $U_i\cap U_j$. Moreover, 
the transition data $\bar f_{ij}$, $\bar F_{ij}$, $\bar T_{ijk}$
satisfy the consistency conditions 
\begin{align}
&\bar f_{ik}=\tau(\bar T_{ijk})\bar f_{ij}\bar f_{jk}, \hspace{8.5cm}
\vphantom{\Big]}
\label{nadifcoh16}
\end{align} 
\begin{align}
&\bar F_{ik}=\Ad \bar T_{ijk}(\bar F_{ij}+\mu{}\dot{}\,(\bar f_{ij},\bar F_{jk}))
-{}\dot{}\mu(\bar\omega_i,\bar T_{ijk})-d_{U_i\cap U_j\cap U_k}\bar T_{ijk}\bar T_{ijk}{}^{-1}
\vphantom{\Big]}
\label{nadifcoh17}
\end{align}
on every non empty intersection $U_i\cap U_j\cap U_k$. Finally, 
\begin{equation}
\bar T_{ikl}\bar T_{ijk}=\bar T_{ijl}\mu(\bar f_{ij},\bar T_{jkl})
\label{nadifcoh18}
\end{equation}
on every non empty intersection $U_i\cap U_j\cap U_k\cap U_l$. 
\end{prop}

\begin{proof}
Relations \ceqref{nadifcoh13}--\ceqref{nadifcoh17} follow from substituting expressions 
\ceqref{nadifcoh8}--\ceqref{nadifcoh11} into relations \ceqref{nadifcoh1}--\ceqref{nadifcoh3},
\ceqref{nadifcoh6}, \ceqref{nadifcoh7} and using \ceqref{nadifcoh12}. The proof involves combined use of 
the identities of app. B of I. The property of $\pi$ being a surjective submersion (cf. prop. 3.2 of I)
is used to deduce that $\bar\tau=0$ from any identity of the form $\pi^*\bar\tau=0$ 
with $\bar\tau$ some local internal function on $M$. \ceqref{nadifcoh18} follows directly from \ceqref{nadifcoh12}
through a simple calculation. 
\end{proof}

\noindent
The above result can be intuitively understood as follows. The local basic data
$\{\omega_{\mathrm{b}i},\varOmega_{\mathrm{b}i}, f_{\mathrm{b}ij}, F_{\mathrm{b}ij}, 1_{\mathsans{E}}\}$
can be viewed  as something like a trivial differential cocycle on $P$. 
By \ceqref{nadifcoh8}--\ceqref{nadifcoh12}, the local basic data 
$\{\pi^*\bar\omega_i,\pi^*\bar\varOmega_i,\pi^*\bar f_{ij},\pi^*\bar F_{ij},\pi^*\bar T_{ijk}\}$
form a trivial differential cocycle on $P$ equivalent to the former. The fundamental cocycle relations 
obeyed by the data $\{\pi^*\bar\omega_i,\pi^*\bar\varOmega_i,\pi^*\bar f_{ij},\pi^*\bar F_{ij},\pi^*\bar T_{ijk}\}$
are then satisfied also by the data $\{\bar\omega_i,\bar\varOmega_i,\bar f_{ij},\bar F_{ij},\bar T_{ijk}\}$, 
since $\pi$ is a surjective submersion. 
The local data $\{\bar\omega_i,\bar\varOmega_i,\bar f_{ij},\bar F_{ij},\bar T_{ijk}\}$ constitute therefore a 
 differential cocycle on $M$. Unlike its counterpart in $P$, 
this cocycle is generally non trivial since in eq. 
\ceqref{nadifcoh12} $T_{\mathrm{b}ij}$ is not necessarily of the form $T_{\mathrm{b}ij}=\pi^*\bar T_{ij}$
for some internal function $\bar T_{ij}\in\iMap(T[1](U_i\cap U_j),\mathsans{E})$.

\begin{defi}\label{defi:nadifcoh1}
Two differential paracocycles $\{\omega_{\mathrm{b}i},\varOmega_{\mathrm{b}i},T_{\mathrm{b}ij}\}$,
$\{\tilde\omega_{\mathrm{b}i},\tilde\varOmega_{\mathrm{b}i},\tilde T_{\mathrm{b}ij}\}$
are said to be equivalent if for every set $U_i$ 
\begin{align}
&\tilde\omega_{\mathrm{b}i}=\omega_{\mathrm{b}i},
\vphantom{\Big]}
\label{nadifcoh19}
\\
&\tilde\varOmega_{\mathrm{b}i}=\varOmega_{\mathrm{b}i}
\vphantom{\Big]}
\label{nadifcoh20}
\end{align}
and for every intersecting set pair $U_i$, $U_j$ there is a Lie group valued internal 
function $\bar T_{ij}\in\iMap(T[1](U_i\cap U_j),\mathsans{E})$ such that 
\begin{equation}
\tilde T_{\mathrm{b}ij}=T_{\mathrm{b}ij}\pi^*\bar T_{ij}{}^{-1}.
\label{nadifcoh21}
\end{equation}
\end{defi}

\noindent
Differential paracocycle equivalence is manifestly an equivalence relation as suggested by its name. 
Further, it implies the equivalence of the underlying differential cocycles.

\begin{prop}\label{prop:nadifcoh1}
If $\{\omega_{\mathrm{b}i},\varOmega_{\mathrm{b}i},T_{\mathrm{b}ij}\}$,
$\{\tilde\omega_{\mathrm{b}i},\tilde\varOmega_{\mathrm{b}i},\tilde T_{\mathrm{b}ij}\}$
are two equivalent differential pa\-racocycles, then their associated  differential cocycles
$\{\bar\omega_i,\bar\varOmega_i,\bar f_{ij},\bar F_{ij},\bar T_{ijk}\}$
 $\{\bar{\tilde\omega}_i,\bar{\tilde\varOmega}_i,\bar{\tilde f}_{ij},\bar{\tilde F}_{ij},\bar{\tilde T}_{ijk}\}$
are equivalent. Indeed, 
\begin{align}
&\bar{\tilde\omega}_i=\bar\omega_i,
\vphantom{\Big]}
\label{nadifcoh22}
\\
&\bar{\tilde\varOmega}_i=\bar\varOmega_i
\vphantom{\Big]}
\label{nadifcoh23}
\end{align}
on each set $U_i$, 
\begin{align}
&\bar{\tilde f}_{ij}=\tau(\bar T_{ij})\bar f_{ij}, \hspace{4.5cm}
\vphantom{\Big]}
\label{nadifcoh24}
\\
&\bar{\tilde F}_{ij}=\Ad\bar T_{ij}(\bar F_{ij})-{}\dot{}\mu(\bar{\tilde \omega}_i,\bar T_{ij})
-d_{U_i\cap U_j}\bar T_{ij}\bar T_{ij}{}^{-1}
\vphantom{\Big]}
\label{nadifcoh25}
\end{align} 
on every non empty intersection $U_i\cap U_j$ and 
\begin{equation}
\bar{\tilde T}_{ijk}=\bar T_{ik}\bar T_{ijk}\mu(\bar f_{ij},\bar T_{jk}{}^{-1})\bar T_{ij}{}^{-1}
\label{nadifcoh26}
\end{equation}
on every non empty intersection  $U_i\cap U_j\cap U_k$
\end{prop}

\begin{proof}
Relations \ceqref{nadifcoh22}, \ceqref{nadifcoh23} are an immediate consequence of
\ceqref{nadifcoh19}, \ceqref{nadifcoh20} and \ceqref{nadifcoh8}, \ceqref{nadifcoh9} and 
their tilded analogues. Relations \ceqref{nadifcoh24}, \ceqref{nadifcoh25} follow from equating 
the tilded and untilded versions of expressions \ceqref{nadifcoh10}, \ceqref{nadifcoh11} and use the resulting 
equations together with \ceqref{nadifcoh21} 
to express $\bar{\tilde f}_{ij}$, $\bar{\tilde F}_{ij}$ in terms of $\bar f_{ij}$, $\bar F_{ij}$.
The proof involves combined use of the identities of app. B of I. 
Finally, \ceqref{nadifcoh26} follows from the tilded version of \ceqref{nadifcoh12} upon using 
\ceqref{nadifcoh21} and the untilded form of \ceqref{nadifcoh12}. 
\end{proof}

\noindent
Intuitively, the above result can be understood as follows. 
In $P$, the differential cocycles 
$\{\pi^*\bar\omega_i,\pi^*\bar\varOmega_i,\pi^*\bar f_{ij},\pi^*\bar F_{ij},\pi^*\bar T_{ijk}\}$,
$\{\pi^*\bar{\tilde\omega}_i,\pi^*\bar{\tilde\varOmega}_i,
\pi^*\bar{\tilde f}_{ij},\pi^*\bar{\tilde F}_{ij},\pi^*\bar{\tilde T}_{ijk}\}$ 
are equivalent to the cocycles 
$\{\omega_{\mathrm{b}i},\varOmega_{\mathrm{b}i}, f_{\mathrm{b}ij}, F_{\mathrm{b}ij}, 1_{\mathsans{E}}\}$, 
$\{\tilde\omega_{\mathrm{b}i},\tilde\varOmega_{\mathrm{b}i}, f_{\mathrm{b}ij}, 
F_{\mathrm{b}ij}, 1_{\mathsans{E}}\}$, respectively. Since the latter two coincide 
by \ceqref{nadifcoh19}, \ceqref{nadifcoh20}, the former two are equivalent. Thanks to 
\ceqref{nadifcoh21}, this property entails the equivalence of the cocycles 
$\{\bar\omega_i,\bar\varOmega_i,\bar f_{ij},\bar F_{ij},\bar T_{ijk}\}$
$\{\bar{\tilde\omega}_i,\bar{\tilde\varOmega}_i,
\bar{\tilde f}_{ij},\bar{\tilde F}_{ij},\bar{\tilde T}_{ijk}\}$. Note that this equivalence is not of the 
most general form, as it does not involve $1$--gauge transformation.

The above analysis shows that the local basic data $\{\omega_{\mathrm{b}i},\varOmega_{\mathrm{b}i}\}$ 
of a fake flat $2$--connection together with the data $\{T_{\mathrm{b}ij}\}$ of a quasi trivializer 
can fit into a differential paracocycle. This in turn is directly related to a genuine
differential cocycle. 
The natural question arises about whether the local basic data $\{g_{\mathrm{b}i},J_{\mathrm{b}i}\}$
of a $1$--gauge transformation can fit into some object with somewhat analogous properties 
capable of relating in a meaningful way to an assigned differential paracocycle. 


\begin{defi}\label{defi:nadifcoh3}
A gauge paraequivalence subordinated to a differential paracocycle 
$\{\omega_{\mathrm{b}i},\varOmega_{\mathrm{b}i},T_{\mathrm{b}ij}\}$ 
consists of a $1$--gauge transformation $\{g_{\mathrm{b}i},J_{\mathrm{b}i}\}$
enjoying the following properties.
\begin{enumerate}

\item For any set $U_i$, there exist Lie group and algebra valued internal functions 
$\bar g_i\in\iMap(T[1]U_i,\mathsans{G})$, $\bar J_i\in\iMap(T[1]U_i,\mathfrak{e}[1])$ 
such that 
\begin{align}
&g_{\mathrm{b}i}=\pi^*\bar g_i,
\vphantom{\Big]}
\label{nadifcoh27}
\\
&J_{\mathrm{b}i}=\pi^*\bar J_i.
\vphantom{\Big]}
\label{nadifcoh28}
\end{align}

\item For any two intersecting sets $U_i$, $U_j$, there exists a Lie group valued internal function
$\bar A_{ij}\in\iMap(T[1](U_i\cap U_j),\mathsans{E})$ such that 
\begin{equation}
\mu(g_{\mathrm{b}i},T_{\mathrm{b}ij}{}^{-1})T_{\mathrm{b}ij}=\pi^*\bar A_{ij}.
\label{nadifcoh29}
\end{equation}
\end{enumerate}
\end{defi}

\noindent
The following proposition shows the naturality of the above definition. 

\begin{prop}
Let $\{\omega_{\mathrm{b}i},\varOmega_{\mathrm{b}i},T_{\mathrm{b}ij}\}$ be a differential 
paracocycle and let $\{g_{\mathrm{b}i},J_{\mathrm{b}i}\}$ be gauge paraequivalence subordinated 
to it. Then, $\{{}^{g,J}\omega_{\mathrm{b}i}r,{}^{g,J}\varOmega_{\mathrm{b}i},T_{\mathrm{b}ij}\}$ 
is a differential paracocycle as well. In terms of the cocycle and equivalence data of
$\{\omega_{\mathrm{b}i},\varOmega_{\mathrm{b}i},T_{\mathrm{b}ij}\}$ and $\{g_{\mathrm{b}i},J_{\mathrm{b}i}\}$
the cocycle data of $\{{}^{g,J}\omega_{\mathrm{b}i}r,{}^{g,J}\varOmega_{\mathrm{b}i},T_{\mathrm{b}ij}\}$ read as 
\begin{align}
&{}^{\bar g,\bar J}\bar\omega_i=\Ad\bar g_i(\bar\omega_i)-d_{U_i}\bar g_i\bar g_i{}^{-1}-\dot\tau(\bar J_i),
\vphantom{\Big]}
\label{nadifcoh33}
\\
&{}^{\bar g,\bar J}\bar\varOmega_i=\mu{}\dot{}(\bar g_i,\bar\varOmega_i)
-d_{U_i}\bar J_i-\frac{1}{2}[\bar J_i,\bar J_i]-{}\dot{}\mu{}\dot{}({}^{\bar g,\bar J}\bar\omega_i,\bar J_i),
\vphantom{\Big]}
\label{nadifcoh34}
\end{align} 
\begin{align}
&{}^{\bar g,\bar J}\bar f_{ij}=\bar f_{ij},
\vphantom{\Big]}
\label{nadifcoh35}
\\
&{}^{\bar g,\bar J}\bar F_{ij}=\Ad\bar A_{ij}{}^{-1}(\bar J_i+\mu{}\dot{}(\bar g_i,\bar F_{ij}))
-\mu{}\dot{}(\bar f_{ij},\bar J_j)
\vphantom{\Big]}
\label{nadifcoh36}
\\
&\hspace{4.5cm}{}-\dot{}\mu({}^{\bar g,\bar J}\bar\omega_i,\bar A_{ij}{}^{-1})
-d_{U_i\cap U_j}\bar A_{ij}{}^{-1}\bar A_{ij},
\vphantom{\Big]}
\nonumber
\\
&{}^{\bar g,\bar J}\bar T_{ijk}=\bar T_{ijk}.
\vphantom{\Big]}
\label{nadifcoh37}
\end{align}
\end{prop}

\begin{proof}
Inserting \ceqref{nadifcoh8}, \ceqref{nadifcoh9} and \ceqref{nadifcoh27}, \ceqref{nadifcoh28}
into \ceqref{nadifcoh2/1}, \ceqref{nadifcoh3/1}, one readily finds that 
${}^{g,J}\omega_{\mathrm{b}i}=\pi^*{}^{\bar g,\bar J}\bar\omega_i$, 
${}^{g,J}\varOmega_{\mathrm{b}i}=\pi^*{}^{\bar g,\bar J}\bar\varOmega_i$ 
with ${}^{\bar g,\bar J}\bar\omega_i$, ${}^{\bar g,\bar J}\bar\varOmega_i$ 
given by \ceqref{nadifcoh33}, \ceqref{nadifcoh34}, respectively. \ceqref{nadifcoh35} is evident 
by relation \ceqref{nadifcoh10} expressing $f_{\mathrm{b}ij}$. To verify \ceqref{nadifcoh36}, 
one has to show that $F_{\mathrm{b}ij}$ can be expressed as in \ceqref{nadifcoh11} with $\bar\omega_i$, 
$\bar F_{ij}$ replaced by ${}^{\bar g,\bar J}\bar\omega_i$, ${}^{\bar g,\bar J}\bar F_{ij}$
as given by \ceqref{nadifcoh33}, \ceqref{nadifcoh36}, respectively. This is straightforward using 
\ceqref{nadifcoh5} together with \ceqref{nadifcoh10}, \ceqref{nadifcoh11} and \ceqref{nadifcoh29}
and the identities of app. B of I. \ceqref{nadifcoh37} is evident 
from relation \ceqref{nadifcoh12}. 
\end{proof}

\noindent
The following proposition describes the global matching of the local data of 
a gauge paraequivalence. 

\begin{prop}
Let $\{g_{\mathrm{b}i},J_{\mathrm{b}i}\}$ be a gauge paraequivalence subordinated 
to the differential paracocycle $\{\omega_{\mathrm{b}i},\varOmega_{\mathrm{b}i},T_{\mathrm{b}ij}\}$. 
Then, 
\begin{align}
&\bar g_i=\tau(\bar A_{ij})\bar f_{ij}\bar g_j\bar f_{ij}{}^{-1},
\vphantom{\Big]}
\label{nadifcoh38}
\\
&\bar J_i=\Ad\bar A_{ij}(\mu{}\dot{}(\bar f_{ij},\bar J_j)+{}^{\bar g,\bar J}\bar F_{ij})
-{}\dot{}\mu({}^{\bar g,\bar J}\bar\omega_i,\bar A_{ij})
\vphantom{\Big]}
\label{nadifcoh39}
\\
&\hspace{5cm}-d_{U_i\cap U_j}\bar A_{ij}\bar A_{ij}{}^{-1}
-\mu{}\dot{}(\bar g_i,\bar F_{ij})
\vphantom{\Big]}
\nonumber
\end{align}
on every non empty intersection $U_i\cap U_j$. Moreover, 
\begin{equation}
\bar A_{ik}=\mu(\bar g_i,\bar T_{ijk})\bar A_{ij}\mu(\bar f_{ij},\bar A_{jk})\bar T_{ijk}{}^{-1}
\label{nadifcoh40}
\end{equation}
on every non empty intersection $U_i\cap U_j\cap U_k$. 
\end{prop}

\begin{proof}
Inserting \ceqref{nadifcoh10} and \ceqref{nadifcoh27}
into \ceqref{nadifcoh4} and rearranging the resulting factors 
in the right hand side using also \ceqref{nadifcoh29}, relation
\ceqref{nadifcoh38} is obtained. To show \ceqref{nadifcoh39}, one substitutes 
\ceqref{nadifcoh10}, \ceqref{nadifcoh11} and \ceqref{nadifcoh27}, 
\ceqref{nadifcoh28} into \ceqref{nadifcoh5}. In the first 
insertion of \ceqref{nadifcoh11}, one expresses $F_{\mathrm{b}ij}$ 
in terms of ${}^{\bar g,\bar J}\bar\omega_i$, ${}^{\bar g,\bar J}\bar F_{ij}$;
in the second, one writes $F_{\mathrm{b}ij}$ through $\bar\omega_i$, $\bar F_{ij}$.
Use of the identities of app. B of I leads to \ceqref{nadifcoh39} straightforwardly. 
\ceqref{nadifcoh40} follows from combining \ceqref{nadifcoh4}, 
\ceqref{nadifcoh10}, \ceqref{nadifcoh12}, \ceqref{nadifcoh27}, \ceqref{nadifcoh29}
trough a simple algebraic computation. 
\end{proof}

\noindent
We note that eq. \ceqref{nadifcoh39} is an equivalent rewriting of eq. \ceqref{nadifcoh36}. 
However, we deduced  \ceqref{nadifcoh39} from \ceqref{nadifcoh5} by suitably expressing 
the latter relation in terms of barred objects. So, eq. \ceqref{nadifcoh39} does not constitute  
anything new, but it merely shows the consistency of eqs. \ceqref{nadifcoh5} and \ceqref{nadifcoh36}. 


\begin{defi}\label{defi:nadifcoh2}
Two pairs of differential paracocycles and subordinated gauge paraequivalences 
$\{\omega_{\mathrm{b}i},\varOmega_{\mathrm{b}i},T_{\mathrm{b}ij}\}$, 
$\{g_{\mathrm{b}i},J_{\mathrm{b}i}\}$,
$\{\tilde\omega_{\mathrm{b}i},\tilde\varOmega_{\mathrm{b}i},\tilde T_{\mathrm{b}ij}\}$,
$\{\tilde g_{\mathrm{b}i},\tilde J_{\mathrm{b}i}\}$
are equivalent if $\{\omega_{\mathrm{b}i},\varOmega_{\mathrm{b}i},T_{\mathrm{b}ij}\}$,
$\{\tilde\omega_{\mathrm{b}i},\tilde\varOmega_{\mathrm{b}i},\tilde T_{\mathrm{b}ij}\}$ are 
equivalent differential paracocycles 
and furthermore for every set $U_i$ 
\begin{align}
&\tilde g_{\mathrm{b}i}=g_{\mathrm{b}i},
\vphantom{\Big]}
\label{nadifcoh30}
\\
&\tilde J_{\mathrm{b}i}=J_{\mathrm{b}i}. 
\vphantom{\Big]}
\label{nadifcoh31}
\end{align}
\end{defi}

\noindent
Equivalence of differential paracocycle and subordinated gauge paraequivalence pairs 
 is manifestly an equivalence relation as suggested by its name. 

\begin{prop}
If $\{\omega_{\mathrm{b}i},\varOmega_{\mathrm{b}i},T_{\mathrm{b}ij}\}$,
$\{g_{\mathrm{b}i},J_{\mathrm{b}i}\}$,
$\{\tilde\omega_{\mathrm{b}i},\tilde\varOmega_{\mathrm{b}i},\tilde T_{\mathrm{b}ij}\}$,
$\{\tilde g_{\mathrm{b}i},\tilde J_{\mathrm{b}i}\}$
are equivalent pairs of differential paracocycles and subordinated gauge paraequivalences, then identities 
\ceqref{nadifcoh22}--\ceqref{nadifcoh26} hold and moreover
\begin{align}
&\bar{\tilde g}_i=\bar g_i,
\vphantom{\Big]}
\label{nadifcoh32}
\\
&\bar{\tilde J}_i=\bar J_i
\vphantom{\Big]}
\label{nadifcoh41}
\end{align}
on each set $U_i$ and \hphantom{xxxxxxxxxx}
\begin{equation}
\bar{\tilde A}_{ij}=\mu(\bar g_i,\bar T_{ij})\bar A_{ij}\bar T_{ij}{}^{-1}
\label{nadifcoh42}
\end{equation}
on every non empty intersection $U_i\cap U_j$.
\end{prop}

\begin{proof}
Since $\{\omega_{\mathrm{b}i},\varOmega_{\mathrm{b}i},T_{\mathrm{b}ij}\}$,
$\{\tilde\omega_{\mathrm{b}i},\tilde\varOmega_{\mathrm{b}i},\tilde T_{\mathrm{b}ij}\}$
are equivalent differential cocycles according to def. \cref{defi:nadifcoh2}, 
eqs. \ceqref{nadifcoh22}--\ceqref{nadifcoh26} hold by virtue of prop. \cref{prop:nadifcoh1}. 
Relations \ceqref{nadifcoh32}, \ceqref{nadifcoh41} are an immediate consequence of 
\ceqref{nadifcoh30}, \ceqref{nadifcoh31} and \ceqref{nadifcoh27}, \ceqref{nadifcoh28} and 
their tilded analogues. \ceqref{nadifcoh42} follows form \ceqref{nadifcoh29} and its tilded form
and \ceqref{nadifcoh21}. 
\end{proof}

Gauge paraequivalences subordinated to the same differential paracocycle  form a group. 

\begin{prop}
The gauge paraequivalences $\{g_{\mathrm{b}i},J_{\mathrm{b}i}\}$ subordinated to a fixed differential 
paracocycle $\{\omega_{\mathrm{b}i},\varOmega_{\mathrm{b}i},T_{\mathrm{b}ij}\}$ constitute a subgroup 
of the $1$--gauge group.
\end{prop}

\begin{proof}
Suppose that $\{g_{1\mathrm{b}i},J_{1\mathrm{b}i}\}$, $\{g_{2\mathrm{b}i},J_{2\mathrm{b}i}\}$ are 
gauge paraequivalences subordinated to $\{\omega_{\mathrm{b}i},\varOmega_{\mathrm{b}i},T_{\mathrm{b}ij}\}$ 
and that $\{g_{3\mathrm{b}i},J_{3\mathrm{b}i}\}$ is their product as $1$--gauge transformations, 
so that $g_{3\mathrm{b}i}=g_{2\mathrm{b}i}g_{1\mathrm{b}i}$, 
$J_{3\mathrm{b}i}=J_{2\mathrm{b}i}+\mu{}\dot{}(g_{2\mathrm{b}i},J_{1\mathrm{b}i})$.
Then, $g_{3\mathrm{b}i}$, $J_{3\mathrm{b}i}$ satisfy \ceqref{nadifcoh27}--\ceqref{nadifcoh29} too
with $\bar g_{3i}=\bar g_{2i}\bar g_{1i}$, $\bar J_{3i}=\bar J_{2i}+\mu{}\dot{}(\bar g_{2i},\bar J_{1i})$
and $\bar A_{3ij}=\mu(\bar g_{2i},\bar A_{1ij})\bar A_{2ij}$. Similarly, suppose that 
$\{g_{1\mathrm{b}i},J_{1\mathrm{b}i}\}$ is a 
gauge paraequivalence subordinated to $\{\omega_{\mathrm{b}i},\varOmega_{\mathrm{b}i},T_{\mathrm{b}ij}\}$ 
and that $\{g_{2\mathrm{b}i},J_{2\mathrm{b}i}\}$ is its inverse as a $1$--gauge transformation,   
so that $g_{2\mathrm{b}i}=g_{1\mathrm{b}i}{}^{-1}$, 
$J_{2\mathrm{b}i}=-\mu{}\dot{}(g_{1\mathrm{b}i},J_{1\mathrm{b}i})$.
Then, $g_{2\mathrm{b}i}$, $J_{2\mathrm{b}i}$ satisfies \ceqref{nadifcoh27}--\ceqref{nadifcoh29} too 
with $\bar g_{2i}=\bar g_{1i}{}^{-1}$, $\bar J_{2i}=-\mu{}\dot{}(\bar g_{1i}{}^{-1},\bar J_{1i})$
and $\bar A_{2ij}=\mu(\bar g_{1i}{}^{-1},\bar A_{1ij}{}^{-1})$. This is enough to show the proposition.
\end{proof}

We assume now that for each set $U_i$ of the covering 
the adapted coordinates $\gamma_i$, $\varGamma_i$ can be chosen to be special
(cf. def. \cref{defi:specadpcoo}). Then, by \ceqref{bmaurer15}  
\begin{equation}
I_i{}^*\varGamma_i=0,
\label{nadifcoh43}
\end{equation}
where $I_i:\pi_0{}^{-1}(U_i)\rightarrow\pi^{-1}(U_i)$ is the injection map.

\begin{prop}
The basic matching data $F_{\mathrm{b}ij}$ satisfy 
\begin{equation}
I_{ij}{}^*F_{\mathrm{b}ij}=0
\label{nadifcoh44}
\end{equation}
for each non empty intersection $U_i\cap U_j$
\end{prop}

\noindent
Above, $I_{ij}:\pi_0{}^{-1}(U_i\cap U_j)\rightarrow\pi^{-1}(U_i\cap U_j)$ is the injection map.

\begin{proof}
Eq. \ceqref{nadifcoh44} follows immediately from \ceqref{2basic7} and \ceqref{nadifcoh43}.
\end{proof}


\begin{prop}
If $\{\omega_{\mathrm{b}i}, \varOmega_{\mathrm{b}i},T_{\mathrm{b}ij}\}$ 
is a differential paracocycle, then for each non empty intersection $U_i\cap U_j$
\begin{align}
&\Ad I_{ij}{}^*T_{\mathrm{b}ij}(\pi_0{}^*\bar F_{ij})-{}\dot{}\mu(\pi_0{}^*\bar\omega_i,I_{ij}{}^*T_{\mathrm{b}ij})
\vphantom{\Big]}
\label{nadifcoh45}
\\
&\hspace{4cm}-d_{\pi_0{}^{-1}(U_i\cap U_j)}I_{ij}{}^*T_{\mathrm{b}ij}I_{ij}{}^*T_{\mathrm{b}ij}{}^{-1}=0.
\vphantom{\Big]}
\nonumber
\end{align}
\end{prop}

\begin{proof}
Eq. \ceqref{nadifcoh45} is a direct consequence of  \ceqref{nadifcoh11} and \ceqref{nadifcoh44}. 
\end{proof}

\begin{defi} \label{defi:specdifparcl}
A differential paracocycle $\{\omega_{\mathrm{b}i}, \varOmega_{\mathrm{b}i},T_{\mathrm{b}ij}\}$ 
is said to be special if the underlying $2$--connection is special (cf. def. \cref{defi:prop2conn}).
\end{defi}

\noindent
By \ceqref{propbas2}, then, in each set $U_i$ 
\begin{equation}
I_i{}^*\varOmega_{\mathrm{b}i}=0. 
\label{nadifcoh46}
\end{equation}
We note that by \ceqref{nadifcoh3} the condition of specialty is globally consistent if 
\ceqref{nadifcoh44} holds. 

\begin{prop}
If the differential paracocycle $\{\omega_{\mathrm{b}i}, \varOmega_{\mathrm{b}i},T_{\mathrm{b}ij}\}$ 
is special, then 
\begin{equation}
\bar\varOmega_i=0
\label{nadifcoh47}
\end{equation}
in each set $U_i$.
\end{prop}

\begin{proof}
By virtue of \ceqref{nadifcoh9} and the relation $\pi\circ I_i=\pi_0|_{\pi_0{}^{-1}(U_i)}$,  
\ceqref{nadifcoh46} implies that  $0=I_i{}^*\pi^*\bar\varOmega_i=\pi_0{}^*\bar\varOmega_i$. 
Since $\pi_0$ is a surjective submersion (cf. prop. 3.2 of I), \ceqref{nadifcoh47} holds. 
\end{proof}

\begin{prop}
If the differential paracocycle $\{\omega_{\mathrm{b}i}, \varOmega_{\mathrm{b}i},T_{\mathrm{b}ij}\}$ 
is special, so is any other paracocycle 
$\{\tilde\omega_{\mathrm{b}i},\tilde\varOmega_{\mathrm{b}i},\tilde T_{\mathrm{b}ij}\}$ equivalent to it. 
\end{prop}

\begin{proof}
By \ceqref{nadifcoh19}, \ceqref{nadifcoh20} and \ceqref{basic1}, \ceqref{basic2} and their 
tilded counterparts, the $2$--connections underlying two equivalent paracocycles are equal.
So, if the first paracocycle is special, so is the second by virtue of def. \cref{defi:specdifparcl}. 
\end{proof}


\begin{defi} \label{defi:specdifpareq}
A gauge paraequivalence $\{g_{\mathrm{b}i},J_{\mathrm{b}i}\}$ subordinated to a
differential paracocycle $\{\omega_{\mathrm{b}i}, \varOmega_{\mathrm{b}i},T_{\mathrm{b}ij}\}$ 
is said to be special if  the underlying $1$--gauge transformation is special.
\end{defi}

\noindent
By \ceqref{propbas3}, then, in each set $U_i$
\begin{equation}
I_i{}^*J_{\mathrm{b}i}=0.
\label{nadifcoh48}
\end{equation}
We note that by \ceqref{nadifcoh5} the condition of specialty is globally consistent if 
\ceqref{nadifcoh44} holds. 

\begin{prop}
It the gauge paraequivalence $\{g_{\mathrm{b}i},J_{\mathrm{b}i}\}$ subordinated to a
differential paracocycle $\{\omega_{\mathrm{b}i}, \varOmega_{\mathrm{b}i},T_{\mathrm{b}ij}\}$ 
is special, then in each set $U_i$. 
\begin{equation}
\bar J_i=0.
\label{nadifcoh49}
\end{equation}
\end{prop}

\begin{proof}
This follows from \ceqref{nadifcoh28} through a reasoning similar to that leading to 
\ceqref{nadifcoh47}. 
\end{proof}

\begin{prop}
If $\{\omega_{\mathrm{b}i},\varOmega_{\mathrm{b}i},T_{\mathrm{b}ij}\}$, 
$\{g_{\mathrm{b}i},J_{\mathrm{b}i}\}$ is a pair of a differential paracocycle and a 
subordinated gauge paraequivalence with $\{g_{\mathrm{b}i},J_{\mathrm{b}i}\}$ special and 
$\{\tilde\omega_{\mathrm{b}i},\tilde\varOmega_{\mathrm{b}i},\tilde T_{\mathrm{b}ij}\}$,
$\{\tilde g_{\mathrm{b}i},\tilde J_{\mathrm{b}i}\}$ is a pair of a differential paracocycle and a 
subordinated gauge paraequi\-valence equivalent to the former, then $\{\tilde g_{\mathrm{b}i},\tilde J_{\mathrm{b}i}\}$ 
is special.
\end{prop}

\begin{proof}
By \ceqref{nadifcoh30}, \ceqref{nadifcoh31} and \ceqref{basic9}, \ceqref{basic10} and their 
tilded counterparts, the $1$--gauge transformations underlying two equivalent  
differential paracocycle and subordinated gauge paraequivalence pairs are equal.
So, if the first paraequivalence is special, so is the second by virtue 
of def. \cref{defi:specdifpareq}. 
\end{proof}

The reader certainly noticed that we did not include $2$--gauge symmetry in our discussion.
The reason for this is that, apparently, there is no way of making it fitting into the 
framework described in this subsection. An analysis of the global matching of the local basic data
$E_{\mathrm{b}i}$, $C_{\mathrm{b}i}$ of a $2$--gauge transformation would unavoidably be based on 
relations \ceqref{2basic24}, \ceqref{2basic25}. Forcing on $E_{\mathrm{b}i}$, $C_{\mathrm{b}i}$ relations 
analogous to \ceqref{nadifcoh8}, \ceqref{nadifcoh9} and \ceqref{nadifcoh27}, \ceqref{nadifcoh28} 
does not seem to yield any reasonable relation on $M$. This is an open problem requiring further 
investigation including possibly a revision of the synthetic theory of subsect. \cref{subsec:cmgauforgau}.

\vfil\eject

\section{\textcolor{blue}{\sffamily Appraisal of the results obtained}}\label{sec:outlk}

It is important to critically assess the strengths and weaknesses of the operational 
synthetic formulation of the total space theory of 
principal $2$--bundles and $2$--connections and $1$-- and $2$--gauge transformations
thereof developed in this paper. A number of points can be raised 
concerning its viability and its eventual relationship with other approaches.
We are going address some of these issues in this section.


\subsection{\textcolor{blue}{\sffamily Some open problems}}\label{subsec:prob}

The geometry of a principal $\hat{\matheul{K}}$--$2$--bundle $\hat{\mathcal{P}}$ is characterized not only by
the right $\hat{\mathsans{K}}$--action but also by the morphism composition of $\hat P$.
Our operational formulation relies heavily of the former while it leaves the latter 
in the background (cf. subsect. 3.8 of I). However, the second is a constitutive element 
of the principal $2$--bundle structure as basic as the first. 

Since morphisms belonging to different fibers of a principal $2$--bundle 
can never be composed, morphism 
composition is essentially a local operation. For a chosen local neighborhood $U$ of $M$,
through a pair of reciprocally weakly inverse trivializing functors
$\hat{\varPhi}_U:\hat\pi^{-1}(U)\rightarrow U\times\hat{\mathsans{K}}$ 
and $\tilde{\hat{\varPhi}}_U:U\times\hat{\mathsans{K}}\rightarrow\hat\pi^{-1}(U)$
composition of morphisms of $\hat\pi^{-1}(U)$ is turned into composition of 
corresponding morphisms of $\hat{\mathsans{K}}$ and viceversa. 
It is known \ccite{Baez5} that the groupoid structure of a strict $2$--group such as $\hat{\mathsans{K}}$
can be reduced to the group one as follows. With any morphism $A\in\hat{\mathsans{K}}$ there is 
associated a morphism $\hat\varrho(A)\in\hat{\mathsans{K}}$ given by 
\begin{equation}
\hat\varrho(A)=\hat t(A)^{-1}A,
\label{}
\end{equation}
such that for composable morphisms $A,\,B\in\hat{\mathsans{K}}$
\begin{equation}
\hat\varrho(B\circ A)=\hat\varrho(B)\hat\varrho(A).
\label{}
\end{equation}
So, as $\hat t(B\circ A)=\hat t(B)$, right composition of $B$ by $A$ is equivalent to right 
mul\-tiplication of $B$ by $\hat\varrho(A)$. It follows that, for any two composable morphisms $X$, 
$Y\in \hat P$, right composition of $Y$ by $X$ can be reduced, in the appropriate categorical sense, 
to the right action of some element of $A_X\in\hat{\mathsans{K}}$ depending on $X$ 
on $Y$. 

Shifting to the synthetic setup of $\hat{\mathcal{P}}$, 
it is in the way explained above that the operation $\iOOO S_{P}$ indirectly includes
the morphism composition structure of $\hat P$ in spite of the fact that 
its synthetic counterpart $P$ has no groupoid structure (cf. subsect. 3.2 of I)
A more explicit incorporation of this latter in our formulation would be desirable. 


To construct the basic theory, we proposed a notion of coordinates adapted to the 
local product structure $U\times\mathsans{K}$ of a principal 
$\hat{\matheul{K}}$--$2$--bundle $\hat{\mathcal{P}}$ in subsect. \cref{subsec:maurer}. 
These coordinates are Lie valued internal functions on $T[1]\pi^{-1}(U)$ behaving in 
a certain way under the action of the derivations of the operation $\iOOO S_{\pi^{-1}(U)}$. 
The definition provided is essentially algebraic. It leads to a seemingly viable basic 
formulation of principal $2$--bundle $2$--connection and $1$--gauge transformation 
theory. However, the eventual relation of 
adapted coordinates to trivialization functors remains blurred at best and calls for 
further investigation. 

To make contact with other widely studied formulations of $2$--connections and $1$--gauge transformations 
of principal $2$--bundles, we introduced the notions
of differential paracocycle and gauge paraequivalence in subsect. \cref{subsec:nadifcoh}.
The definitions of these entities we gave are admittedly somewhat {\it ad hoc}.
The cocycle data $\{\bar\omega_i,\bar\varOmega_i,\bar f_{ij},\bar F_{ij},\bar T_{ijk}\}$
associated with a differential paracocycle $\{\omega_{\mathrm{b}i},\varOmega_{\mathrm{b}i},T_{\mathrm{b}ij}\}$ 
are simply assumed to exist as part of the definition of this latter. 
Similarly, the equivalence data $\{\bar g_i,\bar J_i,\bar A_{ij}\}$ associated with 
a gauge paraequivalence $\{g_{\mathrm{b}i},J_{\mathrm{b}i}\}$ are again assumed to exist. 
It would be desirable instead to have a formulation where the cocycle and equivalence data 
can be constructively shown to exist in analogy to the ordinary theory. 

The viability of the formulation furnished 
here remains to be tested in concrete examples. This left for future work.


\subsection{\textcolor{blue}{\sffamily Toward a more geometric interpretation}}\label{subsec:interpr}

In this paper, we worked out an operational synthetic total space theory of $2$--connections and $1$-- and $2$--gauge 
transformations for strict principal $2$--bundles adopting a graded differential 
geometric approach and mimicking to a large extent the corresponding formulation of connection and 
gauge transformation theory for ordinary principal bundles. In the ordinary case, however, these notions 
have also a more conventional intuitive geometric interpretation in terms of the overall 
geometry of the principal bundle and its fibered structure. The natural question arises whether 
a similar interpretation exists also in the higher theory.  

As already observed in subsect. 3.1 of I, at the moment 
no definition of $2$--connection on a strict principal $2$--bundle akin to that of the ordinary 
theory formulated in terms of a horizontal invariant distribution in the 
tangent bundle of the bundle is available. There exits however a definition of $1$--gauge transformation analogous 
to that of the ordinary theory as an equivariant fiber preserving bundle automorphism
formulated by Wockel in ref. \ccite{Wockel:2008tspb}. The interpretation of $2$--connections
as defined in subsect. \cref{subsec:cmconn} along the lines just indicated remains an open problem.
It is conversely possible to attempt a comparison of the notion of $1$--gauge transformation of 
subsect. \cref{subsec:cmgautr} and Wockel's categorical one. 

For a given strict principal $\hat{\matheul{K}}$--$2$--bundle $\hat{\mathcal{P}}$,
the synthetic counterpart of the gauge $2$--group $\Fun^{\hat{\mathsans{K}}}\!(\hat{P},\hat{\mathsans{K}}_{\mathrm{Ad}})$ 
(cf. subsect. 3.1 of I)
is the group $\Fun^{\mathsans{K}}\!(P,\mathsans{K}_{\mathrm{Ad}})$ of $\mathsans{K}$--equivariant maps 
of $\Map(P,\mathsans{K})$ restricting to $\mathsans{K}_0$--equivariant maps of $\Map(P_0,\mathsans{K}_0)$. 
$\Fun^{\mathsans{K}}\!(P,\mathsans{K}_{\mathrm{Ad}})$ is formally analogous to 
$\Fun^{\hat{\mathsans{K}}}\!(\hat{P},\hat{\mathsans{K}}_{\mathrm{Ad}})$ in se\-veral respects,  but by the 
lack of a groupoid structure of $\mathsans{K}$ (cf. subsect. 3.2 of I)
it has no morphisms and is thus a mere mapping group. 
$\Fun^{\mathsans{K}}\!(P,\mathsans{K}_{\mathrm{Ad}})$ can\-not be directly equated with the $1$--gauge group 
as defined earlier in subsect. \cref{subsec:cmgautr}. 
Rather, $\Fun^{\mathsans{K}}\!(P,\mathsans{K}_{\mathrm{Ad}})$ can be identified as a distinguished 
subgroup of the special subgroup of the $1$--gauge transformation group, as we show next. 

Recalling that $\mathsans{K}=\DD\mathsans{M}$, 
an element of $\Fun^{\mathsans{K}}\!(P,\mathsans{K}_{\mathrm{Ad}})$
is an instance of an internal function $\varPsi\in\iMap(T[1]P,\DD\mathsans{M})$ 
that is $\DD\mathsans{M}$--horizontal and $\DD\mathsans{M}$--equivariant
and restricts to an internal function $\varPsi_0\in\iMap(T[1]P_0,\DD\mathsans{M}_0)$ 
that is $\DD\mathsans{M}_0$--ho\-rizontal and $\DD\mathsans{M}_0$--equivariant, the action  
of $\DD\mathsans{M}$, respectively $\DD\mathsans{M}_0$, 
on itself being the right conjugation one. $\DD\mathsans{M}$--horizontality 
translates directly into relation \ceqref{cmgautrx3}. 
$\DD\mathsans{M}$--equivariance is equivalent to the condition that \hphantom{xxxxxxx}
\begin{equation}
R_F{}^*\varPsi=F^{-1}\varPsi F \pagebreak 
\label{}
\end{equation}
for $F\in\DD\mathsans{M}$, where in the right hand side $F$ is identified with its 
image under the isomorphism $z_{\mathsans{M}}:\DD\mathsans{M}\rightarrow\DD\mathsans{M}^+$ 
defined in eq. 3.6.1 of I. 
In infinitesimal form, expressing $F$ as $1+tZ$, where 
$Z\in\DD\mathfrak{m}$ and $t$ is a formal parameter such that $t^2=0$, this relation 
takes the form \ceqref{cmgautrx5}. Thus $\varPsi$ is the transformation component of 
a $1$--gauge transformation. Since $\varPsi$ restricts on $T[1]P_0$ to a
$\DD\mathsans{M}_0$--valued $\DD\mathsans{M}_0$--horizontal and $\DD\mathsans{M}_0$--equivariant internal map, 
this $1$--gauge transformation is special. 

In the present formulation of the theory, 
$2$--gauge transformations of $\hat{\mathcal{P}}$ cannot be obviously 
related to morphisms of the gauge $2$--group 
$\Fun^{\hat{\mathsans{K}}}\!(\hat{P},\hat{\mathsans{K}}_{\mathrm{Ad}})$, because 
the synthetic form $\Fun^{\mathsans{K}}\!(P,\mathsans{K}_{\mathrm{Ad}})$ of this latter does
not have any. Moreover, $2$--gauge transformations are supposed 
to act on $1$--gauge transformations (cf. subsect. \cref{subsec:cmgauforgau}) and do so in
a proper way depending on an assigned $2$--connection (cf. subsect. \cref{subsec:cmconn}).
As long we do not have a purely geometric total space theory of $2$--connections, any attempt 
to relate $2$--gauge transformations to morphisms of the gauge $2$--group is premature at best. 



\subsection{\textcolor{blue}{\sffamily Comparison with other formulations}}\label{subsec:wald}

An interesting total space formulation of $2$--connections theory has been worked out by Waldorf in refs. 
\ccite{Waldorf:2016tsct,Waldorf:2017ptpb}. We anticipate that Waldorf's theory is not obviously
equivalent to ours and most likely it is not. 
We outline it briefly below referring the interested reader to the cited papers for a full exposition.

Waldorf's approach is based on a special differential geometric framework. For a given 
principal $\hat{\matheul{K}}$--$2$--bundle $\hat{\mathcal{P}}$, its main ingredients are
the morphism and object manifolds $\hat P$ and $\hat P_0$ of $\hat{\mathcal{P}}$
and the Lie group crossed module 
$\mathsans{M}=(\mathsans{E},\mathsans{G},\tau,\mu)$ associated with the structure Lie $2$--group $\hat{\matheul{K}}$. 
He defines a vector space $A^\bullet(\hat P,\hat{\mathfrak{k}})$ of $\hat{\mathfrak{k}}$--valued 
differential forms of $\hat P$ and endows it with a structure of differential graded Lie algebra 
with Lie bracket $[-,-]$ and differential $D$. $A^p(\hat P,\hat{\mathfrak{k}})$ is 
a certain subspace of the vector space 
$\Omega^p(\hat P_0,\mathfrak{g})\,\oplus\,\Omega^p(\hat P,\mathfrak{e})\,\oplus\,\Omega^{p+1}(\hat P_0,\mathfrak{e})$
defined by algebraic constraints expressed in terms of the face maps of the nerve of the groupoid $\hat P$,
the simplicial complex $\hat P_\bullet=\hspace{-4mm}\xymatrix@C-=0.75cm
{\cdots\hspace{-12mm}&\ar@<-3.75pt>[r]\ar@<+3.75pt>[r]\ar@<-1.25pt>[r]\ar@<1.25pt>[r]
&\hat P_2\ar@<-3pt>[r]\ar@<+3pt>[r]\ar[r]&\hat P_1\ar@<-2pt>[r]\ar@<+2pt>[r]&\hat P_0}$ of composable sequences of $\hat P$.
Each $p$--form possesses therefore three components. 
The Lie bracket combines the form wedge product and the Lie bracket of the Lie algebra 
$\mathfrak{e}\rtimes_{\dot{}\mu\dot{}}\mathfrak{g}$. The differential $D$ is constructed assembling 
the de Rahm differentials $d_{\hat P}$, $d_{\hat P_0}$, the face maps 
of $\hat P_\bullet$ and target map $\dot\tau$. An adjoint action $\Ad$ of functors 
$\hat P\rightarrow\hat{\mathsans{K}}$ on $A^\bullet(\hat P,\hat{\mathfrak{k}})$  
preserving degree is also defined.

A $2$--connection is defined again in terms of its behaviour under the 
right $\hat{\mathsans{K}}$ action of $\hat P$. A $\hat{\mathsans{K}}$--valued variable $Q$ is 
considered and a Maurer--Cartan $1$--form $\varGamma\in A^1(\hat{\mathsans{K}},\hat{\mathfrak{k}})$ 
obeying the Maurer--Cartan equation $D\varGamma+[\varGamma,\varGamma]/2=0$ is defined. 
A $2$--connection of $\hat P$ is  a $1$--form $A\in A^1(\hat P,\hat{\mathfrak{k}})$ 
such that 
\begin{equation}
R_Q{}^*A=\Ad Q^{-1}(A)+\varGamma,
\label{wald1}
\end{equation}
which must be viewed as a $1$--form of $A^1(\hat P\times\hat{\mathsans{K}},\hat{\mathfrak{k}})$. 
The curvature of $A$ is a $2$--form $B\in A^2(\hat P,\hat{\mathfrak{k}})$ defined by \hphantom{xxxxxxxxxxx}
\begin{equation}
B=DA+\frac{1}{2}[A,A].
\label{wald2}
\end{equation}
By \ceqref{wald1}, it obeys \hphantom{xxxxxxxxxxx} 
\begin{equation}
R_Q{}^*B=\Ad Q^{-1}(B). 
\label{wald3}
\end{equation}
More explicitly, denoting by $g$ and $(H,h)$ the $\mathsans{G}$ and $\mathsans{E}\rtimes_\mu\mathsans{G}$ 
variables underlying $Q$ above, a $2$--connection $A$ consists of a triplet of forms
$\omega\in\Omega^1(\hat P_0,\mathfrak{g})$, $\varOmega'\in\Omega^1(\hat P,\mathfrak{e})$, $\varOmega\in\Omega^2(\hat P_0,\mathfrak{e})$
satisfying the simplicial constraints and such that 
\begin{align}
&R_g{}^*\omega=\Ad g^{-1}(\omega)+g^{-1}dg,
\vphantom{\Big]}
\label{wald4}
\\
&R_{(H,h)}{}^*\varOmega'=\mu{}\dot{}\,(h^{-1},\Ad H(\varOmega'+{}\dot{}\mu(\hat s^*\omega,H))+H^{-1}dH),
\vphantom{\Big]}
\label{wald5}
\\
&R_g{}^*\varOmega=\mu{}\dot{}\,(g^{-1},\varOmega),
\vphantom{\Big]}
\label{wald6}
\end{align}
where $\hat s$, $\hat tt$ are the source and target maps of $\hat P$. The curvature of the $2$--connection is the triplet of forms 
$\theta\in\Omega^2(\hat P_0,\mathfrak{g})$, $\varTheta'\in\Omega^2(\hat P,\mathfrak{e})$, $\varTheta\in\Omega^3(\hat P_0,\mathfrak{e})$
given by 
\begin{align}
&\theta=d_{\hat P_0}\omega+\frac{1}{2}[\omega,\omega]-\dot\tau(\varOmega),
\vphantom{\Big]}
\label{wald7}
\\
&\varTheta'=(\hat t^*-\hat s^*)\varOmega+d_{\hat P}\varOmega'+\frac{1}{2}[\varOmega',\varOmega']
+\dot{}\mu\dot{}\,(\hat s^*\omega,\varOmega'), 
\vphantom{\Big]}
\label{wald8}
\\
&\varTheta=d_{\hat P_0}\varOmega+\dot{}\mu\dot{}\,(\omega,\varOmega)
\vphantom{\Big]}
\label{wald9}
\end{align}
satisfying certain simplicial constraints and such that 
\begin{align}
&R_g{}^*\theta=\Ad g^{-1}(\theta),
\vphantom{\Big]}
\label{wald10}
\\
&R_{(H,h)}{}^*\varTheta'=\mu{}\dot{}\,(h^{-1},\Ad H(\varTheta'+{}\dot{}\mu(\hat s^*\theta,H))), 
\vphantom{\Big]}
\label{wald11}
\\
&R_g{}^*\varTheta=\mu{}\dot{}\,(g^{-1},\varTheta).
\vphantom{\Big]}
\label{wald12}
\end{align}
The following differences between Waldorf's formulation, henceforth marked as W, 
and  the formulation presented in this paper emerge, marked as O, emerge
even leaving aside the non synthetic nature of W and the synthetic one of O. 

In W, a $2$--connection has three  components
$\omega$, $\varOmega'$, $\varOmega$ whereas, in O, it has only two components $\omega$, $\varOmega$. 
In W, $\omega$, $\varOmega$ are forms on $\hat P_0$, while, in O, $\omega$, $\varOmega$ 
are forms on $P$. It is not possible to forget $\varOmega'$ in W because it enters into the simplicial 
constraints together with $\omega$, nor it is possible to set $\varOmega'=0$ because this would be 
inconsistent with \ceqref{wald5}. Apparently, the components $\omega$, $\varOmega$ of W correspond to the pull-back
components $\omega_0=I^*\omega$, $\varOmega_0=I^*\varOmega$ of O (cf. subsect. \cref{subsec:cmconn}). 
Relations \ceqref{wald4}, \ceqref{wald6}
of W in infinitesimal form are compatible with relations \ceqref{cmconn9}, \ceqref{cmconn10} of O
under the operation morphism $\iOOO L$. Similar remarks apply when comparing the three curvature components
$\theta$, $\varTheta'$, $\varTheta$ of W and the two components $\theta$, $\varTheta$ of O. 
From these remarks, it appears that W is not obviously equivalent to O and most likely it is not. 
Yet, the two formulations may yield at the end equivalent descriptions of $2$--connections on the 
bundle's base manifold. This remains an issue deserving further investigation.



\vfil\eject

\noindent
\textcolor{blue}{Acknowledgements.} 
The author thanks R. Picken, J. Huerta and C. Saemann for useful discussions.
He acknowledges financial support from INFN Research Agency
under the provisions of the agreement between University of Bologna and INFN. 
He also thanks the organizer of the 2018 EPSRC Durham Symposium 
on ``Higher Structures in M-Theory'' during which part of this work was done.

\vfil\eject

\end{document}